\documentclass[preprint,3p,10pt]{elsarticle}
\usepackage{lineno,hyperref}
\usepackage{amsmath,bm,geometry}
\usepackage{amsfonts}
\usepackage{color}
\usepackage{nomencl,etoolbox,ragged2e,siunitx}
\usepackage{mathtools}
\usepackage{multirow}
\usepackage{caption}
\usepackage{parskip}
\usepackage{calligra}
\usepackage{makeidx}
\usepackage{tikz}
\usepackage{mathrsfs}
\usepackage{bbm}
\usepackage{algorithm}
\usepackage{algorithmic}
\usepackage{amsthm}
\usepackage{graphicx}
\usepackage{subfigure}
\usepackage{hyperref}
\usepackage{cleveref}
\biboptions{sort&compress}
\hypersetup{pdfauthor=author}
\usetikzlibrary{matrix}
\makeindex
\DeclareMathAlphabet{\mathcalligra}{T1}{calligra}{m}{n}
\newtheorem{theorem}{Theorem}

\begin{document}
\title{An implicit unified gas-kinetic wave-particle method for radiative transport process}
\author[ad1]{Chang Liu}
\ead{liuchang@iapcm.ac.cn}
\author[ad1]{Weiming Li\corref{cor1}}
\ead{li\_weiming@iapcm.ac.cn}
\author[ad5]{Yanli Wang}
\ead{wang\_yanli@csrc.ac.cn}
\author[ad1,ad2]{Peng Song}
\ead{song\_peng@iapcm.ac.cn}
\author[ad3,ad4]{Kun Xu}
\ead{makxu@ust.hk}
\address[ad1]{Institute of Applied Physics and Computational Mathematics, Beijing, China}
\address[ad2]{HEDPS, Center for Applied Physics and Technology, College of Engineering, Peking University, Beijing, China}
\address[ad3]{Department of Mathematics, Hong Kong University of Science and Technology, Hong Kong}
\address[ad4]{Department of Mechanical and Aerospace Engineering, Hong Kong University of Science and Technology, Hong Kong, China}
\address[ad5]{Beijing Computational Science Research Center, Beijing, China, 100193}
\cortext[cor1]{Corresponding author}

\begin{abstract}
The unified gas-kinetic wave-particle method (UGKWP) has been developed for the multiscale gas, plasma, and multiphase flow transport processes for the past years. 
In this work, we propose an implicit unified gas-kinetic wave-particle (IUGKWP) method to remove the Courant–Friedrichs–Lewy (CFL) time step constraint. 
Based on the local integral solution of the radiative transfer equation (RTE),
the particle transport processes are categorized into 
the long-$\lambda$ streaming process and the short-$\lambda$ streaming process 
comparing to a local physical characteristic time $t_p$.
In the construction of the IUGKWP method,
the long-$\lambda$ streaming process is tracked by the implicit Monte Carlo (IMC) method;
the short-$\lambda$ streaming process is evolved by solving the implicit moments equations;
and the photon distribution is closed by a local integral solution of RTE.
In the IUGKWP method, the multiscale flux of radiation energy and the multiscale closure of photon distribution
are constructed based on the local integral solution.
The IUGKWP method preserves the second-order asymptotic expansion of RTE in the optically thick regime
and adapts its computational complexity to the flow regime.
The numerical dissipation is well controlled, and the teleportation error is significantly reduced in the optically thick regime.
The computational complexity of the IUGKWP method decreases exponentially as the Knudsen number approaches zero, 
and the computational efficiency is remarkably improved in the optically thick regime.
The IUGKWP is formulated on a generalized unstructured mesh, and multidimensional 2D and 3D algorithms are developed. 
Numerical tests are presented to validate the capability of IUGKWP in capturing the multiscale photon transport process.
The algorithm and code will apply in the engineering applications of inertial confinement fusion (ICF).
\end{abstract}
\maketitle

\section{Introduction}
The radiative transport process is one of the fundamental energy transfer processes in the high-energy-density physics,
such as in astrophysics and the inertial confinement fusion (ICF) \cite{lan2022dream,chen2022determination}.
The study of radiation physics dates back to the early nineteenth century.
In the early 1900s, a mathematical relationship was formulated by Planck 
to explain the spectral-energy distribution of radiation emitted by a black body, 
known as the Planck's radiation law.
Ever since then, radiation and photon transport have been studied extensively.
The radiative transport equation (RTE) is the mathematical equation system 
that describes the photon transport and interaction with the medium.
The RTE models the photon transport on the mesoscopic scale of the mean free path, 
and describe the evolution of specific intensity.
The multiscale physics of photon transport has been understood mathematically by the asymptotic theories,
such as the Hilbert expansion and the Chapman-Enskog expansion \cite{chapman1990mathematical,larsen1987asymptotic},
that bridge the mesoscopic kinetic equation and the macroscopic diffusion equation.
Characterized by the Knudsen number, 
i.e., the ratio between the photon mean free path to the characteristic length,
the flow regimes can be divided into the optically thin regime, 
the transitional regime, and the optically thick regime.
In the optically thick regime, the high-order RTE degenerate into a low-order diffusive equation,
and the diffusive coefficient is proportional to the reciprocal of the scattering coefficient.
Though the physical and mathematical theories of radiative transport have been well-established,
accurate and effective numerical methods and robust simulation programs are still highly demanded 
for engineering applications in high-energy-density physics.

Over the past decades, efforts have been made to construct 
concise moment models and effective numerical methods. 
The moment models, such as the spherical harmonics ($P_N$) model, 
expands the specific intensity in a specific functional space to reduce model order,
and at the same time preserves the essential physical properties such as the rotational invariance \cite{larsen1996asymptotic,cai2012efficient,fu2022asymptotic}.
The machine-learning models have also been developed in recent years, 
by implanting the neural network in the closure modeling \cite{huang2022machine,li2021learning}.
For the numerical methods, the kinetic equation solvers can be categorized into 
the deterministic discrete ordinate ($S_N$) method \cite{sun2015asymptotic,mieussens2000discrete}
and the stochastic Monte Carlo (MC) method \cite{fleck1971implicit,gentile2001implicit,gentile2016iterative}.
The $S_N$ methods use quadrature to discretize phase space, and the MC methods use stochastic particles.
The asymptotic persevering (AP) property is important in the construction of multiscale numerical schemes,
which states that the discretized numerical scheme preserves the asymptotics of RTE,
i.e., the collisionless Boltzmann equation in the vacuum regime 
and the diffusion equation in the optically thick regime,
without severe restrictions on the numerical resolution \cite{jin2010asymptotic}.
The asymptotic preserving schemes have been developed under the framework of $S_N$ \cite{jin1999ap} and MC \cite{shi2020asymptotic}, with a well-controlled numerical dissipation.
To achieve high resolution,
the high-order schemes have been developed under the discrete Galerkin (DG) framework \cite{xiong2022high}.
To overcome the high computational cost of solving the high dimensional RTE, 
acceleration techniques have been developed.
The diffusive synthetic acceleration (DSA) proposed by Larsen et al. \cite{morel1982synthetic} is a monumental achievement that couples the evolution of the high-order microscopic kinetic equation and the low-order macroscopic diffusion equation.
In recent years, the general synthetic iterative scheme \cite{su2020fast},
the high-order low-order (HOLO) coupling scheme \cite{chacon2017multiscale},
the variance-reduced methods \cite{sadr2023variance},
macro-micro decomposition methods \cite{gamba2019micro},
fast kinetic method \cite{dimarco2018efficient}
have been well developed to improve accuracy and boost efficiency.

The unified gas-kinetic scheme (UGKS) proposed by Xu et al. achieves high accuracy and efficiency in the simulation of multiscale transport processes of gas, radiation, plasma, and multiphase flow \cite{xu2010unified,liu2016unified,mieussens2013asymptotic,liu2019unified,liu2017unified,liu2020unifiedlinear}.
The UGKS provides a unified-preserving framework in the construction of kinetic schemes \cite{guo2023unified},
under which the discrete unified gas kinetic scheme (DUGKS) \cite{guo2013discrete,guo2021progress},
the unified gas-kinetic particle (UGKP) method and 
the unified gas-kinetic wave-particle (UGKWP) method is constructed
and applied in the multiscale transport processes \cite{liu2020unified,zhu2019unified,li2020unified,xu2021unified,yang2021unified,liu2021unified}.
A unified-preserving property is proposed to measure 
the capability of a numerical scheme to preserve the asymptotic limits 
on a space-time resolution larger than $O(\text{Kn}^{1/2})$ \cite{guo2023unified}.

In this work, an implicit UGKWP method is developed that removes 
the stiff-source constraint and CFL constraint on the time step.
Therefore, the time step can be chosen purely according to the local time resolution.
The local non-equilibrium flow physics, including the entropy,  
is preserved under various numerical resolutions.
The IUGKWP has the property of regime-adaptive, 
which states that the scheme adapts its degree of freedom (DOF) to the local flow regime.
In the optically thin regime, the DOF of IUGKWP is consistent with the implicit Monte Carlo (IMC) method. 
In the optically thick regime, the DOF of IUGKWP exponentially degenerates to a diffusion scheme.
We develop multidimensional codes for 2D and 3D radiative transport simulations.
The algorithms and codes are validated by a series of numerical tests.

The rest of this paper is organized as follows.
In Section \ref{section_rte}, we briefly introduce the kinetic model of photon transport, 
namely the radiative transfer equation.
We will also review the asymptotic theory and derive the local integral solution to RTE.
In Section \ref{section_ugkwp}, the implicit unified gas kinetic wave-particle is presented, 
including the three-step algorithm, namely 
(i) tracking long-$\lambda$ transport process by IMC; 
(ii) update short-$\lambda$ transport process by implicit moments equations;
(iii) photon distribution closure.
In Section \ref{section_analysis}, we will analyze the numerical property of the scheme, 
i.e., the asymptotic preserving (AP) property and regime adaptive property.
The numerical tests are presented in Section \ref{section_tests}, 
and Section \ref{section_conclusion} is the conclusion.

\section{Radiative transport model and asymptotic theory}\label{section_rte}
The physics of photon transport process is described by the radiative transfer equation 
on the mesoscopic mean free path scale.
The radiative transfer equation system consists of the Boltzmann equation of photon transport and 
the evolution equation of material temperature.
The Boltzmann equation describes the physical process of photon streaming, emission, absorption, and scattering.
In this paper, we consider the emission and absorption process, and the radiative transfer equation is written as
\begin{equation}\label{eq_RTE}
    \left\{
        \begin{aligned}
        &\frac{1}{c}\frac{\partial I}{\partial t}+\vec{\Omega}\cdot\nabla I=\sigma\left(B-I\right),\\
        &\frac{\partial C_v T}{\partial t}\equiv\frac{\partial u}{\partial t}=
        \int_{\mathcal{S}^2} \int_{\mathcal{R}}\sigma\left( I -
        B \right)\mathrm{d}\nu \mathrm{d}\vec{\Omega}.
        \end{aligned}
    \right.
\end{equation}
The spatial variable is denoted by $\vec{x}$,
the time variable is $t$,
the frequency variable is $\nu$,
and the angular variable is $\vec{\Omega}$.
The physical constants are the speed of light $c$,
the Boltzmann constant $k$,
and the Planck constant $h$.
$I(\vec{x},t,\vec{\Omega},\nu)$ is the spectral intensity,
$T(\vec{x},t)$ is the material temperature,
$C_v$ is the specific heat opacity,
$\sigma(\vec{x},\nu,t)$ is the opacity,
and $u(\vec{x},t)$ is the material energy density.
The emission radiance follows the Planck distribution
\begin{equation}\label{eq_planck}
    B(\nu,T)=\frac{2h\nu^3}{c^2}\frac{1}{\exp\left(\frac{h\nu}{kT}\right)-1},
\end{equation}
the zeroth order moment of which gives the radiant flux $\phi(\vec{x},t)$,
\begin{equation}\label{eq_radiantflux}
    \phi(\vec{x},t)=\int_{\mathcal{S}^2}\int_\mathcal{R} B(\nu,T) \mathrm{d} \nu \mathrm{d} \vec{\Omega}
        =\frac{8k^2\pi^5}{15h^3c^3}cT^4.
\end{equation}
The constant $a=\frac{8k^2\pi^5}{15h^3c^3}$ is referred to as the radiation constant.
The zeroth order moment of the spectral intensity $I$ is the radiation flux $\rho(\vec{x},t)$,
\begin{equation}\label{eq_radiantfluxrho}
    \rho(\vec{x},t)=\int_{\mathcal{S}^2}\int_\mathcal{R} I(\vec{x},t,\vec{\Omega},\nu) \mathrm{d} \nu \mathrm{d} \vec{\Omega}.
\end{equation}

The macroscopic energy equations can be derived by integrating Eq.\eqref{eq_RTE} in angular and frequency space,
\begin{equation}\label{eq_RTEmac}
    \left\{
      \begin{aligned}
        &\frac{\partial \rho}{\partial t}+\nabla\cdot \vec{F}=
        \frac{c}{\varepsilon^2}\left(<\sigma,B>-<\sigma,I>\right),\\
        &\frac{ \partial C_v T}{\partial t}=
        \frac{1}{\varepsilon^2}\left(<\sigma,I>-<\sigma,B>\right).
      \end{aligned}
    \right.
\end{equation}
Here $<\sigma,I>$ and $<\sigma,B>$ is the moment of spectral intensity, defined as
\begin{equation}\label{zeromoment}
\begin{aligned}
    &<\sigma,I>=\int_{\mathcal{S}^2} \int_{\mathcal{R}} \sigma(\vec{x},\nu,T) I(\vec{x},\vec{\Omega},\nu,t) \mathrm{d}\nu \mathrm{d}\vec{\omega},\\
    &<\sigma,B>=\int_{\mathcal{S}^2} \int_{\mathcal{R}} \sigma(\vec{x},\nu,T) B(\vec{x},\vec{\Omega},\nu,t) \mathrm{d}\nu \mathrm{d}\vec{\omega}.
\end{aligned}
\end{equation}
The first order moment of spectral intensity $\vec{F}(\vec{x},t)$ describes the spatial flux of $\rho$,
\begin{equation}\label{firstmoment}
    \vec{F}(\vec{x},t)=\int_{\mathcal{S}^2}\frac{c\vec{\Omega}}{\varepsilon} I(\vec{x},\vec{\Omega},t) \mathrm{d}\vec{\Omega}.
\end{equation}
In a certain flow regime, the energy equation \eqref{eq_RTEmac} can be closed 
with a specific closure modeling of radiant intensity.

In the optically thick regime, the closure of radiant intensity can be derived by asymptotic theory, 
and the radiant transfer equation degenerates to the diffusion equation \cite{chapman1990mathematical,larsen1987asymptotic}.
The physical quantities are re-scaled to order $O(1)$ as following 
\begin{equation}
    \begin{aligned}
        &\hat{t}=t/t_\infty, \quad \hat{x}=x/x_\infty, \quad 
        \hat{c}=c/v_\infty, \quad \hat{\sigma}=\sigma/\sigma_\infty, \\
        &\hat{I}=I/I_\infty, \quad \hat{B}=B/I_\infty, \quad
        \hat{C_v}=C_v/C_{v\infty}, \quad \hat{T}=T/T_\infty.
    \end{aligned}
\end{equation}
Define $\varepsilon^2=\sqrt{v_\infty t_\infty \sigma_\infty}$ as the dimensionless Knudsen number, 
and the slow diffusive process is characterized as ${x_\infty}=\varepsilon{v_\infty t_\infty}$.
The re-scaled radiant transport equation reads
\begin{equation}\label{eq_rescaleRTE}
    \left\{
    \begin{aligned}
        &\frac{1}{\hat{c}}\frac{\partial \hat{I}}{\partial \hat{t}}+
        \frac{1}{\varepsilon}\Omega\cdot\nabla \hat{I}=\frac{1}{\varepsilon^2}\hat{\sigma}
        \left(\hat{B}-\hat{I}\right),\\
        &\frac{\partial \hat{C_v} \hat{T}}{\partial \hat{t}}=\frac{1}{\varepsilon^2}
        \int_{\mathcal{R}}\int_{\mathcal{S}^2}\hat{\sigma}\left(\hat{B}-\hat{I}\right)\mathrm{d}\vec{\omega}d\nu.
    \end{aligned}
    \right.
\end{equation}
The asymptotic preserving property is one of the main topics in this paper, 
and therefore, the re-scaled radiative transfer equation \eqref{eq_rescaleRTE} is discussed in the following.
The hats will be omitted in the following for simplicity.
We expand the radiant intensity and time derivative with respect to the Knudsen number,
\begin{equation}\label{eq_expansion}
    \begin{aligned}
        &\partial_t=\varepsilon^0\partial_{t0}+\varepsilon^1\partial_{t1}+O(\varepsilon^2),\\
        &I=\varepsilon^0I_0+\varepsilon^1I_1+O(\varepsilon^2).\\
    \end{aligned}
\end{equation}
Operator $\partial_{tk}$ is defined as the time evolution contributed by the $k$th order of flux and source terms.
We obtain the $k$th order equations by balancing the $O(\varepsilon^{k})$ order terms, 
and specifically, the $O(\varepsilon^{0})$ order gives
\begin{equation}\label{eq_expansion0}
    I_0(\vec{x},\vec{\Omega},\nu,t)=B(\vec{x},\vec{\Omega},\nu,t),
\end{equation}
showing that the leading order of radiant intensity is a local Plankian, 
which is also referred to as the local equilibrium state.
The $O(\varepsilon^{1})$ order of the radiative transfer equation \eqref{eq_RTE} gives
\begin{equation}\label{eq_expansion1}
  I(\vec{x},\vec{\Omega},\nu,t)=B(\vec{x},\vec{\Omega},\nu,t)-
  \varepsilon\frac{\vec{\Omega}}{\sigma}\cdot\nabla B(\vec{x},\vec{\Omega},\nu,t),
\end{equation}
which is referred to as the diffusion expansion.
By substituting the diffusion expansion \eqref{eq_expansion1} into the macroscopic equations \eqref{eq_RTEmac},
the corresponding diffusion equation is obtained,
\begin{equation}\label{diffusion}
  a\frac{\partial T_0^4}{\partial t}+C_v\frac{\partial T_0}{\partial t}=
  \nabla\cdot\kappa_R\nabla T_0^4,
\end{equation}
where $\kappa_R$ is the Rosseland heat conductivity coefficient,
\begin{equation}\label{eq_rosseland}
    \kappa_R=\frac{ac}{3}\frac{\int \frac{1}{\sigma}\frac{\partial B}{\partial T} \mathrm{d}\nu}
    {\int \frac{\partial B}{\partial T} \mathrm{d}\nu}.
\end{equation}

A routine multigroup treatment can be used to discrete the frequency space \cite{sun2015asymptotic}.
For each energy group, the Planckian degenerates into a uniform distribution in both frequency and angular velocity space. The main topic of this paper is to introduce an implicit UGKWP method that removes the CFL limitation,
and therefore, for simplicity, the formulations for the grey model are presented in the following sections.
For the frequency-dependent RTE with a more general scattering operator, one can refer to our series of papers, 
where a multi-group treatment \cite{Hu2023ugkp} and continuous treatment \cite{li2023ugkp} of the frequency space are presented.
The grey model radiative transfer equation reads
\begin{equation}\label{eq_RTEgrey}
  \left\{
  \begin{aligned}
  &\frac{1}{c}\frac{\partial I}{\partial t}+\frac{1}{\varepsilon}\vec{\Omega}\cdot\nabla I=
  \frac{1}{\varepsilon^2}\sigma\left(\frac{1}{4\pi}\phi-I\right),\\
  &\frac{\partial C_v T}{\partial t}=
  \frac{1}{\varepsilon^2}\sigma\left(\int_{\mathcal{S}^2}  I  \mathrm{d}\vec{\Omega}-\phi\right),
  \end{aligned}
  \right.
\end{equation}
and $\phi=acT^4$. The corresponding macroscopic energy equations are
\begin{equation}\label{eq_RTEmacgrey}
\left\{
  \begin{aligned}
    &\frac{\partial \rho}{\partial t}+\nabla\cdot \vec{F}=
    \frac{c}{\varepsilon^2}\sigma\left(\phi-\rho\right),\\
    &\frac{\partial C_v  T}{\partial t}=\frac{1}{\varepsilon^2}\sigma\left(\rho-\phi\right).
  \end{aligned}
\right.
\end{equation}
For an initial value problem,
\begin{equation}
\left\{
    \begin{aligned}
        &I(\vec{x},0,\vec{\Omega})=I_0(\vec{x},\vec{\Omega}),\\
        &\phi(\vec{x},0)=\phi_0(\vec{x}),
    \end{aligned}
    \right.
\end{equation}
the integral solution to the radiative transfer equation \eqref{eq_RTE} can be written as
\begin{equation}\label{eq_integralsolution0}
    I(\vec{x},t,\vec{\Omega})=\int_0^t\frac{c\sigma}{\varepsilon^2}
    \mathrm{e}^{-\frac{c\sigma}{\varepsilon^2}(t-s)}
    B(\vec{x}(s),s,\vec{\Omega}) \mathrm{d}s+
    \mathrm{e}^{-\frac{c\sigma}{\varepsilon^2} t}I_0(\vec{x}_0),
\end{equation}
where $B(\vec{x},t,\vec{\Omega})=\phi(\vec{x},t,\vec{\Omega})/4\pi$, 
and $\vec{x}(s)=\vec{x}_0+\frac{c\vec{\Omega}}{\varepsilon}(s-t)$ is the characteristics.
Expanding $B_0(\vec{x},\vec{\Omega})=B(\vec{x},0,\vec{\Omega})$ and $I_0(\vec{x},\vec{\Omega})$ up to second order in the space and time,
the integral solution is reformulated as
\begin{equation}
    \begin{aligned}
        I(\vec{x},\vec{\Omega},t)=&
        c_1(t)B_0(\vec{x},\vec{\Omega})+
        c_2(t)\vec{\Omega}\cdot\nabla B_0(\vec{x},\vec{\Omega})+
        c_3(t)\partial_t B_0(\vec{x},\vec{\Omega})\\
        &+c_4(t)I_0(\vec{x},\vec{\Omega})+
        c_5(t)\vec{\Omega}\cdot\nabla I_0(\vec{x},\vec{\Omega}),
    \end{aligned}
\end{equation}
where the coefficients are
\begin{equation}\label{eq_coefficients}
    \begin{aligned}
        &c_1(t)=1-\mathrm{e}^{-\frac{c\sigma}{\varepsilon^2} t}, \quad
        c_2(t)=-\frac{\varepsilon}{\sigma}(1-\mathrm{e}^{-\frac{c\sigma}{\varepsilon^2} t})+
        ct\mathrm{e}^{-\frac{c\sigma}{\varepsilon^2} t},\\
        &c_3(t)=-\frac{\varepsilon^2}{c\sigma}(1-\mathrm{e}^{-\frac{c\sigma}{\varepsilon^2} t})+t,\quad
        c_4(t)=\mathrm{e}^{-\frac{c\sigma}{\varepsilon^2} t},\quad
        c_5(t)=-\frac{c}{\varepsilon}t\mathrm{e}^{-\frac{c\sigma}{\varepsilon^2} t}.
    \end{aligned}
\end{equation}
It can be shown that the integral solution preserves the asymptotic limits of the radiative transfer equation \cite{li2020unified}.
More specifically, in the optically thick regime,
the integral solution converges to the second order asymptotic expansion \eqref{eq_expansion1},
\begin{equation}
    \lim_{\varepsilon\to 0}I(\vec{x},\vec{\Omega},t)=\frac{1}{4\pi}\phi_0-
    \frac{\varepsilon}{4\pi\sigma}\vec{\Omega}\cdot\nabla\phi_0+
    \frac{t}{4\pi}\partial_t\phi_0,
\end{equation}
which gives the diffusion closure of the radiant intensity.
The integral solution is essential in the construction of the IUGKWP methods.
Both the multiscale flux of the radiant energy and 
the closure of the photon distribution function is constructed based on the integral solution.
In the following section, we will present the detailed formulation of the implicit UGKWP method.

\section{Implicit unified gas-kinetic wave-particle method}\label{section_ugkwp}
In this section, we will present the implicit unified gas-kinetic wave-particle (IUGKWP) method 
on a general unstructured mesh.
The IUGKWP method is a continuous development of the multiscale UGKWP method for transport processes \cite{liu2020unified,zhu2019unified}, 
that combines the advantages of the implicit Monte Carlo (IMC) method and the implicit moments method, 
providing a stable and efficient numerical method for the simulation of multiscale photon transport.
In the IUGKWP method, we define a local physical characteristic time 
\begin{equation}\label{eq_tp}
    t_p=\varepsilon L_\infty/c, 
\end{equation}
where 
\begin{equation}
    L_\infty=\max\left(\frac{\rho}{\nabla \rho},\Delta x\right)
\end{equation} 
is the local characteristic length of the flow field.
The photon streaming processes can be categorized into 
long-$\lambda$ streaming process with free stream time longer than $t_p$, 
and short-$\lambda$ streaming process with free stream time shorter than $t_p$.
In the IUGKWP method, the long-$\lambda$ streaming process is solved by the IMC method, and 
the short-$\lambda$ streaming process is solved by the implicit moment method.
The closure of photon distribution is constructed based on the integral solution of the radiative transfer equation.
In the following subsections, we present 
(i) categorization of the photon streaming processes, 
(ii) IMC for long-$\lambda$ streaming process, 
(iii) implicit Moments method for short-$\lambda$ streaming process, and 
(iv) the closure of the photon distribution function. 
The algorithm of the IUGKWP method is presented at the end of this section.

\subsection{Categorization of the photon streaming processes}\label{section-categorization}
According to the radiative transport equation, 
the photon free time $\tau=\lambda/c$ follows a exponential distribution,
\begin{equation}\label{eq_dist_tau}
    F_{\tau}(\tau<s)=1-\exp(-c\sigma s), \quad s>0.
\end{equation}
The probability of long-$\lambda$ streaming processes with $\tau>t_p$ is 
\begin{equation}\label{eq_pl}
    P_l=\exp(-c\sigma t_p),
\end{equation}
and the probability of short-$\lambda$ streaming processes with $\tau\le t_p$ is 
\begin{equation}\label{eq_ps}
    P_s=1-\exp(-c\sigma t_p).
\end{equation}
For the long-$\lambda$ streaming processes, the free time follows the cumulative distribution
\begin{equation}\label{eq_dist_taul}
    F_{\tau,l}(s)=\left\{
    \begin{aligned}
        &1-\exp(-c\sigma (s-t_p)), \quad s\ge t_p,\\
        &0, \quad s<t_p,
    \end{aligned}\right.
\end{equation}
and the probability density function
\begin{equation}
    f_{\tau,l}(s)=\left\{
    \begin{aligned}
        &c\sigma \exp(-c\sigma (s-t_p)), \quad s\ge t_p,\\
        &0, \quad s<t_p.
    \end{aligned}\right.
\end{equation}
The mean free time of a long-$\lambda$ stream process is 
\begin{equation}
    \tau_l=\int_{t_p}^\infty f_{\tau,l}(s)s \mathrm{d}s=
    t_p+\frac{1}{c\sigma},
\end{equation}
For the short-$\lambda$ streaming processes, the free time follows the distribution
\begin{equation}\label{eq_dist_taus}
    F_{\tau,s}(s)=\left\{
    \begin{aligned}
        &\frac{1-\exp(-c\sigma s)}{1-\exp(-c\sigma t_p)},\quad s< t_p,\\
        &0, \quad s\ge t_p,
    \end{aligned}
    \right.
\end{equation}
and the probability density function
\begin{equation}
    f_{\tau,s}(s)=\left\{
    \begin{aligned}
        &\frac{c\sigma\exp(-c\sigma s)}{1-\exp(-c\sigma t_p)},\quad s< t_p,\\
        &0, \quad s\ge t_p.
    \end{aligned}
    \right.
\end{equation}
The mean free time of a short-$\lambda$ stream is 
\begin{equation}
    \tau_s=\int_0^{t_p} f_{\tau,s}(s)s \mathrm{d}s=
    \frac{1}{c\sigma}-\frac{t_p}{\left(e^{c\sigma t_p}-1\right)}.
\end{equation}
For a successive series of stream-collision processes, 
the number of short-$\lambda$ stream processes 
between two long-$\lambda$ stream processes follows a geometric distribution, i.e.,
\begin{equation}
    P_{n_s}(n_s)=P_s^{n_s}P_l, \quad n_s\ge 0.
\end{equation}
and therefore, the total short-$\lambda$ stream time between 
two long-$\lambda$ stream processes is $n_s\tau_s$.
In a time step $\Delta t$, for a photon particle, 
the probability of the long-$\lambda$ stream process number $n_l\ge1$ is
\begin{equation}\label{eq_pp}
    P_p=\sum_{n_s\tau_s<\Delta t}P_s^{n_s}P_l=1-P_s^{\lceil \frac{\Delta t}{\tau_s} \rceil}.
\end{equation}
We present the IMC equations for the evolution of long-$\lambda$ stream processes in the next subsection.

\subsection{Implicit Monte Carlo method for long-\texorpdfstring{$\lambda$}{} stream processes}\label{section-long}
The long-$\lambda$ stream processes are the non-equilibrium transport processes, 
which drives the photon distribution to deviate from the local equilibrium.
To capture the non-equilibrium transport process, the kinetic equation needs to be solved.
In the IUGKWP method, the long-$\lambda$ stream processes are simulated by the IMC method.
The IMC formulation is developed based on an implicit discretization of the energy exchange term, i.e.,
\begin{equation}\label{eq_imc0}
    \left\{
    \begin{aligned}
        &\frac{1}{c}\frac{\partial I}{\partial t}+\frac{1}{\varepsilon}\vec{\Omega}\cdot\nabla I=
        \frac{1}{\varepsilon^2}\sigma\left(\phi^{n+1}b^n-I\right),\\
        &\frac{1}{\beta^n}\frac{\phi^{n+1}-\phi^{n}}{\Delta t}=
        \frac{1}{\varepsilon^2}c\sigma_p\left(\frac{1}{\sigma_p}\int_{R}\int_{s^2}\sigma I \mathrm{d}\Omega\mathrm{d}\nu-\phi^{n+1}\right),
    \end{aligned}
    \right.
\end{equation}
where $\beta =\frac{1}{c C_v}\frac{\partial \phi}{\partial T}$ is the ratio of specific heat, 
$b(T,\nu)=B(T,\nu)/\phi$ is the normalized Planckian, 
and $\sigma_p=\int_{R} \sigma b \mathrm{d}\nu$ is the Planck-averaged opacity.
The radiation flux $\phi^{n+1}$ can be expressed as
\begin{equation}\label{eq_imcphi}
    \phi^{n+1}=\frac{1}{1+f}\phi^{n}+\frac{f}{1+f}\frac{1}{\sigma_p}\int_{R}\int_{s^2}\sigma I\mathrm{d}\Omega\mathrm{d}\nu,
\end{equation}
where $f=\frac{\beta c\sigma_p\Delta t}{\varepsilon^2}$ is the Fleck factor\cite{fleck1971implicit}.
Substituting Eq.\eqref{eq_imcphi} into Eq.\eqref{eq_imc0}, we derive the evolution equation for the long-$\lambda$ particles and the corresponding material temperature equation
\begin{equation}\label{eq_imc}
    \left\{
    \begin{aligned}
    &\frac{1}{c}\frac{\partial I}{\partial t}+\frac{1}{\varepsilon}\vec{\Omega}\cdot\nabla I=
    \underbrace{\frac{1}{\varepsilon^2}\frac{1}{1+f}\sigma^nB^n}_{\text{emission}}+
    \underbrace{\frac{1}{\varepsilon^2}\frac{f}{1+f}\frac{\sigma^nb^n}{\sigma_p^n}\int_{R}\int_{s^2}\sigma^nI\mathrm{d}\Omega\mathrm{d}\nu}_{\text{effective scattering}}-
    \underbrace{\frac{1}{\varepsilon^2}\sigma^nI}_{\text{absorption}},\\
    &T^{*}=T^{n}-\underbrace{\frac{1}{C_v}\frac{1}{\varepsilon^2}\frac{1}{1+f}
    \int_{t^n}^{t^{n+1}}\int_{R}\int_{s^2} \sigma^n B^{n} \mathrm{d}\Omega\mathrm{d}\nu\mathrm{d}t}_{\text{emission}}+
    \underbrace{\frac{1}{C_v}\frac{1}{\varepsilon^2}\int_{t^n}^{t^{n+1}}\int_{R}\int_{s^2} \sigma^n I \mathrm{d}\Omega\mathrm{d}\nu\mathrm{d}t}_{\text{absorption}}.
    \end{aligned}
    \right.
\end{equation}
Here $T^*$ stands for the evolved material temperature after calculating 
the long-$\lambda$ transport processes and the corresponding evolved radiation energy is 
$\rho^*=\int_{\mathcal{R}}\int_{\mathcal{S}^2} I^* \mathrm{d}\omega \mathrm{d}\nu$.
For the grey radiative transfer equation \eqref{eq_RTEgrey}, 
terms can be simplified as $\sigma_p=\sigma$, $B(\vec{x},t,\vec{\Omega})=acT^4/4\pi$, 
and $b(\vec{x},t,\vec{\Omega})=1$.
To evolve the radiant flow field from $t^n$ to $t^{n+1}$, 
the IUGKWP method first evolves all the long-$\lambda$ stream processes by IMC equations \eqref{eq_imc}.
Different from the traditional IMC method, the IUGKWP method only tracks the long-$\lambda$ non-equilibrium transport processes using the Monte Carlo method.
For the short-$\lambda$ transport processes, the IUGKWP method solves the corresponding implicit moments equations. 
It is derived in the particle categorization subsection \ref{section-categorization} the probability of a particle experiencing at least one long-$\lambda$ stream process is 
\begin{equation}\label{eq_largeproportion}
    P_p=1-P_s^{\lceil \frac{\Delta t}{\tau_s} \rceil},
\end{equation}
and those particles are called the long-$\lambda$ particles.
In the sampling process of IUGKWP method, for cell $C_i$ with radiant flux $\rho_i|C_i|$, 
the number of long-$\lambda$ particles to be sampled is 
\begin{equation}\label{eq_np}
    N_p=\rho_i|C_i|(1-P_s^{\lceil \frac{\Delta t}{\tau_s} \rceil})/w_{\text{ref}},
\end{equation}
where $w_{\text{ref}}$ is the reference MC particle energy.
For each long-$\lambda$ particle, we track all long-$\lambda$ transport processes 
$\text{LT}_k,$ for $k=1,2,...N_L$, from $t^n$ to $t^{n+1}$.
The long-$\lambda$ stream time is sampled from the distribution function \eqref{eq_dist_taul}, 
The time of short-$\lambda$ transport processes between two long-$\lambda$ transport processes 
$\text{LT}_k$ and $\text{LT}_{k+1}$, i.e., the waiting time, is calculated by
\begin{equation}\label{eq_shorttime}
    \Delta t_{s,k}=\left\{
    \begin{aligned}
        &\frac{\tau_s}{P_l}\left(P_s^{n_{0}}-n_{0}P_s^{n_{0}}P_l+1\right), \quad k=0,\\
         &\frac{\tau_s}{P_l}, \quad k\ge 1,
    \end{aligned}
    \right.
\end{equation}
where $n_{0}=\lceil\frac{\Delta t}{\tau_s}\rceil$. 

The MC computational cost of IUGKWP can be estimated 
by calculating the particle number and particle collision frequency.
In the optically thin regime, the MC particle number of IUGKWP and traditional IMC are similar, i.e.,
\begin{equation}
    \lim_{\sigma\to 0}N_p=\rho_i|C_i|/\omega_{\text{ref}},
\end{equation}
and for each particle, the collision frequency $\nu$ is very low,
\begin{equation}
    \lim_{\sigma\to 0}\nu=\lim_{\sigma\to 0}\left(t_p+\frac{1}{c\sigma}\right)^{-1}=0.
\end{equation}
Therefore, both methods achieve high efficiency.
In the optically thick regime, 
the efficiency of the traditional IMC method is low 
due to a tremendous amount of effective collision.
However, the MC particle number of IUGKWP converges to zero, i.e.,
\begin{equation}
    \lim_{\sigma\to \infty}N_p=0.
\end{equation}
Therefore, the IUGKWP method is much more effective than IMC in the optically thick regime.

Once the long-$\lambda$ transport processes are simulated, we concurrently obtain 
(i) the energy exchange between photon and material contributed by long-$\lambda$ transport processes; 
(ii) the non-equilibrium part of photon distribution at $t^{n+1}$.
In the next subsection, we present the implicit calculation of short-$\lambda$ transport processes.

\subsection{Implicit moments method for short-\texorpdfstring{$\lambda$}{} transport processes}\label{section_short}
The integral solution of the radiative transport equation \eqref{eq_RTEmacgrey} 
in time interval $t\in[t^{n}+t_p,t^{n+1}]$ is
\begin{equation}\label{eq_integralsolution1}
    \begin{aligned}
        I(\vec{x},t,\vec{\Omega})=&\int_{t^n}^{t}
        \mathrm{e}^{-\frac{c\sigma}{\varepsilon^2}(t-s)}
        B(\vec{x}(s),s,\vec{\Omega}) \frac{c\sigma}{\varepsilon^2} \mathrm{d}s
        +\mathrm{e}^{-\frac{c\sigma}{\varepsilon^2} t}I_0(\vec{x}_0,\vec{\Omega})\\
        =&\underbrace{\int_{t-t_p}^{t}
        \frac{\mathrm{e}^{-\frac{c\sigma}{\varepsilon^2}(t-s)}
        -\mathrm{e}^{-\frac{c\sigma}{\varepsilon^2}t_p}}
        {1-\mathrm{e}^{-\frac{c\sigma}{\varepsilon^2}t_p}}
        B(\vec{x}(s),s,\vec{\Omega}) \frac{c\sigma}{\varepsilon^2} P_s \mathrm{d}s}
        _{\text{short-$\lambda$ near-equilibrium part $I^w$}}
        +\underbrace{\int_{t-t_p}^{t}
        B(\vec{x}(s),s,\vec{\Omega}) \frac{c\sigma}{\varepsilon^2} P_l \mathrm{d}s}
        _{\text{long-$\lambda$ non-equilibrium part $I^p$}}\\
        &+\underbrace{\int_{t^n}^{t-t_p}
        \mathrm{e}^{-\frac{c\sigma}{\varepsilon^2}(t-s)}
        B(\vec{x}(s),s,\vec{\Omega}) \frac{c\sigma}{\varepsilon^2}  \mathrm{d}s
        +\mathrm{e}^{-\frac{c\sigma}{\varepsilon^2} t}I_0(\vec{x}_0,\vec{\Omega})}
        _{\text{long-$\lambda$ non-equilibrium part $I^p$}}.
    \end{aligned}
\end{equation}
Here, $\vec{x}(s)=\vec{x}(t)-\frac{c\vec{\Omega}}{\varepsilon}(t-s)$ is the characteristics,
$I_0(\vec{x},\vec{\Omega})=I(\vec{x},t^n,\vec{\Omega})$ is the initial condition at $t^n$,
$P_s$ is the short-$\lambda$ probability given in Eq. \eqref{eq_ps},
and $P_l$ is the long-$\lambda$ probability given in Eq. \eqref{eq_pl}. 
Note that in the time interval $t\in[t^n,t-t_p]$, 
the emission probability of the long-$\lambda$ particle is $P_l$, 
and the absorption rate of the long-$\lambda$ particle is 
$\mathrm{e}^{-\frac{c\sigma}{\varepsilon^2}(t-s-t_p)}$, 
the multiplication of these two probabilities gives
\begin{equation}
    \mathrm{e}^{-\frac{c\sigma}{\varepsilon^2}(t-s)}=
    \mathrm{e}^{-\frac{c\sigma}{\varepsilon^2}(t-s-t_p)}P_l.
\end{equation}
The energy flux and energy exchange between long-$\lambda$ particles and material 
are calculated by the IMC method as presented in subsection \ref{section-long}. 
In this subsection, we present the numerical scheme for the energy flux and energy exchange 
between short-$\lambda$ particles and material.
We expand the local Plankian $B^{n+1}=B(\vec{x},t^{n+1},\vec{\Omega})$ 
to second order around $(\vec{x},t^{n+1})$,
and the distribution of the emitted short-$\lambda$ particle distribution $I^w$ 
can be expressed as
\begin{equation}\label{eq_Iw}
    \begin{aligned}
        I^w(\vec{x},t,\vec{\omega})&=\int_{t^{n+1}-t_p}^{t}
        \frac{\mathrm{e}^{-\frac{c\sigma}{\varepsilon^2}(t-s)}
        -\mathrm{e}^{-\frac{c\sigma}{\varepsilon^2}t_p}}
        {1-\mathrm{e}^{-\frac{c\sigma}{\varepsilon^2}t_p}}
        B(\vec{x}(s),s,\vec{\Omega}) P_s \frac{c\sigma}{\varepsilon^2}\mathrm{d}s\\
        &=C_1(t)B^{n+1}+
        C_2(t)\vec{\Omega}\cdot\nabla B^{n+1}+
        C_3(t)\partial_t B^{n+1}+O(t^2).
    \end{aligned}
\end{equation}
and the coefficients are
\begin{equation}\label{eq_coeff1}   
    \begin{aligned}
        C_1(t)=&1-e^{-\frac{c\sigma}{\varepsilon^2}t_p}
        -\frac{c\sigma}{\varepsilon^2}t_p\mathrm{e}^{-\frac{c\sigma}{\varepsilon^2}t_p},\\
        C_2(t)=&-\frac{\varepsilon}{\sigma}
        \left(1-e^{-\frac{c\sigma}{\varepsilon^2}t_p}\right)
        +\frac{c}{\varepsilon}t_p
        e^{-\frac{c\sigma}{\varepsilon^2}t_p}
        +\frac{c^2\sigma}{2\varepsilon^3}t_p^2\mathrm{e}^{-\frac{c\sigma}{\varepsilon^2}t_p},\\
        C_3(t)=&-\frac{\varepsilon^2}{c\sigma}-\left(t^{n+1}-t\right)+
        \left(\frac{\varepsilon^2}{c\sigma}+t^{n+1}-t\right)e^{-\frac{c\sigma}{\varepsilon^2}t_p}\\
        &+\left(\frac{c\sigma}{\varepsilon^2}\left(t^{n+1}-t\right)+1\right) t_p e^{-\frac{c\sigma}{\varepsilon^2}t_p}
        +\frac{c\sigma}{2\varepsilon^2}t_p^2\mathrm{e}^{-\frac{c\sigma}{\varepsilon^2}t_p}.
    \end{aligned}
\end{equation}
The distribution of the locally emitted photon \eqref{eq_Iw} 
is fully determined by the macroscopic radiance field, 
the order of which can be reduced.
Under a finite volume framework, the low-order moments equations of the short-$\lambda$ particles 
can be derived, and the numerical flux can be derived by taking moments to Eq. \eqref{eq_Iw}.
Once the radiance field is evolved, the short-$\lambda$ near-equilibrium distribution can be closed 
according to Eq. \eqref{eq_Iw}.

We discretize the computational domain $\mathcal{D}$ into control volumes $\mathcal{D}=\cup_{i\in N_c}\mathcal{C}_i$.
The cell averaged value of the macroscopic quantities, such as the radiant density, the emitted radiant density, and the absorption coefficient, are defined as
\begin{equation}\label{eq_cellaveragemacro}
    \rho_i^n=\frac{1}{|\mathcal{C}_i|}\int_{\mathcal{C}_i}\rho_i(\vec{x},t^n) \mathrm{d} \vec{x},\quad
    \phi_i^n=\frac{1}{|\mathcal{C}_i|}\int_{\mathcal{C}_i}\phi_i(\vec{x},t^n) \mathrm{d} \vec{x},\quad
    \sigma_i^n=\frac{1}{|\mathcal{C}_i|}\int_{\mathcal{C}_i}\sigma_i(\vec{x},t^n) \mathrm{d} \vec{x}.
\end{equation}
The numerical evolution equations for macroscopic quantities are derived by 
integrating Eq.\eqref{eq_RTEmac} in space $\mathcal{C}_i$ and time $[t^n,t^{n+1}]$,
\begin{equation}\label{eq_ime}
\left\{
    \begin{aligned}
        &\rho^{n+1}_i=\rho^*_i-\Delta tc\mathcal{F}^w_i+
        \Delta t \frac{c\sigma_i^{n+1}}{\varepsilon^2}(\phi^{n+1}_i-\rho^{n+1}_i),\\
        &T^{n+1}_i=T^*_i+\Delta t\frac{c\sigma_i^{n+1}}{\varepsilon^2 C_v}
        (\rho^{n+1}_i-\phi^{n+1}_i),\\
        &\phi^{n+1}=\phi^{*}+C_v\beta^{n+1}(T^{n+1}_i-T^*_i),
    \end{aligned}
\right.
\end{equation}
where $\beta$ is the scaled heat capacity
\begin{equation}\label{eq_scaledcapacity}
  \beta^{n+1}_i=\frac{1}{C_v}\left(\frac{\partial \phi}{\partial T}\right)_i^{n+1}
  =\left(\frac{4acT^3}{C_v}\right)_i^{n+1}.
\end{equation}
The short-$\lambda$ particle flux can be implicitly calculated by 
substituting the second-order expansion of the local Planckian around $t^{n+1}$
\begin{equation}\label{eq_fluxw1}
    \begin{aligned}
        \mathcal{F}^{w}_i&=
        \frac{1}{|\mathcal{C}_i|}\frac{1}{\Delta t}\frac{c}{\varepsilon}
        \int_{t^n}^{t^{n+1}} \int_{\partial \Omega_i} \int_{\mathcal{S}^2}
        \vec{\Omega}\cdot\vec{n} I^{w}(\vec{x},\vec{\Omega},t)
        \mathrm{d}\vec{\Omega} \mathrm{d}\vec{x} \mathrm{d}t\\
        &=\sum_{L_l \in \partial \mathcal{C}_i}
        \frac{|L_l|}{|\mathcal{C}_i|}\frac{1}{t_p}\frac{c}{\varepsilon}
        \int_{t^{n+1}-t_p}^{t^{n+1}} \int_{\mathcal{S}^2}
        \vec{\Omega}\cdot\vec{n}_l I^{w}(\vec{l}_m,\vec{\Omega},t)
        \mathrm{d}\vec{\Omega} \mathrm{d}t\\
        &=\sum_{L_l \in \partial \mathcal{C}_i}
        \frac{|L_l|}{|\mathcal{C}_i|}\kappa^w_{\text{eff}}
        \nabla \phi^{n+1}_{m} \vec{n}_l,
    \end{aligned}
\end{equation}
where $\vec{n}_l$ is the outer normal vector 
of cell interface $L_l\in \partial \mathcal{C}_i$,
$|L_l|$ is the length of interface $L_l$,
and $|\mathcal{C}_i|$ is the volume of cell $\mathcal{C}_i$.
The effective heat conduction coefficient $\kappa^{w}_{\text{eff}}$ is
\begin{equation}\label{kappa-ugkp}
    \begin{aligned}
        \kappa^w_{\text{eff}}&=\frac{1}{ t_p}
        \frac{c}{\varepsilon}\int_{t^{n+1}- t_p}^{t^{n+1}}
         \frac{1}{3}c_2(t)  \mathrm{d} t \\
        &=-\frac{c}{3\sigma_l}
        \left(1-e^{-\frac{c\sigma}{\varepsilon^2}t_p}
        -\frac{c\sigma}{\varepsilon^2}t_p
        e^{-\frac{c\sigma}{\varepsilon^2}t_p}
        -\frac{c^2\sigma^2}{2\varepsilon^4}t_p^2
        \mathrm{e}^{-\frac{c\sigma}{\varepsilon^2}t_p}\right)\\
        &=-\frac{c}{3\sigma_l}\mathcal{L}_{p},
    \end{aligned}
\end{equation}
where $\mathcal{L}_{p}$ is the UGKP flux limiter
\begin{equation}
    \mathcal{L}_{p}=1-e^{-\frac{c\sigma}{\varepsilon^2}t_p}
        -\frac{c\sigma}{\varepsilon^2}t_p
        e^{-\frac{c\sigma}{\varepsilon^2}t_p}
        -\frac{c^2\sigma^2}{2\varepsilon^4}t_p^2
        \mathrm{e}^{-\frac{c\sigma}{\varepsilon^2}t_p}.
\end{equation}
In the optically thick regime, the effective diffusion coefficient of IUGKWP converges to 
the asymptotic diffusion coefficient of RTE, i.e.,
\begin{equation}
    \lim_{\sigma_l\to\infty}\mathcal{L}_{p}=1, \quad 
    \lim_{\sigma_l\to\infty}\kappa^w_{\text{eff}}=-\frac{c}{3\sigma_l}
\end{equation}
The scattering coefficient at the cell interface $L_l$ is calculated as
\begin{equation}
    \sigma_l=\frac{\sigma_i\sigma_j}{\sigma_i+\sigma_j},
\end{equation}
where $i,j$ are the id of cells sharing interface $L_l$.

The reconstruction of the gradient of $\phi$ in Eq.\eqref{eq_fluxw1}
is shown in figures \ref{fig_ninepoint},
where $c_i$ and $c_j$ are the centers of cell $\mathcal{C}_i$ and $\mathcal{C}_j$
sharing a cell interface $L_l$.
Two vertexes of $L_l$ are $\vec{v}_1$ and $\vec{v}_2$,
and the middle point of edge $L_l$ is $v_m$.
The unit tangential vector of $L_l$ is $\vec{\tau}$.
The unit normal vector pointing to $\mathcal{C}_j$ is $\vec{n}$.
The unit vectors along directions $\overrightarrow{c_i v_m}$,
$\overrightarrow{v_m c_j}$,$\overrightarrow{c_i c_j}$
are $\vec{\tau}_i$, $\vec{\tau}_j$, and $\vec{\tau}_{ij}$, respectively.
The angle from $\vec{n}$ to $\vec{\tau}_i$ is $\theta_i$;
the angle from $\vec{n}$ to $\vec{\tau}_j$ is $\theta_j$;
and the angle from $\vec{n}$ to $\vec{\tau}_i$ is $\theta_{ij}$.
Based on the directional derivative $\partial_{\vec{\tau}}\phi$,
and $\partial_{\vec{\tau}_i}\phi$, 
the left normal derivative of $\phi^l$ can be expressed as
\begin{equation}\label{eq_gradientleft}
    \nabla\phi^l_m\cdot\vec{n}=
    \frac{\phi_m-\phi_{c_i}}{|\vec{v}_m-\vec{c}_i|}\sec\theta_i+
    \frac{\phi_{v_2}-\phi_{v_1}}{|\vec{v}_2-\vec{v}_1|}\tan\theta_i.
\end{equation}
Based on the directional derivative $\partial_{\vec{\tau}}\phi$,
and $\partial_{\vec{\tau}_j}\phi$,
the right normal derivative of $\phi^r$ can be expressed as
\begin{equation}\label{eq_gradientright}
    \nabla\phi^r_m\cdot\vec{n}=
    \frac{\phi_{c_j}-\phi_{m}}{|\vec{c}_j-\vec{v}_m|}\sec\theta_j+
    \frac{\phi_{v_2}-\phi_{v_1}}{|\vec{v}_2-\vec{v}_1|}\tan\theta_j.
\end{equation}
The flux continuity condition states
$\nabla\phi^l_m\cdot\vec{n}=\nabla\phi^r_m\cdot\vec{n}$,
and the discretization of $\nabla\phi_m\cdot\vec{n}$ can be derived from Eq.\eqref{eq_gradientleft} and Eq.\eqref{eq_gradientright},
\begin{equation}\label{eq_gradient}
    \nabla\phi_m\cdot\vec{n}=
    \frac{\phi_{c_j}-\phi_{c_i}}{|\vec{c}_j-\vec{c}_i|}\sec\theta_{ij}+
    \frac{\phi_{v_2}-\phi_{v_1}}{|\vec{v}_2-\vec{v}_1|}\tan\theta_{ij}.
\end{equation}
Here, the vertex value of $\phi$ is the averaged value among its surrounding cells, i.e.,
\begin{equation}
    \phi_{v_2}=\sum_{i\in S(v_2)}w_i\phi(c_i),
\end{equation}
where $S(v_2)$ are the set of indexes of $v_2$-surrounding cells, and $w_i=\frac{\kappa_i}{|c_i-v_2|}$ is the averaging weight.
\begin{figure}
     \centering
     \subfigure[]{\includegraphics[width=0.32\textwidth]{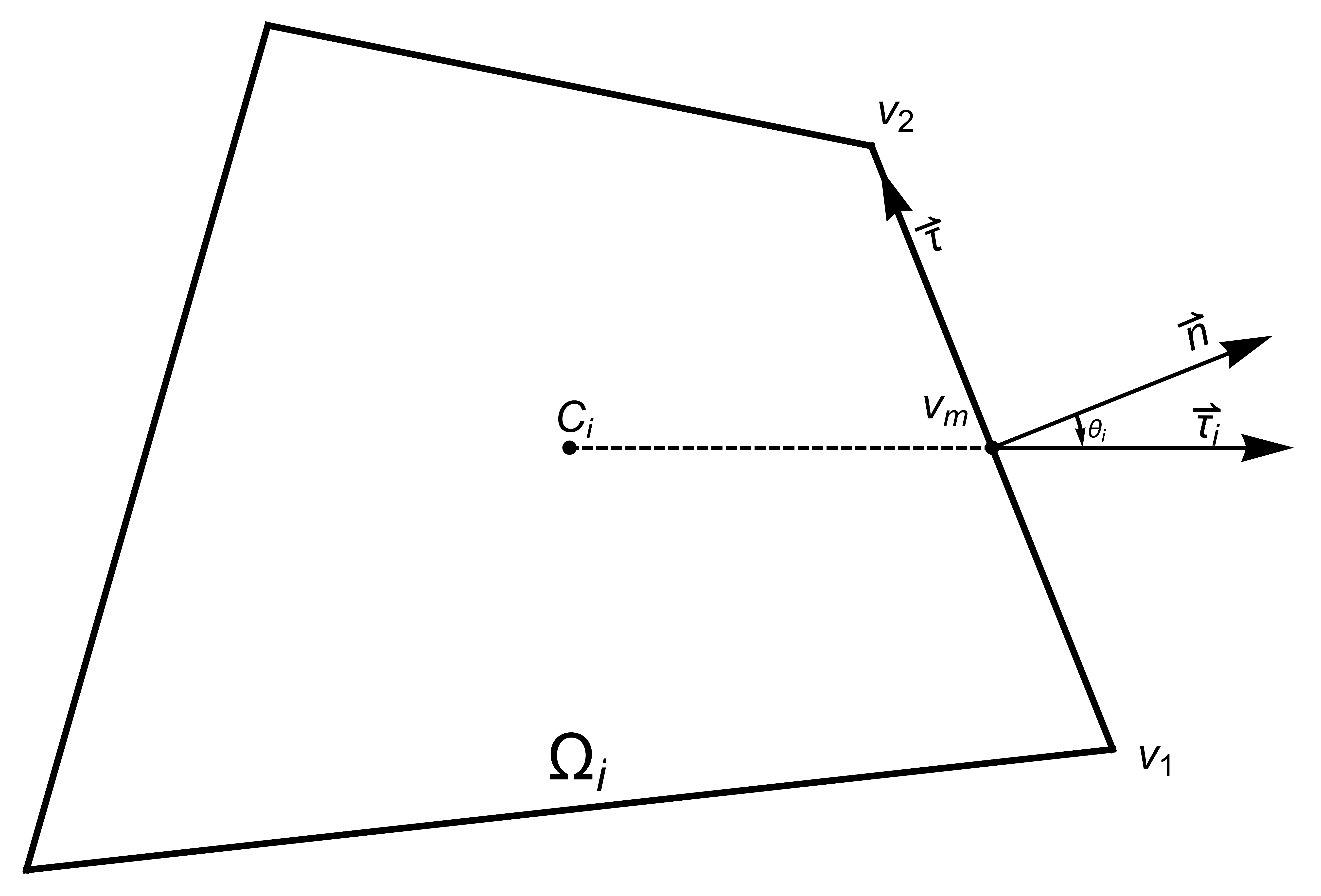}}\hspace{5mm}
     \subfigure[]{\includegraphics[width=0.16\textwidth]{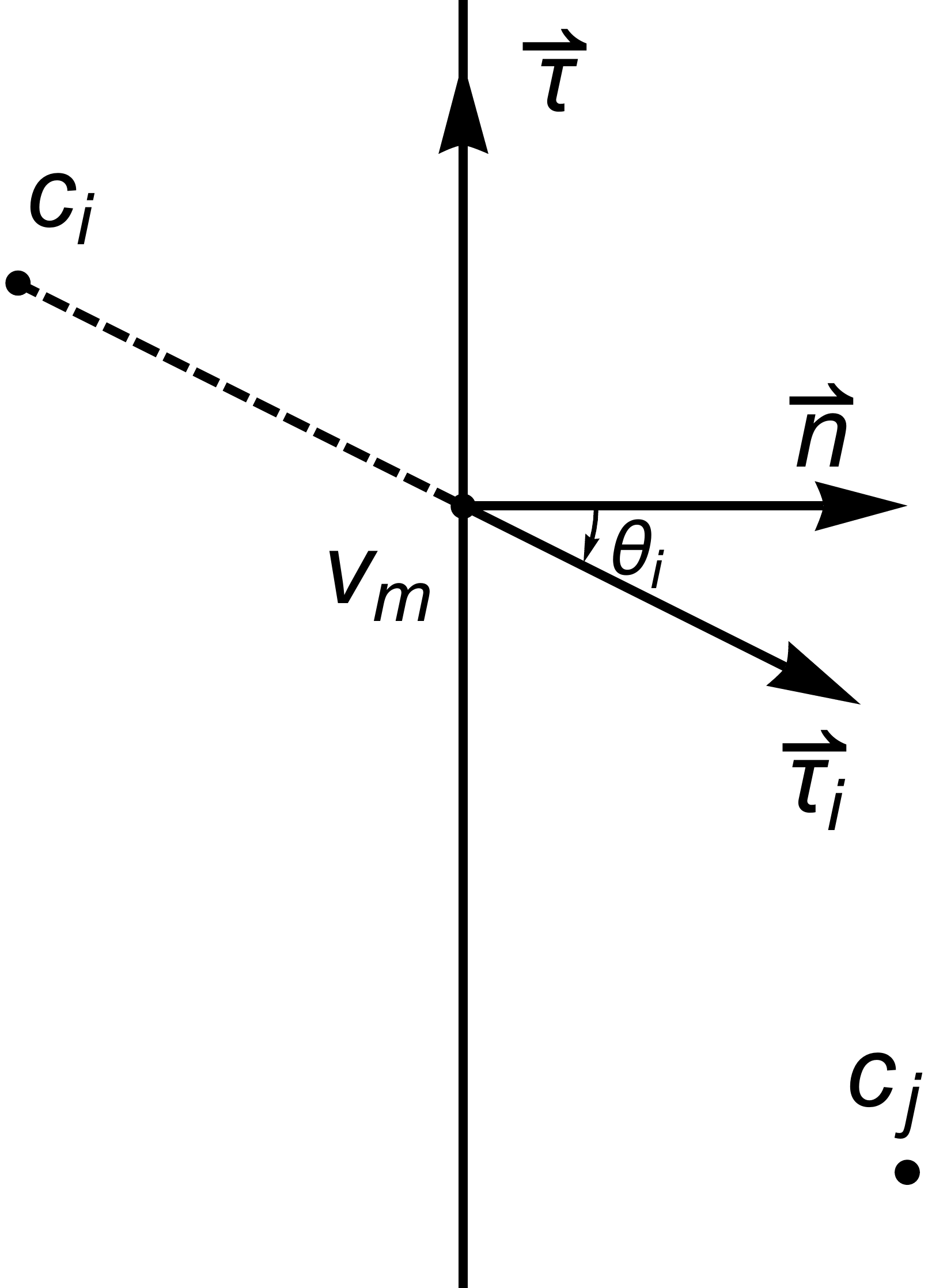}}\hspace{5mm}
     \subfigure[]{\includegraphics[width=0.16\textwidth]{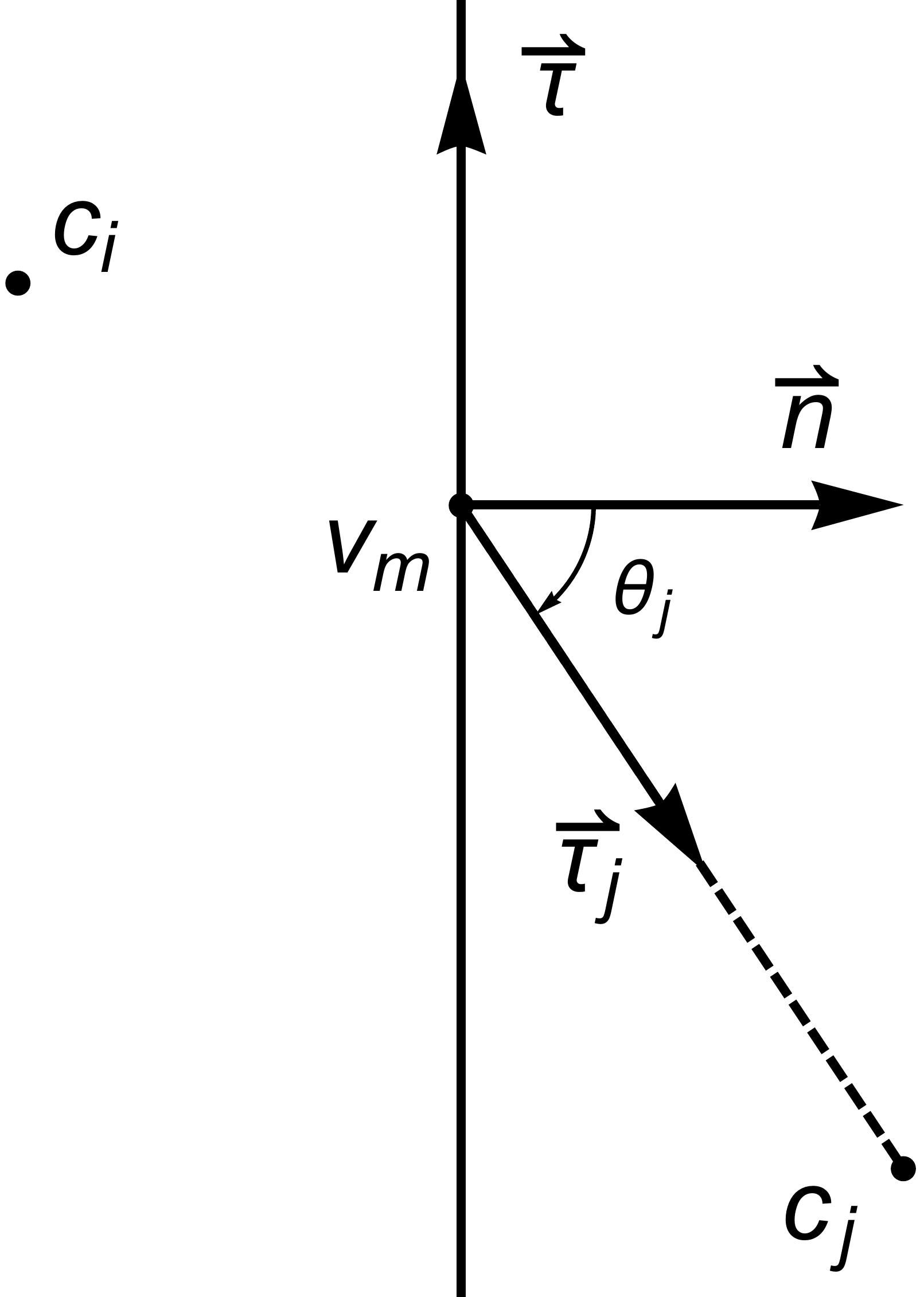}}\hspace{5mm}
     \subfigure[]{\includegraphics[width=0.16\textwidth]{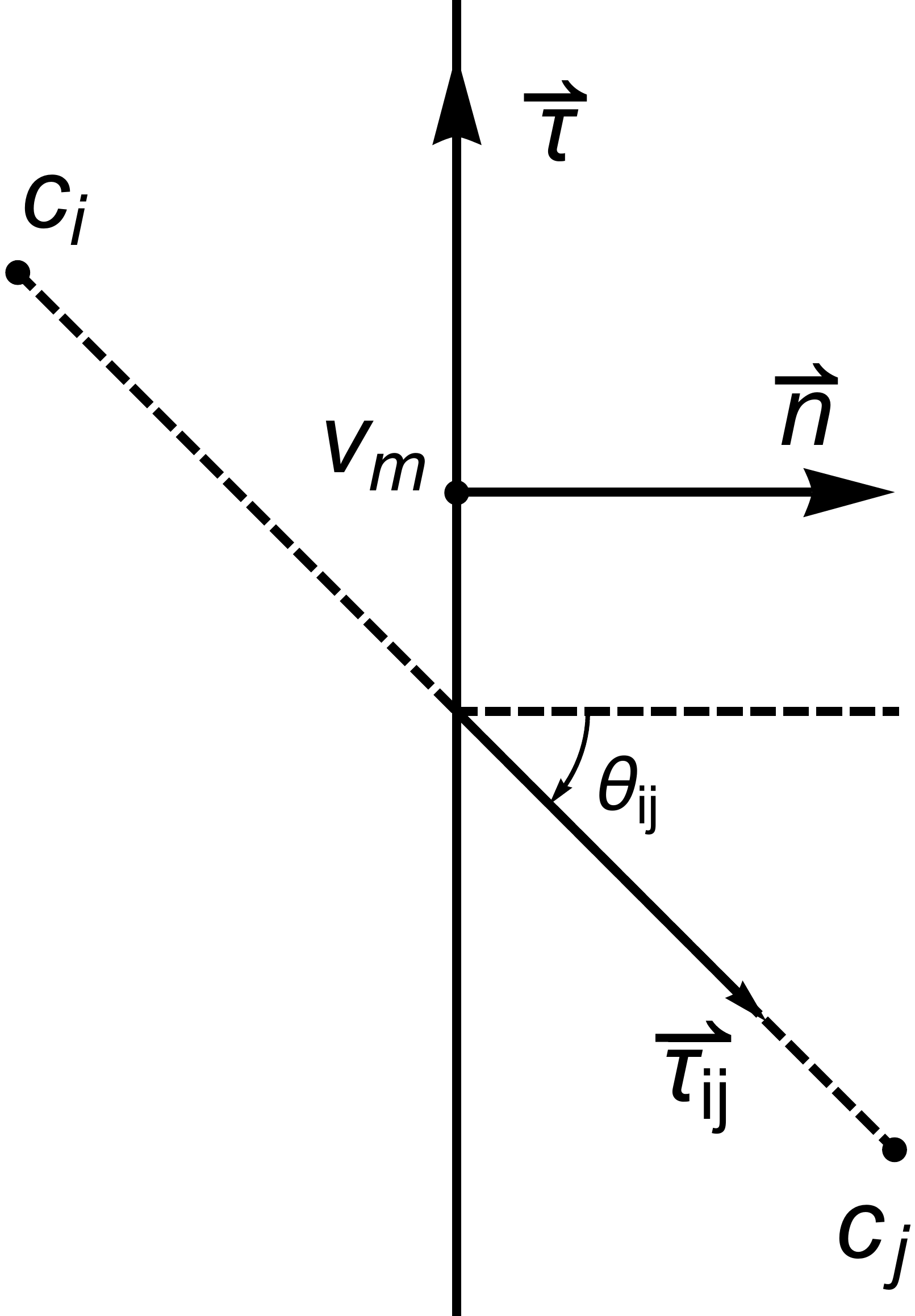}}
  \caption{(a) Sketch of a general mesh cell. The cell center is $c_i$, and the edge center is $v_m$. 
  The edge unit normal vector is $\vec{n}$, and the unit tangential vector is $\vec{\tau}$.
  (b-d) Sketch of the interface geometry of a general mesh.}
  \label{fig_ninepoint}
\end{figure}

The numerical fluxes Eq.\eqref{eq_fluxw1} and 
the discretization of the normal directional derivative of $\phi_m$ Eq.\eqref{eq_gradient}
close the macroscopic evolution equations Eq.\eqref{eq_ime}.
The macroscopic evolution equations are solved by the source iteration method.
Equations \eqref{eq_imc} and \eqref{eq_ime} present the scheme for macroscopic radiance field evolution
and the microscopic radiant intensity for long-$\lambda$ particle.
In the next subsection, we derive the closure modeling of 
the radiant intensity for short-$\lambda$ particle based on the integral solution.

\subsection{Closure modeling of photon distribution}
The integral solution of the radiative transport equation \eqref{eq_RTEmacgrey} 
at time $t^{n+1}$ is
\begin{equation}\label{eq_integralsolution2}
    \begin{aligned}
        I(\vec{x},t^{n+1},\vec{\Omega})=&\underbrace{\int_{t^{n+1}-t_p}^{t^{n+1}}
        \frac{\mathrm{e}^{-\frac{c\sigma}{\varepsilon^2}(t-s)}
        -\mathrm{e}^{-\frac{c\sigma}{\varepsilon^2}t_p}}
        {1-\mathrm{e}^{-\frac{c\sigma}{\varepsilon^2}t_p}}
        B(\vec{x}(s),s,\vec{\Omega}) \frac{c\sigma}{\varepsilon^2} P_s \mathrm{d}s}
        _{\text{short-$\lambda$ near-equilibrium part $I^w$}}
        +\underbrace{\int_{t^{n+1}-t_p}^{t^{n+1}}
        B(\vec{x}(s),s,\vec{\Omega}) \frac{c\sigma}{\varepsilon^2} P_l \mathrm{d}s}
        _{\text{long-$\lambda$ non-equilibrium part $I^p$}}\\
        &+\underbrace{\int_{t^n}^{t^{n+1}-t_p}
        \mathrm{e}^{-\frac{c\sigma}{\varepsilon^2}(t-s)}
        B(\vec{x}(s),s,\vec{\Omega}) \frac{c\sigma}{\varepsilon^2}  \mathrm{d}s
        +\mathrm{e}^{-\frac{c\sigma}{\varepsilon^2} t}I_0(\vec{x}_0,\vec{\Omega})}
        _{\text{long-$\lambda$ non-equilibrium part $I^p$}}.
    \end{aligned}
\end{equation}
A second order expansion of $\phi$ at $t^{n+1}$ gives the closure modeling of 
the photon distribution for the short-$\lambda$ particles
\begin{equation}\label{eq_Iw2}
    \begin{aligned}
        I^w(\vec{x},t^{n+1},\vec{\omega})&=\int_{t^{n+1}-t_p}^{t^{n+1}}
        \frac{\mathrm{e}^{-\frac{c\sigma}{\varepsilon^2}(t-s)}
        -\mathrm{e}^{-\frac{c\sigma}{\varepsilon^2}t_p}}
        {1-\mathrm{e}^{-\frac{c\sigma}{\varepsilon^2}t_p}}
        B(\vec{x}(s),s,\vec{\Omega}) P_s \frac{c\sigma}{\varepsilon^2}\mathrm{d}s\\
        &=C_1(t^{n+1})B^{n+1}+
        C_2(t^{n+1})\vec{\Omega}\cdot\nabla B^{n+1}+
        C_3(t^{n+1})\partial_t B^{n+1}+O(t^2).
    \end{aligned}
\end{equation}
The coefficients are
\begin{equation}\label{eq_coeff2}
    \begin{aligned}
        C_1(t^{n+1})=&1-e^{-\frac{c\sigma}{\varepsilon^2}t_p}
        -\frac{c\sigma}{\varepsilon^2}t_p\mathrm{e}^{-\frac{c\sigma}{\varepsilon^2}t_p},\\
        C_2(t^{n+1})=&-\frac{\varepsilon}{\sigma}
        \left(1-e^{-\frac{c\sigma}{\varepsilon^2}t_p}\right)
        +\frac{c}{\varepsilon}t_p
        e^{-\frac{c\sigma}{\varepsilon^2}t_p}
        +\frac{c^2\sigma}{2\varepsilon^3}t_p^2\mathrm{e}^{-\frac{c\sigma}{\varepsilon^2}t_p},\\
        C_3(t^{n+1})=&-\frac{\varepsilon^2}{c\sigma}
        \left(1-e^{-\frac{c\sigma}{\varepsilon^2}t_p}\right)
        +t_p e^{-\frac{c\sigma}{\varepsilon^2}t_p}
        +\frac{c\sigma}{2\varepsilon^2}t_p^2\mathrm{e}^{-\frac{c\sigma}{\varepsilon^2}t_p}.
    \end{aligned}
\end{equation}
In the re-sample process, only the long-$\lambda$ particles need to be sampled, 
which takes a proportion of $P_p$ as given in Eq.\eqref{eq_largeproportion}.
In the integral equation \ref{eq_integralsolution2}, 
the long-$\lambda$ particles are tracked by the IMC method as presented in section \ref{section-long},
and the short-$\lambda$ particles are re-sampled from Eq.\eqref{eq_Iw2} based on the evolved radiance field.
The macroscopic implicit moments equations \ref{eq_ime}, the microscopic IMC equation \ref{eq_imc}, 
and the closure modeling equation \ref{eq_integralsolution2} close the numerical scheme of the IUGKWP method.

\subsection{Algorithm of IUGKWP method}
The algorithm of the IUGKWP method is composed of three major steps: (i) track long-$\lambda$ transport processes by IMC equations; (ii) calculate short-$\lambda$ transport processes by solving the implicit moments equation system; (iii) close photon distribution. 
The flow chart of the IUGKWP method is presented in figure \ref{flowchart}.
\begin{algorithm}[H]
    \caption{Implicit unified gas-kinetic wave-particle method}
    \begin{algorithmic}[1]
        \STATE {Initialize radiation energy and material temperature}
        \FOR {$ t = 0 $; $ t \le t_{\text{end}} $; $ t +\Delta t $ }
        \STATE {Sample long-$\lambda$ particles and simulate long-$\lambda$ transport process by IMC equations \eqref{eq_imc}}
        \STATE {Simulate short-$\lambda$ transport process by the implicit moments equations \eqref{eq_ime}}
        \STATE {Re-sample the near-equilibrium photon particles by equation \eqref{eq_Iw2} and close photon distribution by the closure relation \eqref{eq_integralsolution2}}
        \ENDFOR
    \end{algorithmic}
\end{algorithm}
\begin{figure}[H]
  \centering
  \includegraphics[width=0.9\textwidth]{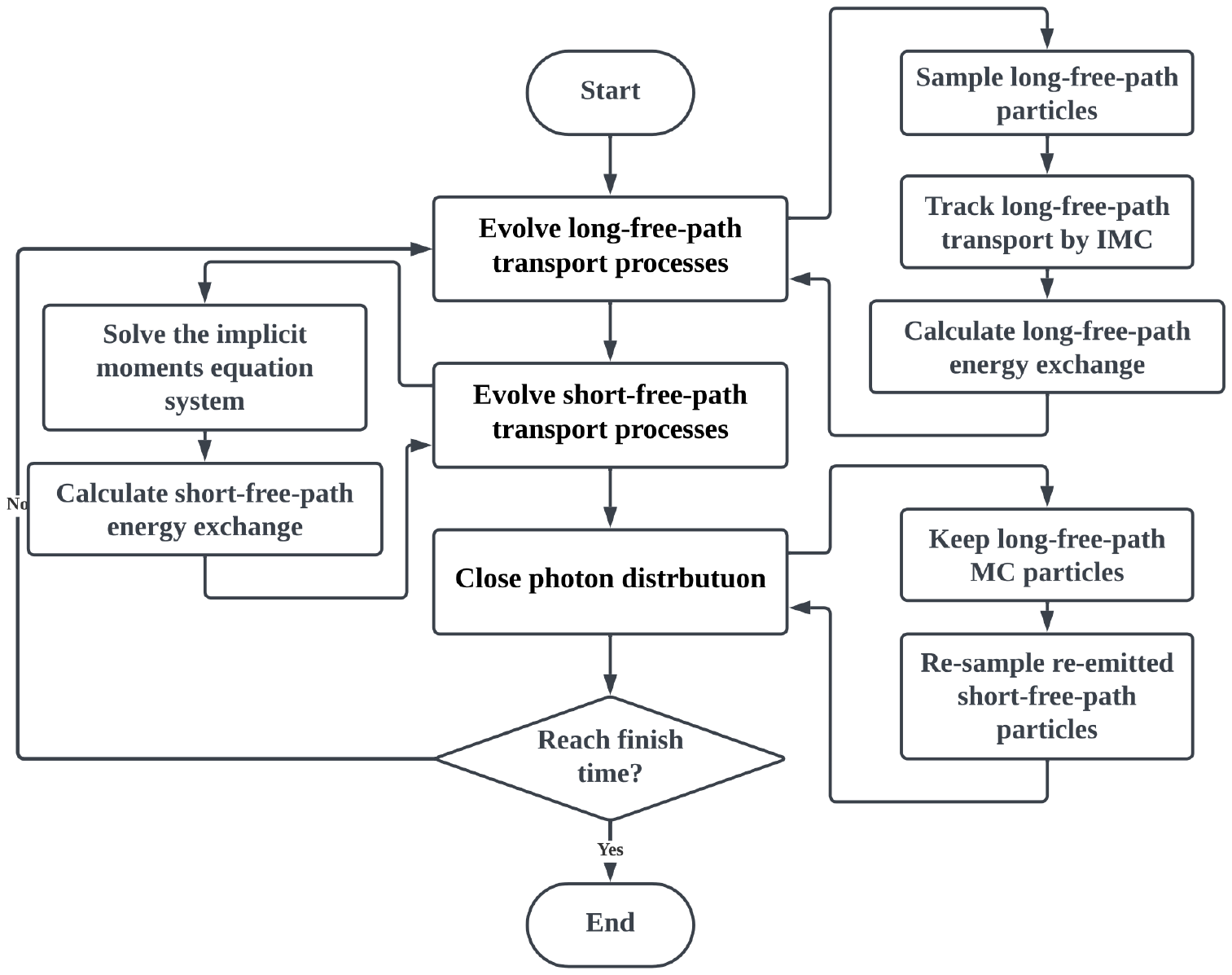}
  \caption{The flow chart of the IUGKWP method}
  \label{flowchart}
\end{figure}

\section{Numerical analysis}\label{section_analysis}
In this section, the numerical properties of the IUGKWP method are discussed, i.e.,
the asymptotic preserving property and the regime adaptive property.
The asymptotic preserving property states that the IUGKWP method converges to 
a nine-point scheme for the diffusion equation in an optically thick regime
and degenerates to a Monte Carlo method in an optically thin regime.
The regime adaptive method states that the DOF of the IUGKWP method exponentially decreases to 
the DOF of a nine-point scheme in an optically thick regime,
indicating that its computational complexity exponentially decrease as the Knudsen number approaches zero.

\begin{theorem}
The IUGKWP method is asymptotic-persevering on a numerical resolution $\Delta x \le \sigma^{-\frac12}$.
\begin{itemize}
\item[I.] In the optically thick regime, the IUGKWP method converges to a nine-point scheme to the diffusion equation \eqref{diffusion} as $\varepsilon\to0$.
\item[II.] In the optically thin regime, the IUGKWP method degenerates to a collisionless radiative transfer equation.
\end{itemize}
\end{theorem}
\begin{proof}
In the optically thick regime, the proportion of long-$\lambda$ particles goes to zero, i.e.,
\begin{equation}
    \lim_{\varepsilon\to0} P_p=\lim_{\varepsilon\to0} 1-P_s^{\lceil \frac{\Delta t}{\tau_s} \rceil}=0,
\end{equation}
and the IMC method for long-$\lambda$ particles gives
$\rho^{*}_i=\rho^n_i$ and $\phi^{*}_i=\phi^n_i$.
The effective diffusion coefficient becomes
\begin{equation}
    \lim_{\varepsilon\to0} \kappa_{\text{eff}}^{w}=-\frac{ac}{3\sigma_m}+o(\varepsilon^2),
\end{equation}
based on which the macroscopic evolution equations of IUGKWP converge to
\begin{equation}\label{theorem11}
\left\{
    \begin{aligned}
      &\rho^{n+1}_i+\phi^{n+1}_i=\rho^n_i+\phi^n_i-
      \Delta t\sum_{L_l \in \partial \mathcal{C}_i}
   \frac{|L_l|}{|\mathcal{C}_i|}\frac{c}{3\sigma_l}
   \nabla \phi^{n+1}_{l}\vec{n}_l,\\
      &\rho^{n+1}_i=\phi^{n+1}_i.
    \end{aligned}
\right.
\end{equation}
Here
\begin{equation}\label{theorem12}
    \sigma_l=\frac{\sigma_i\sigma_j}{\sigma_i+\sigma_j},
\end{equation}
and
\begin{equation}\label{theorem13}
    \nabla\phi_l \cdot\vec{n}=
    \frac{\phi_{c_j}-\phi_{c_i}}{|\vec{c}_j-\vec{c}_i|}\sec\theta+
    \frac{\phi_{v_2}-\phi_{v_1}}{|\vec{v}_2-\vec{c}_1|}\tan\theta.
\end{equation}
The scheme Eq.\eqref{theorem11}-\eqref{theorem13} are consistent to 
the nine-point scheme for diffusion equation Eq.\eqref{diffusion} \cite{Yuan2008}.
Therefore, in the optically thick regime, 
the IUGKWP method is consistent with the diffusion equation system.

In the optically thin regime as $\varepsilon\to\infty$, we have
\begin{equation}\label{theorem14}
    \lim_{\varepsilon\to\infty} c_i(t)=\left\{
    \begin{aligned}
        &0, \quad i=1,2,3,\\
        &1, \quad i=4,
    \end{aligned}
    \right.
\end{equation}
based on which all particles are free streaming long-$\lambda$ particles,
and the microscopic radiative intensity evolves by
\begin{equation}\label{theorem15}
    I(\vec{x},\vec{\Omega},t^{n+1})=I_0(\vec{x}_0).
\end{equation}
Therefore, in the collisionless regime, 
the IUGKWP method is consistent with the Monte Carlo method for collisionless radiative transfer equation.
\end{proof}

\begin{theorem}
The IUGKWP is regime-adaptive.
In an optically thick regime, the computational complexity of the IUGKWP method degenerates to the computational complexity of the diffusion equation.
\end{theorem}
\begin{proof}
In the optically thin regime as $\varepsilon\to\infty$, the Monte Carlo particle number becomes
\begin{equation}\label{theorem21}
    \lim_{\varepsilon\to\infty}N_p=
    \lim_{\varepsilon\to\infty}\rho_i|C_i|(1-P_s^{\lceil \frac{\Delta t}{\tau_s} \rceil})/w_{\text{ref}}=
    \rho_i|C_i|/w_{\text{ref}},
\end{equation}
which is consistent with the IMC method.
In the optically thick regime as $\varepsilon\to0$, we have
\begin{equation}\label{theorem22}
    \lim_{\varepsilon\to0}N_p=
    \lim_{\varepsilon\to0}\rho_i|C_i|(1-P_s^{\lceil \frac{\Delta t}{\tau_s} \rceil})/w_{\text{ref}}=
    0,
\end{equation}
The number of MC particles of IUGKWP exponentially converges to zero as $\varepsilon\to0$.
Therefore, Eq.\eqref{theorem21} and \eqref{theorem22} state that 
the computational complexity of IUGKWP adapts to the flow regime, i.e.,
the computational complexity of the IUGKWP method is similar to the IMC method in the optically thin regime 
and degenerates to the computational complexity of the diffusion equation in the optically thick regime.
\end{proof}

\section{Numerical examples}\label{section_tests}
In this section, we study seven numerical examples, including the Marshak wave, tophat, 2D hohlraum, and 3D hohlraum problems.
For all test cases, the unit of length is taken to be centimeter (cm), the unit of time is nanosecond (ns),
the unit of temperature is kilo electron-volt (KeV), and the unit of energy is gigajoule (GJ).
Under above units, the speed of light $c$ is $29.98\text{cm/ns}$, and the radiation constant $a$ is $0.01372\text{GJ/cm3/KeV}^4$.
The absorption coefficient of the tests ranges from $10^{-4}$ to $10^4$, covering flow regimes from optically thin ballistic regime to optically thick diffusive regime.
We simulate the IUGKWP, SN (S16), and IMC methods with CFL numbers up to $100$.
The IUGKWP and IMC are both Monte Carlo-based methods with similar code structures, while the SN method uses quadrature to discretize the phase space.
The IUGKWP and SN differ in both algorithm and code structure, and their comparison is for accuracy heck.
For efficiency, the IUGKWP and IMC codes are compared to demonstrate the advantages of the current asymptotic preserving Monte Carlo method over the traditional Monte Carlo method.
In the following tests, the opacity changes significantly; namely, the tests are all multiscale photon transport problems.
Simulations are performed on one 2.70 GHz CPU with 128GB memory.

\subsection{Marshak wave-2A problem}
The Marshak wave-2A problem describes the propagation of a thermal wave driven by a constant intensity incident on the left boundary of a slab.
The material opacity is $\sigma=30.0/T^{3} \text{cm}^{-1}$, and the heat capacity $C_v=0.3\text{GJ}/\text{KeV}/\text{cm}^{-3}$.
The initial temperature is in equilibrium with $T_r=T_e=10^{-6} \text{KeV}$, and the incident intensity on the left boundary is $T_r=1.0\text{KeV}$.
Three algorithms, i.e., the IUGKWP, IMC, and SN methods, are implemented in the 2D simulation with $5\times10^{-3}\text{cm}$ in the y-direction and $2\times10^{-1}\text{cm}$ in the x-direction.
The physical domain is discretized into triangular meshes with cell size $2.5\times10^{-4}\text{cm}$.
The IUGKWP method is performed with two CFL numbers, $1$ and $10$.
The reference particle energy is $5\times10^{-11}\text{GJ}$.

For accuracy, we compare the IUGKWP solutions with the reference SN solution as shown in Fig. \ref{fig_marshaka1}.
The material temperature and radiation temperature at times $t=0.2,0.4,0.6,0.8,1.0$ are plotted, respectively.
It can be observed that the IUGKWP solutions with CFL numbers 1 and 10 agree well with the SN solutions.
For efficiency, the IMC takes 61.3mins with CFL=1 and 45.2mins with CFL=10.
The IUGKWP method takes 49.3mins with CFL=1 and 32.1mins with CFL=10.
The IUGKWP is generally $20-30\%$ faster than IMC in the relatively optically thin regime.
The time-saving is due to IUGKWP's effective particle tracking and particle scattering algorithm.
For particle tracking, the IUGKWP only tracks the non-equilibrium particles, which takes about 80 percent of IMC particles for the Marshak wave-2A problem.
For scattering, the IUGKWP avoids the calculation of IMC's effective scattering, which is marginal for the Marshak wave-2A problem.

\begin{figure}
  \centering
  \subfigure[Material temperature]{\includegraphics[width=0.49\textwidth]{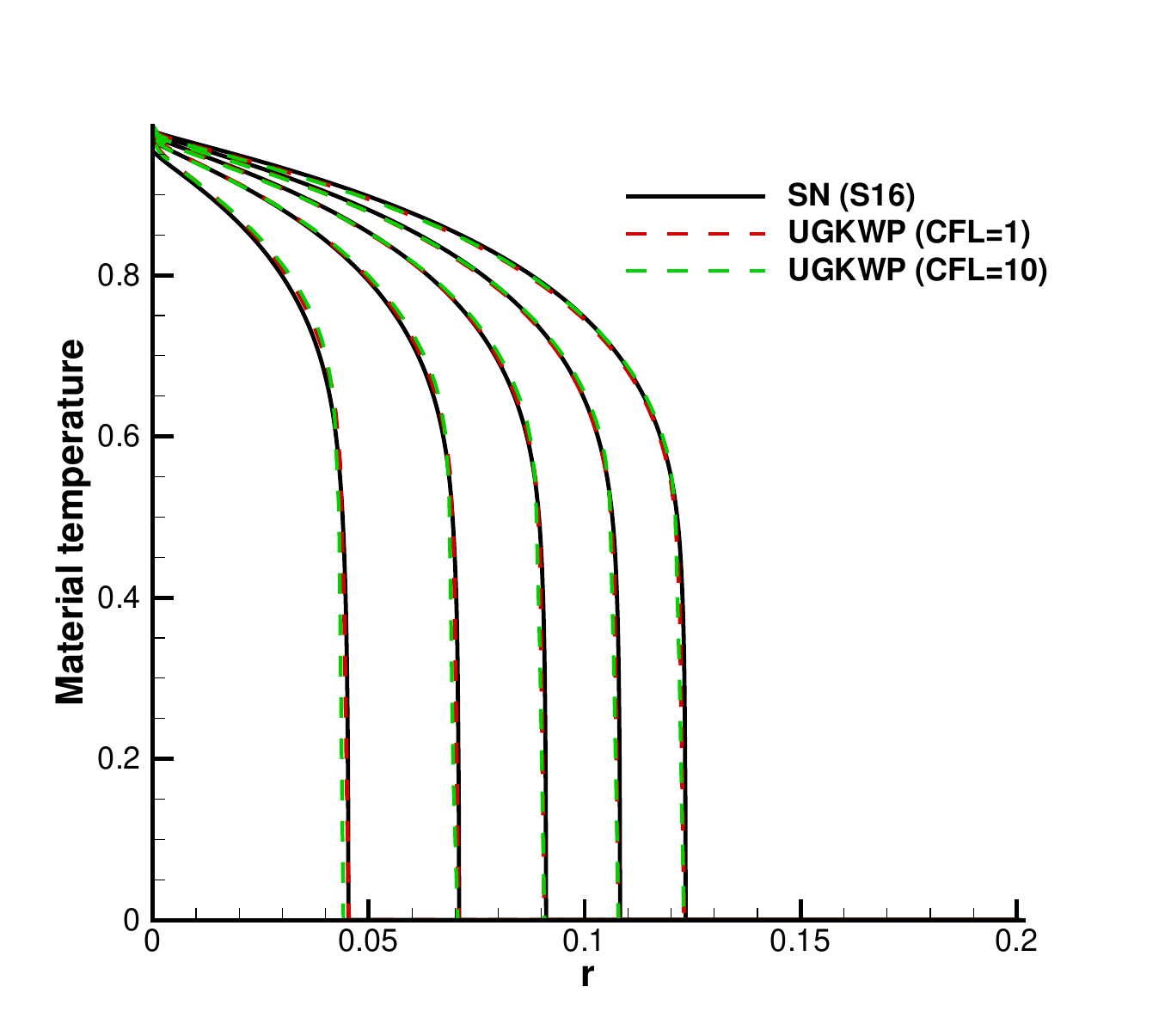}}
  \subfigure[Radiation temperature]{\includegraphics[width=0.49\textwidth]{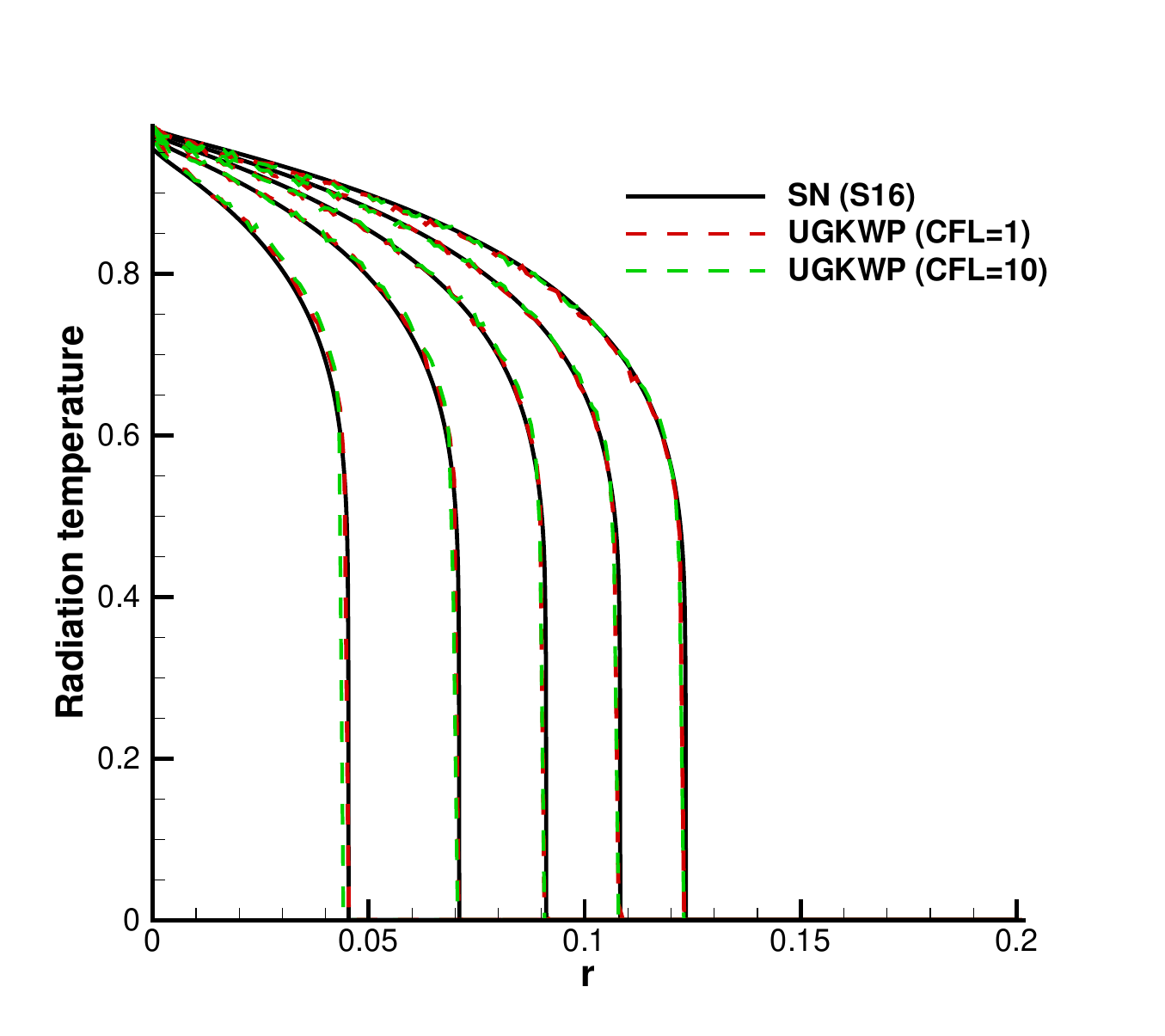}}
  \caption{Comparison of the material and radiation temperatures between IUGKWP and SN at
    $t=0.2, 0.4, 0.6, 0.8, 1.0$ for Marshak wave-2A problem.}
  \label{fig_marshaka1}
\end{figure}

\subsection{Marshak wave-2B problem}
The Marshak wave-2B shares the same geometry setup as Marshak wave-2A with a thick opacity.
The material opacity is $\sigma=300.0/T^{3}\text{cm}^{-1}$, and the heat capacity $C_v=0.3\text{GJ}/\text{KeV}/\text{cm}^{-3}$.
The initial temperature is in equilibrium with $T_r=T_e=10^{-6}\text{KeV}$, and the incident intensity on the left boundary is $T_r=1.0\text{KeV}$.
We perform IUGKWP, IMC, SN simulations on a 2D domain with $5\times10^{-3}\text{cm}$ in y-direction and $6\times10^{-1}\text{cm}$ in x-direction.
The physical domain is discretized into triangular meshes with cell size $2.5\times10^{-3}\text{cm}$.
The IUGKWP method is performed with two CFL numbers $1$ and $10$.
The reference particle energy is $5\times10^{-11}\text{GJ}$.

It can be observed in Fig. \ref{fig_marshakb1} that the IUGKWP solutions with CFL numbers 1 and 10 agree well with the SN solutions at times $t=15, 30, 45, 60, 74$.
For efficiency, to reach a simulation time of 100ns, the IMC takes 1090.8mins with CFL=1 and 796.2mins with CFL=10.
The IUGKWP method takes 253.6mins with $\text{CFL}=1$, and 48.1mins with $\text{CFL}=10$.
The IUGKWP is generally 3.3 times faster than IMC for $\text{CFL}=1$ and 15.6 times faster than IMC for $\text{CFL}=10$.
On the one hand, the IUGKWP only tracks the non-equilibrium particles, which takes about 10 percent IMC particles for Marshak wave-2B problem.
On the other hand,
 the IUGKWP avoids the massive calculation of IMC's effective scattering, which is computationally intensive for the Marshak wave-2B problem.

For the above two Marshak wave problems, the IUGKWP is accurate by comparing it to the $S_N$ solution. 
The time step of IUGKWP can be enlarged to $\text{CFL}=10$. 
In the optically thin regime, the IUGKWP is slightly fast compared to IMC, 
and in the optically thick regime, the IUGKWP is significantly faster than IMC.

\begin{figure}
  \centering
  \subfigure[Material temperature]{\includegraphics[width=0.49\textwidth]{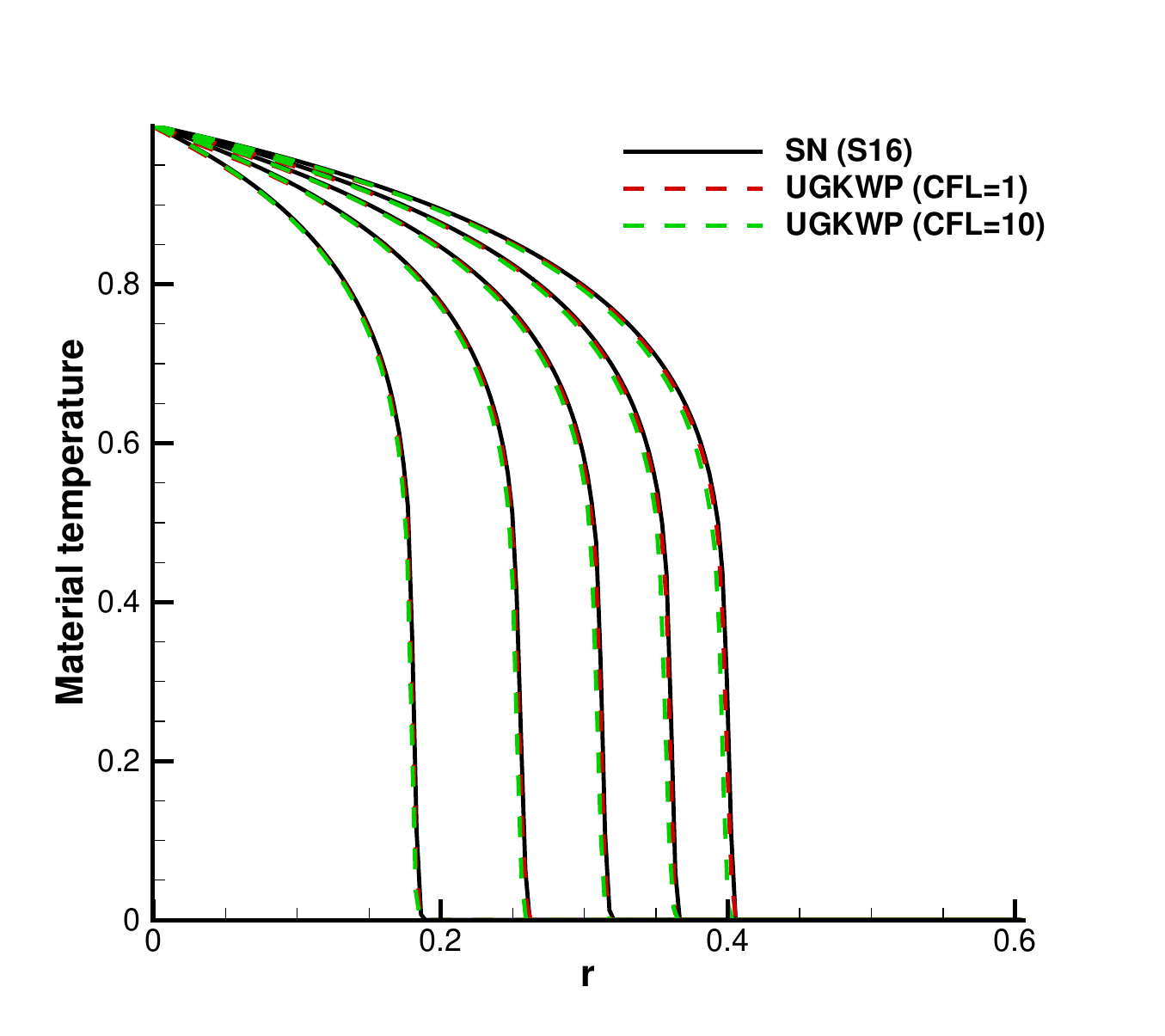}}
  \subfigure[Radiation temperature]{\includegraphics[width=0.49\textwidth]{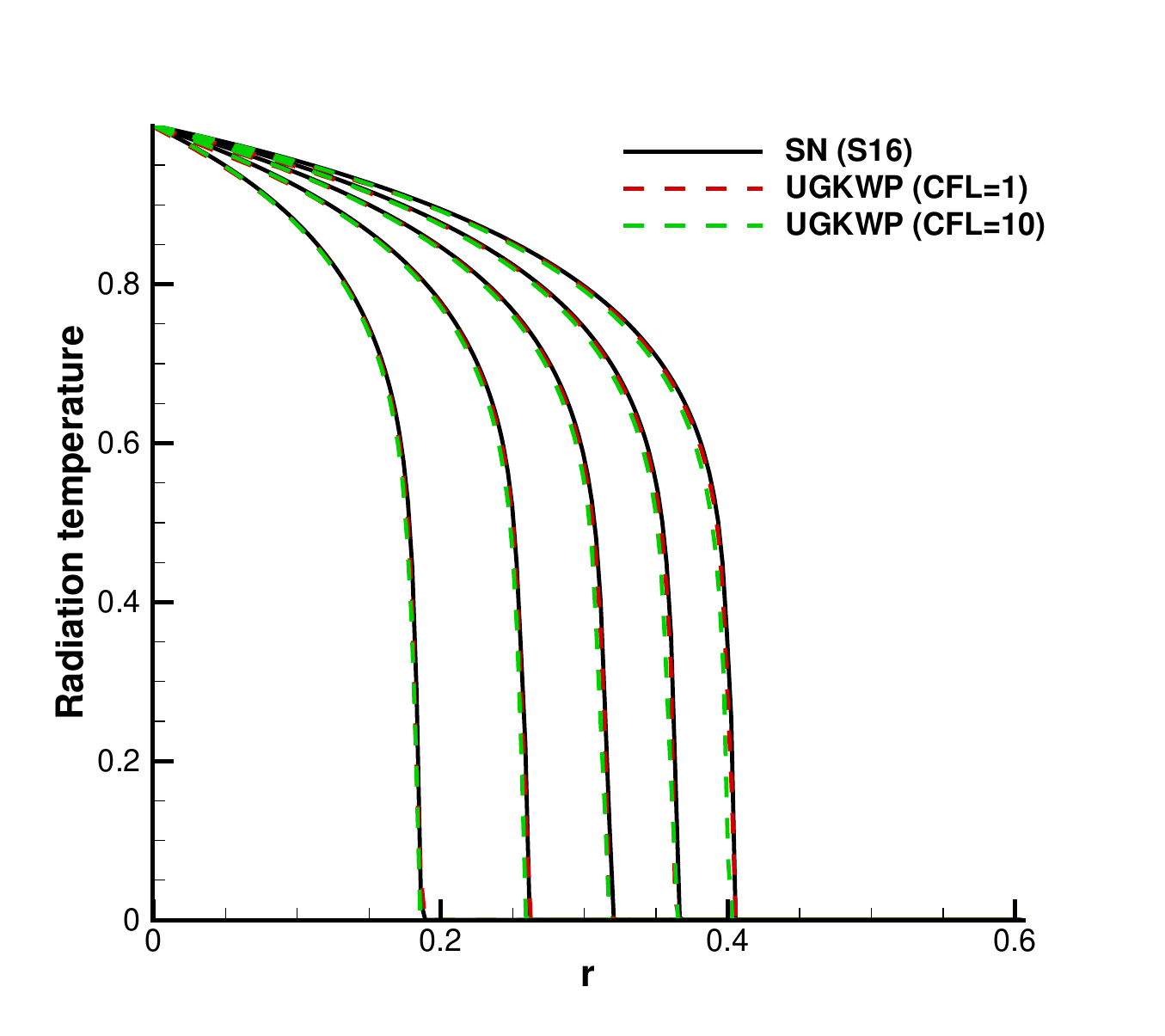}}
  \caption{Comparison of the material and radiation temperatures between IUGKWP and SN at
    $t=15, 30, 45, 60, 74$ for Marshak wave-2B problem.}
  \label{fig_marshakb1}
\end{figure}

\begin{figure}
  \centering
  \subfigure[Material temperature]{\includegraphics[width=0.49\textwidth]{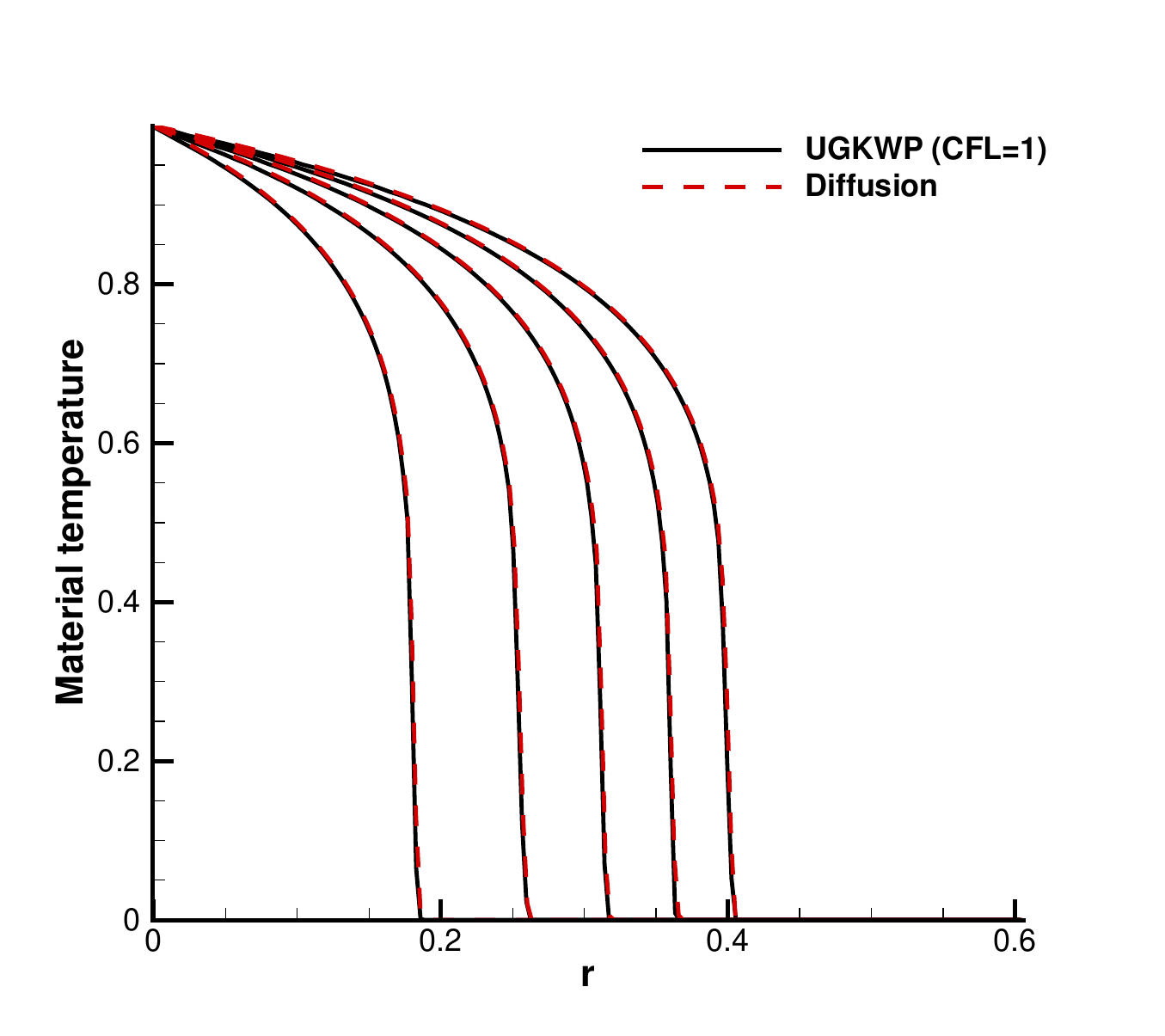}}
  \subfigure[Radiation temperature]{\includegraphics[width=0.49\textwidth]{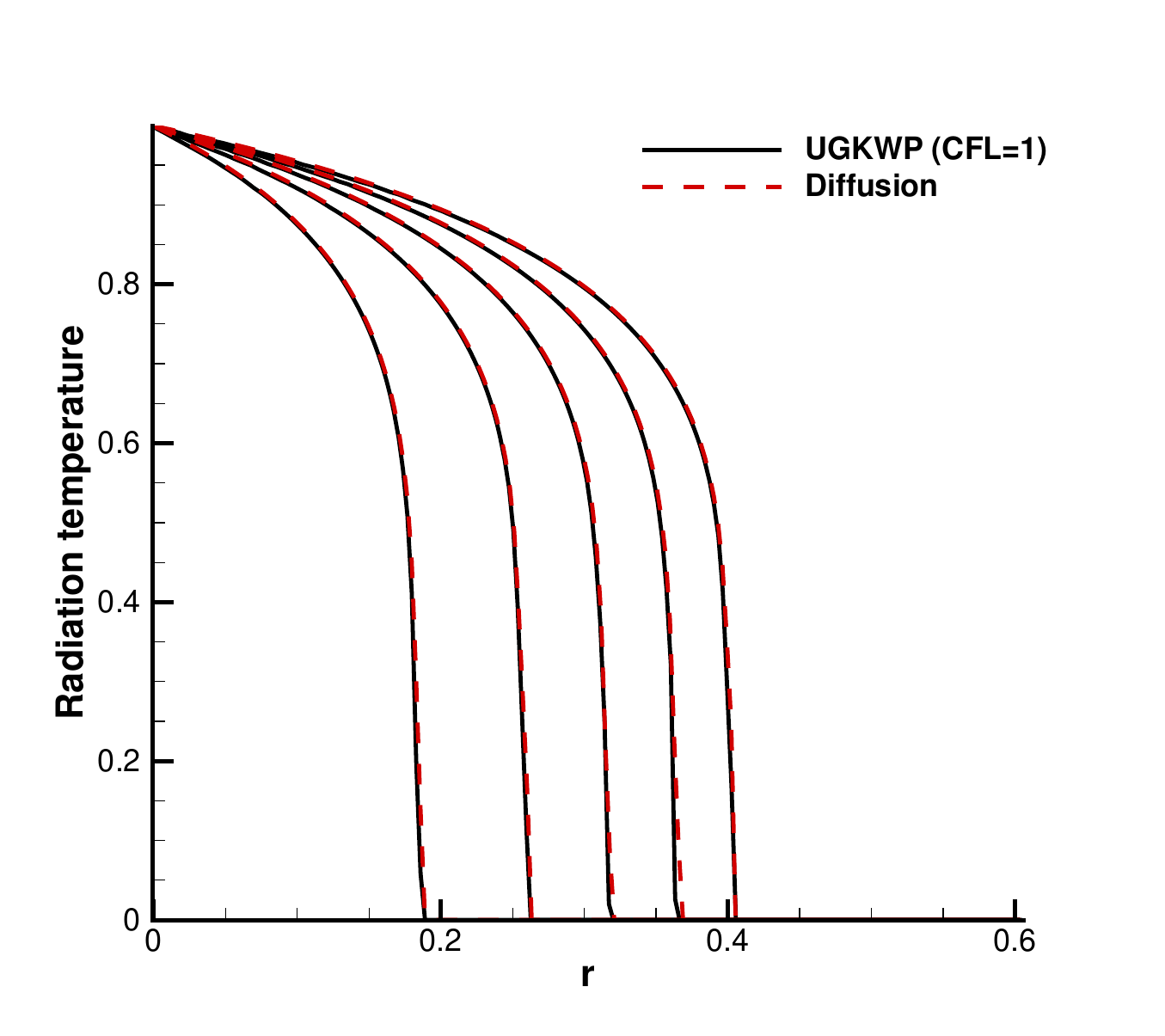}}
  \caption{Comparison of the material and radiation temperatures between IUGKWP and Diffusion solution at 
  $t=15, 30, 45, 60, 74$ for Marshak wave-2B problem.}
  \label{fig_marshakb2}
\end{figure}

\begin{figure}
  \centering
  \subfigure[Material temperature]{\includegraphics[width=0.49\textwidth]{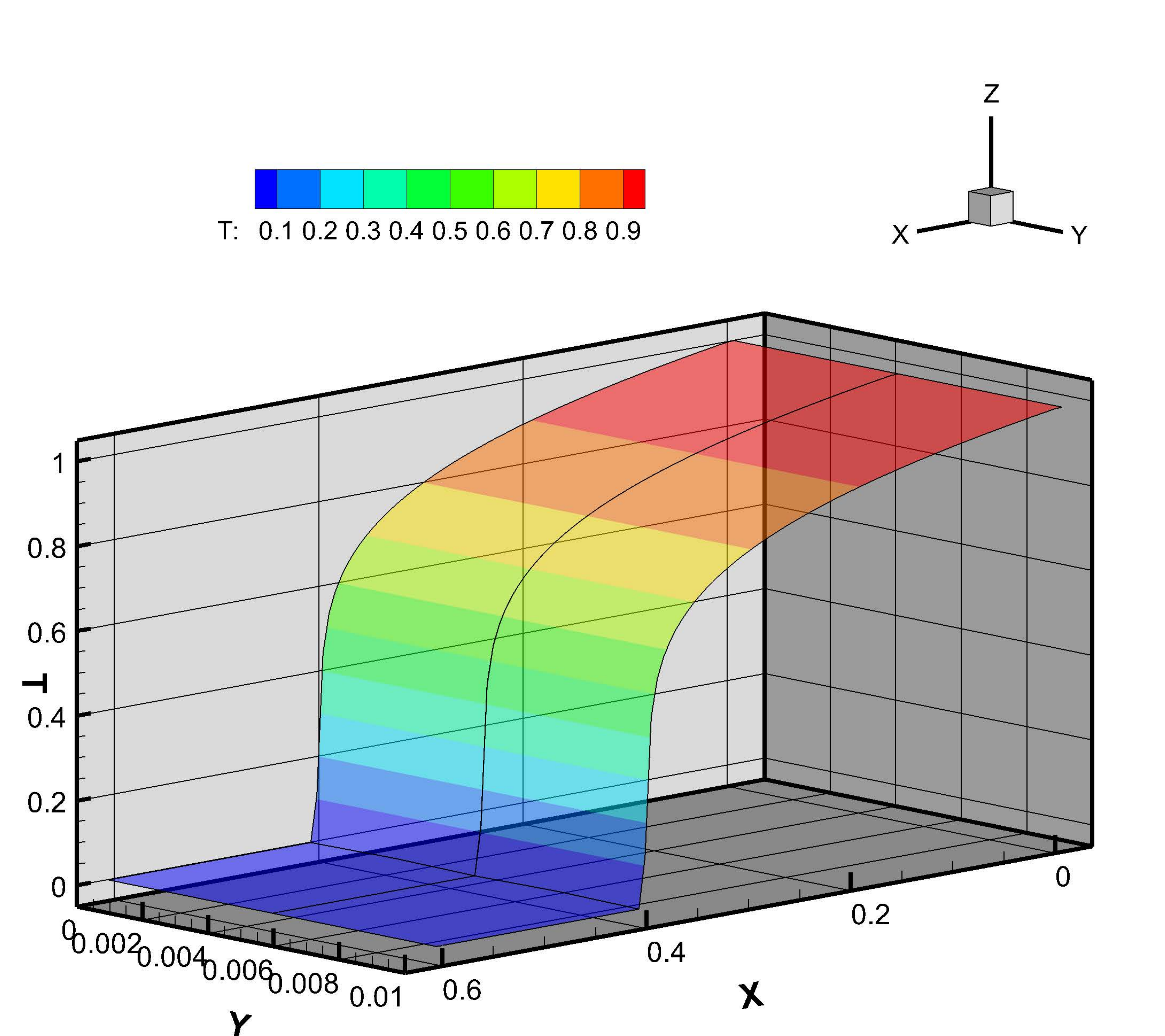}}
  \subfigure[Radiation temperature]{\includegraphics[width=0.49\textwidth]{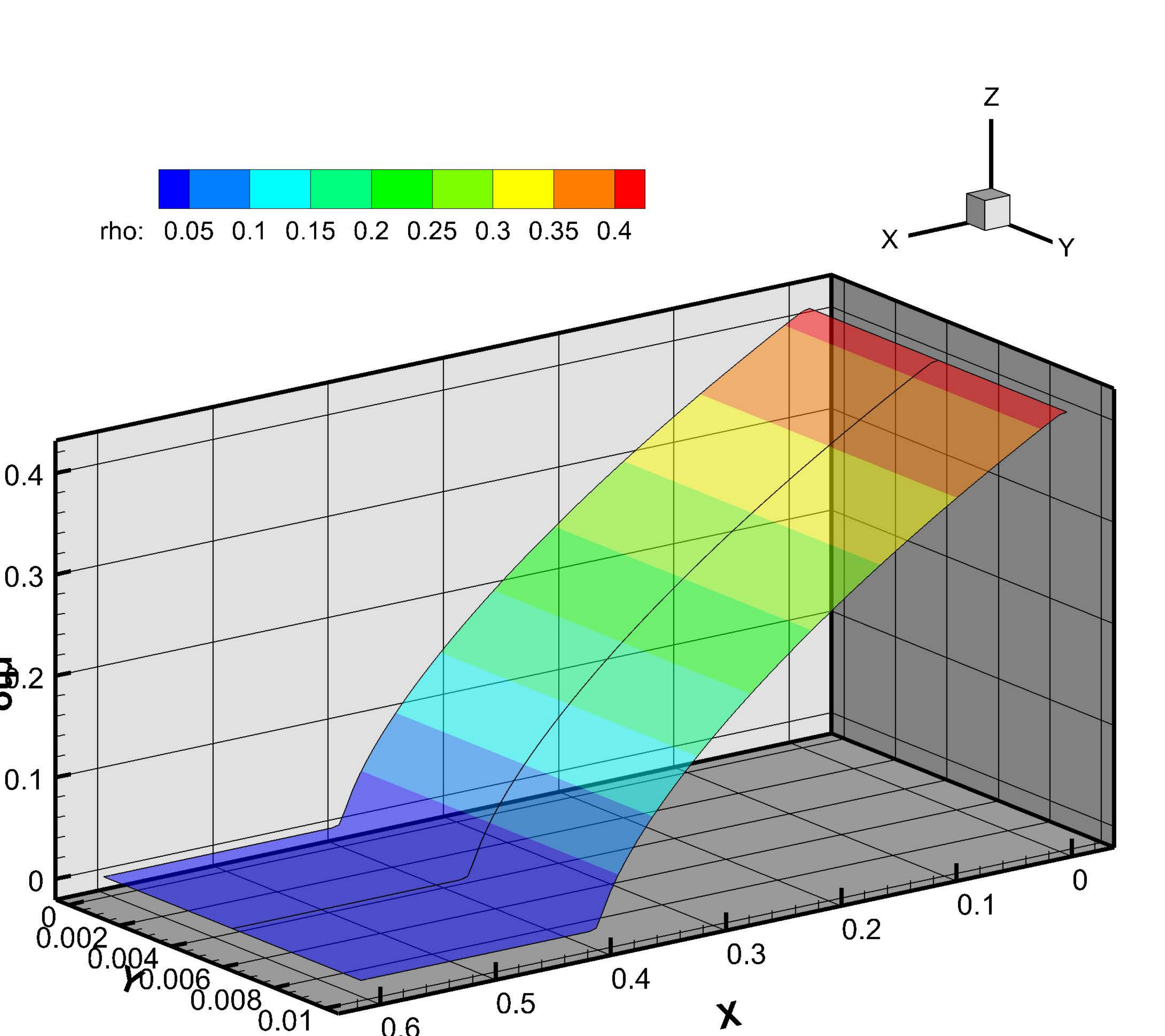}}
  \caption{Comparison of the material temperature and radiation density contour between IUGKWP $(y<0.4)$ 
  and diffusion solution $(y>0.4)$ at $t=74$.}
  \label{fig_marshakb3}
\end{figure}

\subsection{Tophat-A problem}
The tophat problem is also known as the crooked pipe problem,
which describes radiation wave propagation inside a two-dimensional domain consisting of optically thick and optically thin regions.
The computational domain is a $[0,7\text{cm}]\times[0,2\text{cm}]$ square in the x and y directions.
The optically thick regions with $\sigma=2\times10^{3}\text{cm}^{-1}$ are located in the regions
$[3.0, 4.0]\times[0, 1.0]$, $[0, 2.5]\times[0.5, 2.0]$, $[4.5, 7.0]\times[0.5, 2.0]$, and $[2.5, 4.5]\times[1.5, 2.0]$.
The optically thin material with $\sigma=2.0\times10^{-1}\text{cm}^{-1}$ occupies the rest regions.
The heat capacities of optically thin and optically thick regions are $0.001 \text{GJ/KeV/cm}^3$ and $1.0 \text{GJ/KeV/cm}^3$.
Five probes are placed at $(0.25, 0)$, $(2.75, 0)$, $(3.5, 1.25)$, $(4.25, 0)$, and $(6.75, 0)$ to monitor the change of the temperature in the thin opacity material.
The system is initially in equilibrium at the temperature of $0.05\text{KeV}$.
A $0.5\text{KeV}$ isotropic surface source is applied on the left boundary for $0<r<0.5\text{cm}$.
The computational domain is discretized into triangular mesh with mesh size $5\times10^{-2}\text{cm}$.
The time step is set as $\text{CFL}=0.5+0.003t$, which increases from $0.5$ to $30.5$ for a simulation time of $1000\text{ns}$.
The reference particle energy is $5\times10^{-10}\text{GJ}$.

We compare the time evolution of radiation energy at five probes with reference SN solution as shown in Fig. \ref{fig_tophata1}.
The radiation wave propagates fast for the first several nanoseconds at approximately the speed of light and slows down as it reaches the crooked region.
Therefore, the CFL number is set to $0.5$ initially and gradually increases to $30.5$.
The dynamically increasing time step of IUGKWP captures the time evolution of radiation energy accurately and effectively.
The tophat problem is well-designed to examine the numerical dissipation of the scheme.
If a scheme is over-dissipated, the radiation energy will teleport excessively into the optically thick material
and slows down the propagation speed of the radiation wave in the optically thin region.
We compare the radiation energy profile and distribution to the reference SN solution at $t=500\text{ns}$ as shown in Fig. \ref{fig_tophata2}-\ref{fig_tophata3},
which shows the accurate propagation speed and proper control of dissipation in IUGKWP.

For efficiency, to finish $1000\text{ns}$ simulation, the IUGKWP takes 14112mins, and the IMC method takes 23990mins under the same mesh and time step.
The time saving is due to less number of particles and a more effective collision algorithm.

\begin{figure}
  \centering
  \includegraphics[width=0.6\textwidth]{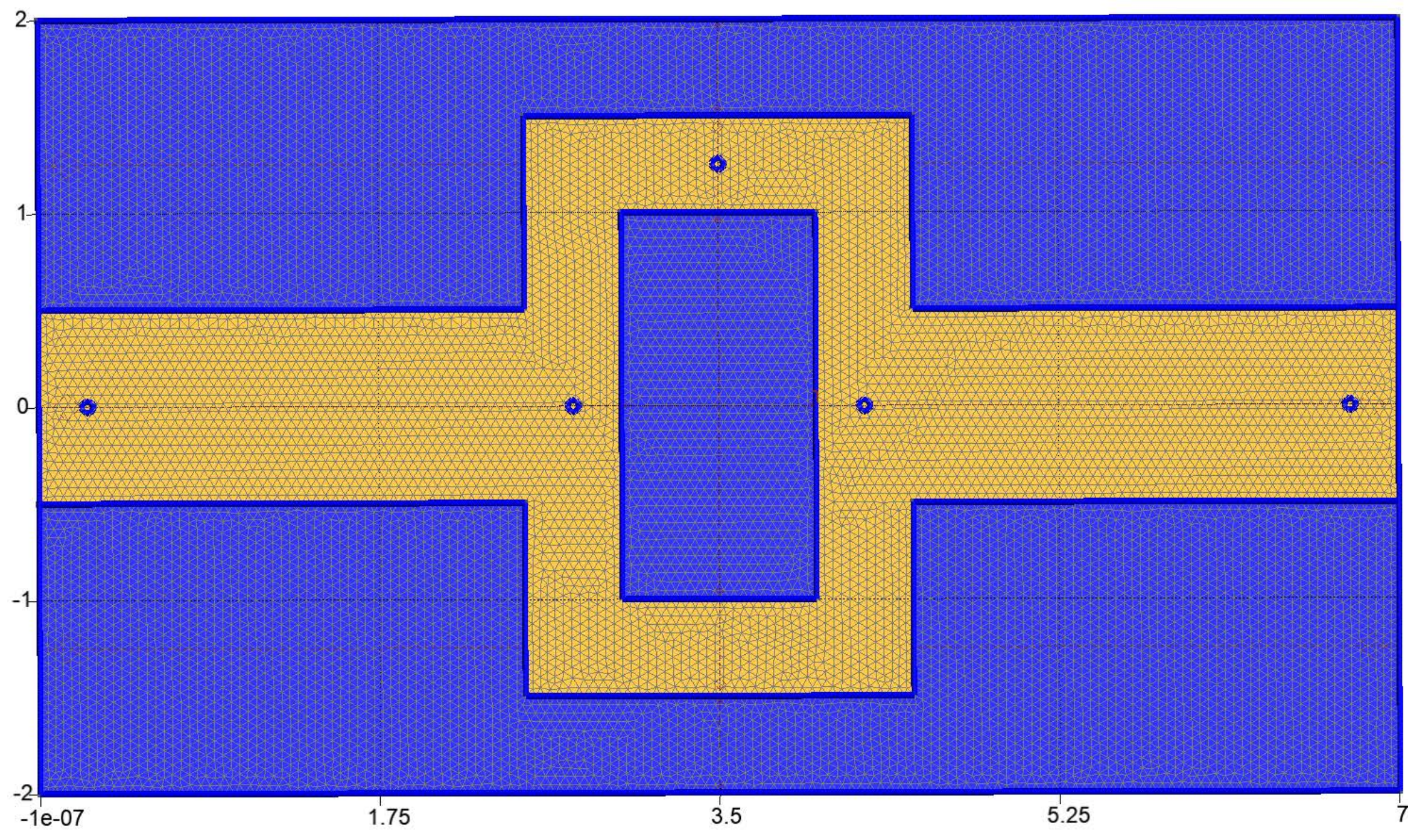}
  \caption{The geometry and mesh of the tophat problem.}
  \label{fig_tophata0}
\end{figure}

\begin{figure}
  \centering
    \begin{minipage}{0.49\textwidth}
        \centering 
        \includegraphics[width=1.0\textwidth]{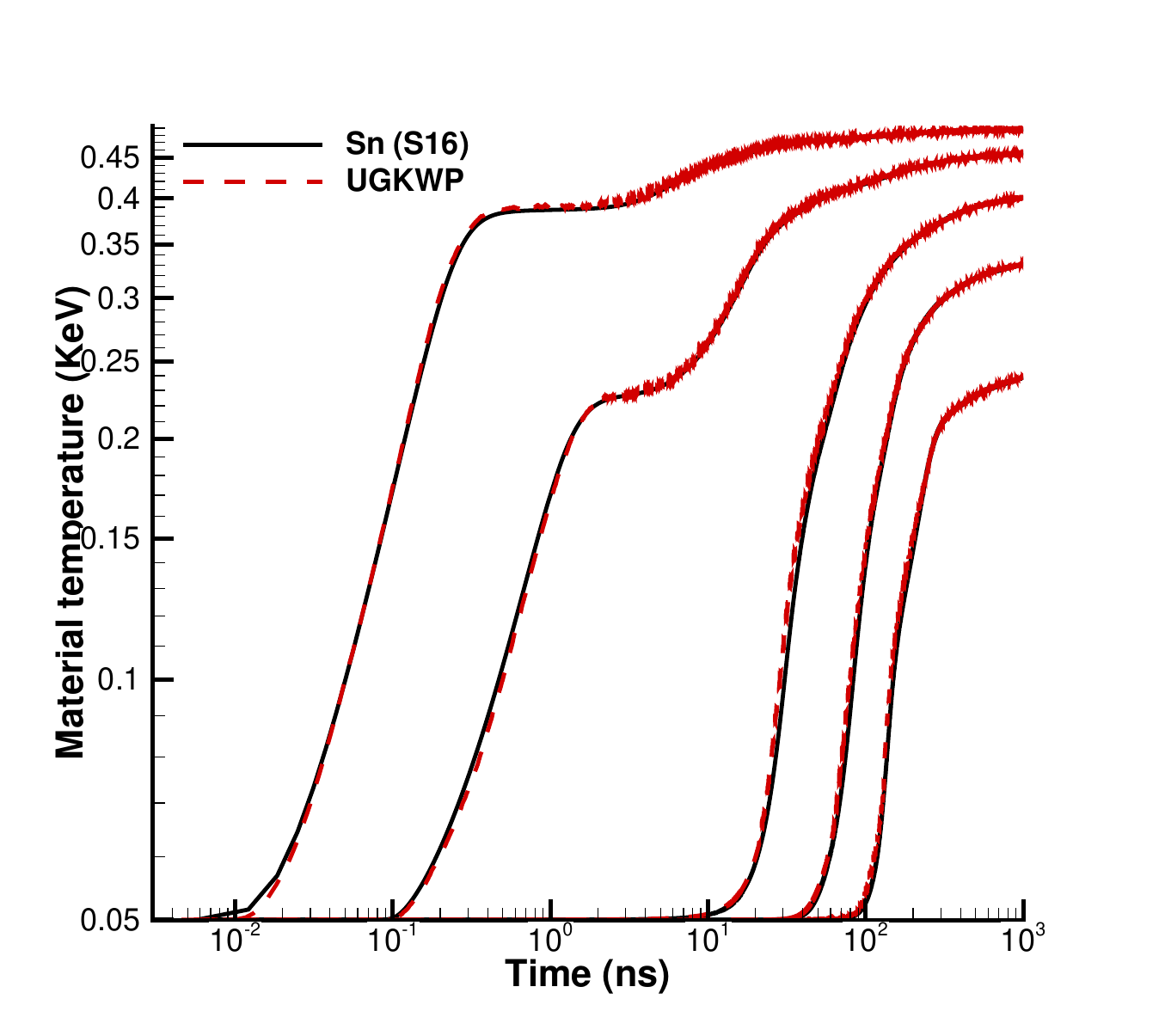}
        \caption{Time evolution of $T_e$ at five probes.}
        \label{fig_tophata1}
    \end{minipage}
    \begin{minipage}{0.49\textwidth}
        \centering
        \includegraphics[width=1.0\textwidth]{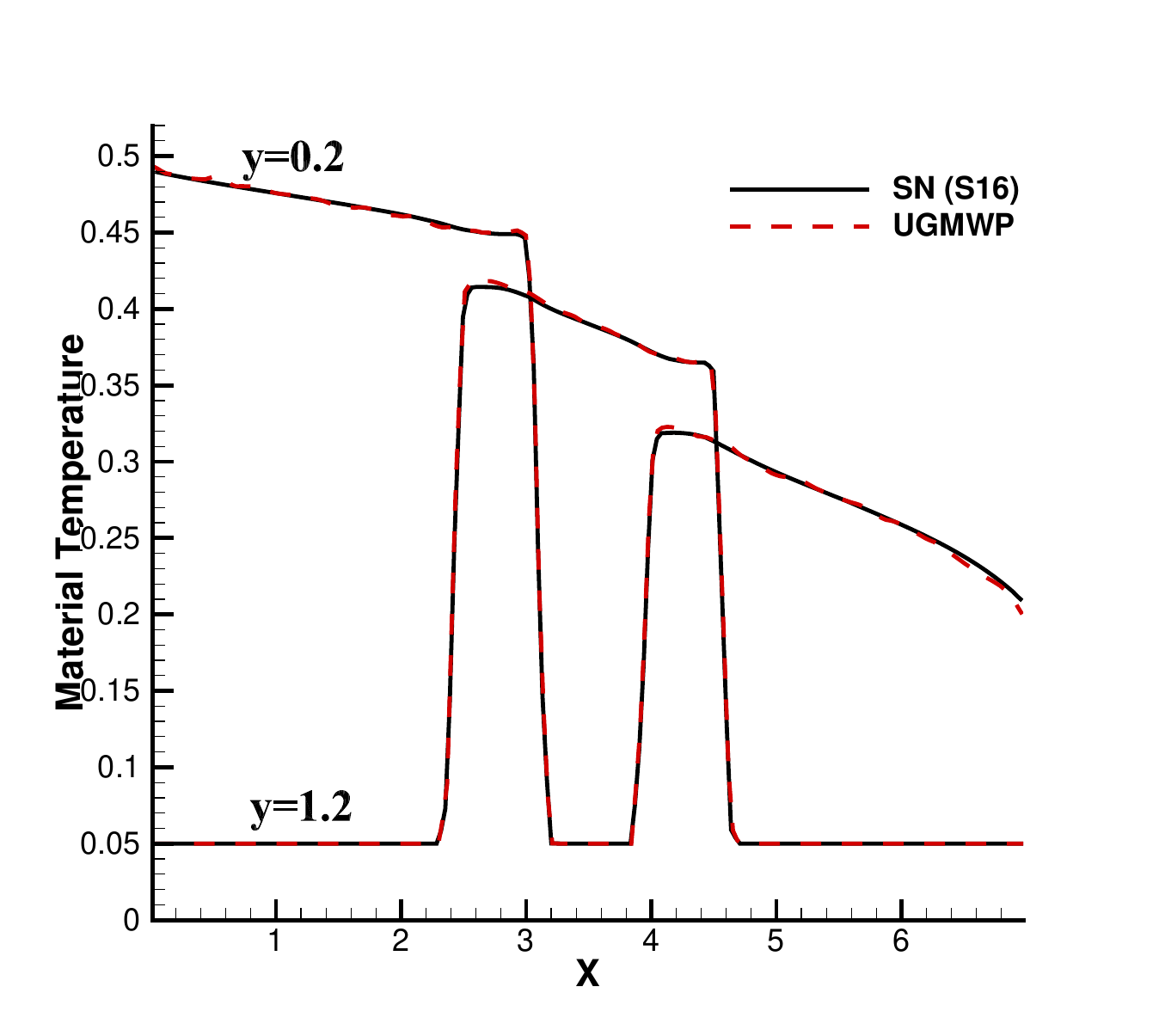}
        \caption{Comparison of $T_e$ between IUGKWP and SN along $y=0.2,1.2$ at time $t=500\text{ns}$.}
        \label{fig_tophata2}
    \end{minipage}
\end{figure}

\begin{figure}
  \centering
  \subfigure[Material temperature]{\includegraphics[width=0.49\textwidth]{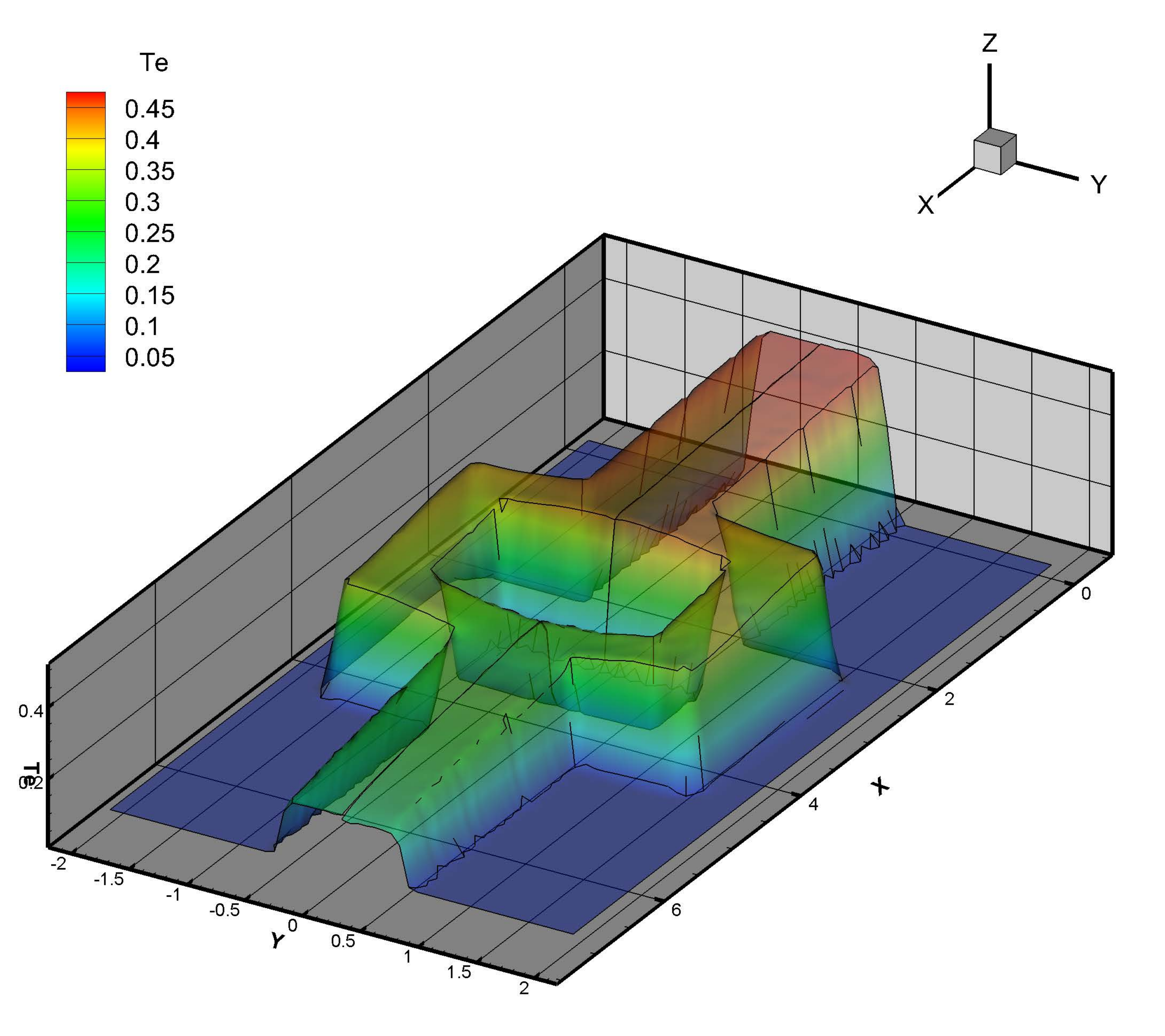}}
  \subfigure[Radiation temperature]{\includegraphics[width=0.49\textwidth]{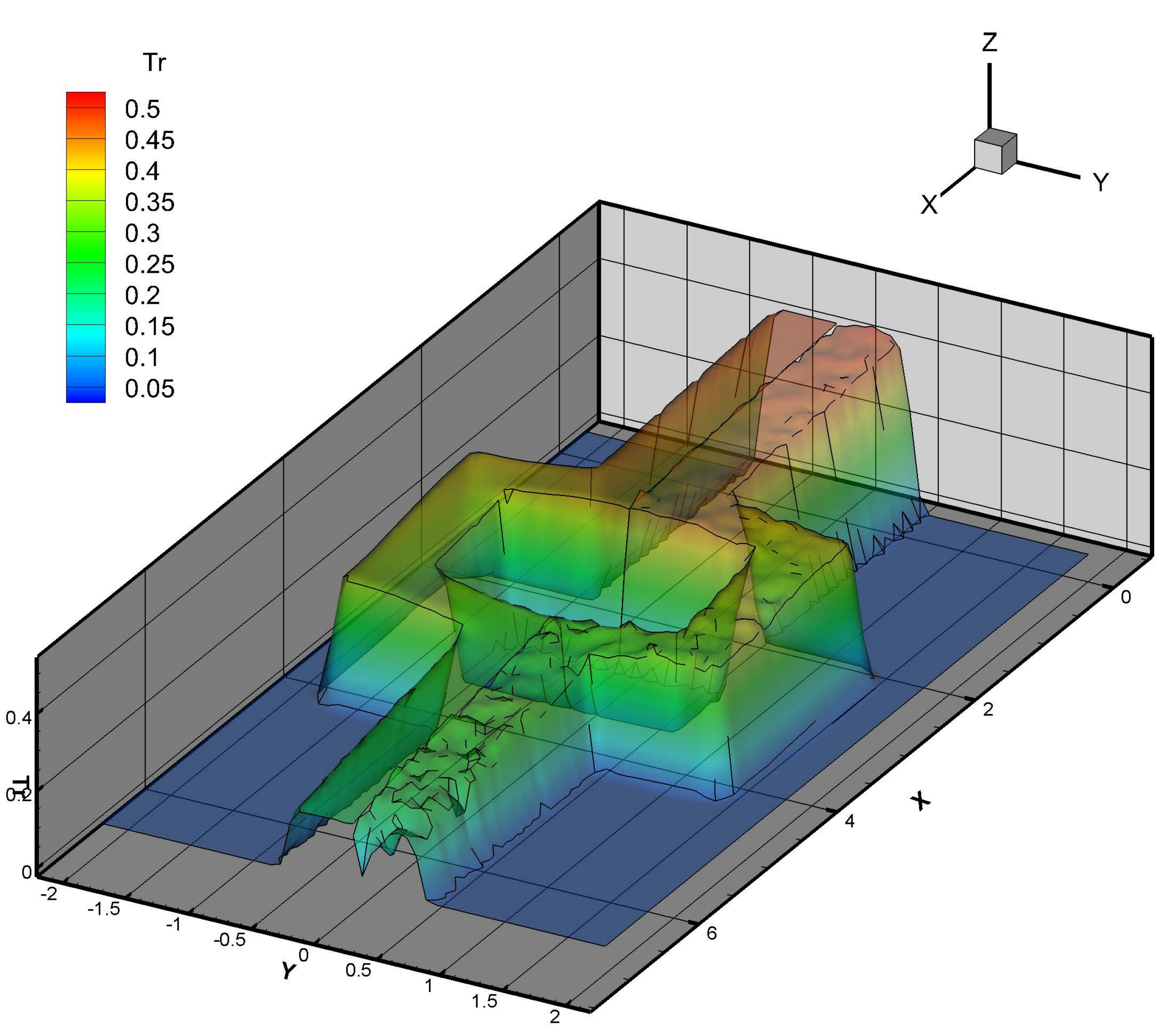}}
  \caption{Comparison of the material and radiation temperature contour between IUGKWP $(y>0)$ and SN $(y<0)$ at $t=500\text{ns}$.}
  \label{fig_tophata3}
\end{figure}

\subsection{Tophat-B problem}
The tophat-B problem shares the same geometry, initial condition, and boundary condition as the tophat-A problem.
The opacity of optically thin material is increased to $\sigma=1.0\times10^{2}\text{cm}^{-1}$ and the opacity of optically thick material is increased to $\sigma=2\times10^{4}\text{cm}^{-1}$.
The heat capacities of optically thin and optically thick regions are kept the same as tophat-A, i.e., $0.001\text{GJ/KeV/cm}^3$ and $1.0\text{GJ/KeV/cm}^3$ respectively.
We place five probes $(0.25, 0)$, $(2.75, 0)$, $(3.5, 1.25)$, $(4.25, 0)$, and $(6.75, 0)$ to monitor the change of the temperature in the thin opacity material.
The computational domain is discretized into triangular mesh with mesh size $5\times10^{-2}\text{cm}$.
The time step is set as $\text{CFL}=0.5+0.003t$, which increases from $0.5$ to $30.5$ for a simulation time of $1000\text{ns}$.
The reference particle energy is $5\times10^{-10}\text{GJ}$.

We compare the time evolution of radiation energy at five probes with reference SN solution as shown in Fig. \ref{fig_tophatb1}-\ref{fig_tophatb3},
which show good agreement in both the time evolution and spatial distribution of the radiation energy.
We simulate a time period of $1\times10^{3}\text{ns}$, which takes the IUGKWP only 1383.2mins.
For the IMC method, it is too expensive to finish the $1000\text{ns}$ simulation. It takes IMC 16943.6mins to simulate $10\text{ns}$.
It can be roughly estimated that the IUGKWP is more than 100 times faster than IMC for the tophat-B problem.

\begin{figure}
  \centering
    \begin{minipage}{0.49\textwidth}
        \centering 
        \includegraphics[width=1.0\textwidth]{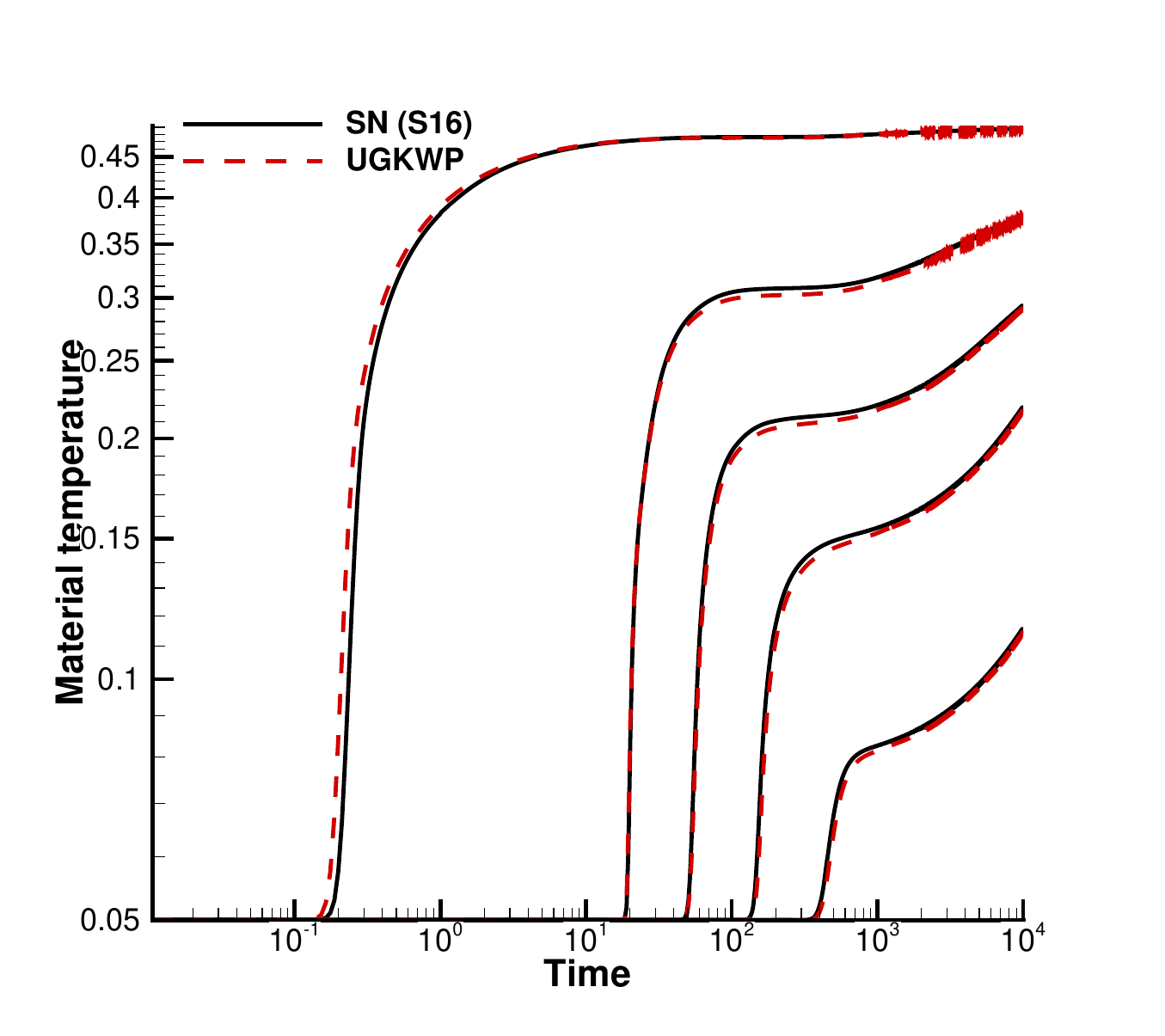}
        \caption{Time evolution of $T_e$ at five probes.}
        \label{fig_tophatb1}
    \end{minipage}
    \begin{minipage}{0.49\textwidth}
        \centering
        \includegraphics[width=1.0\textwidth]{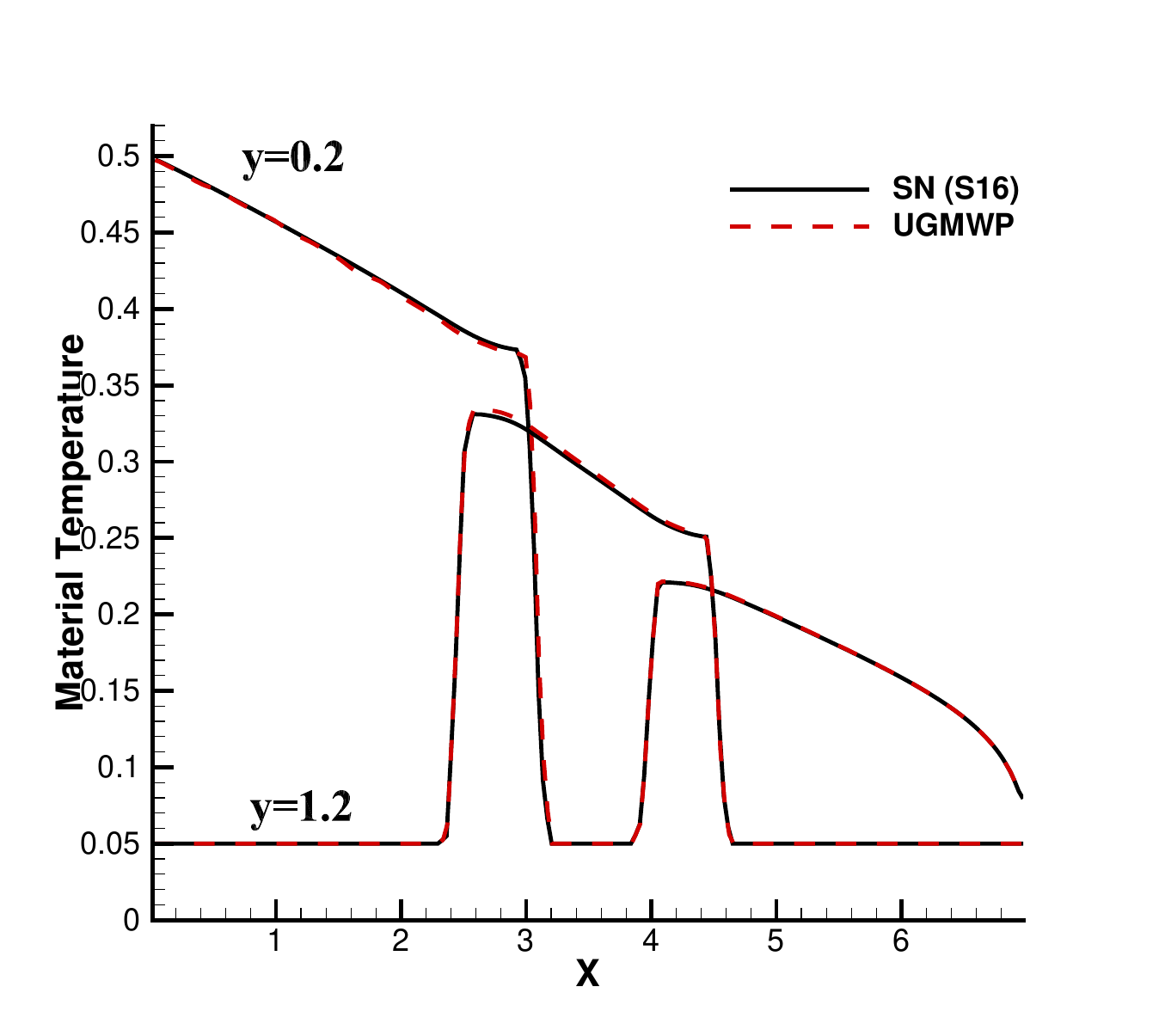}
        \caption{Comparison of $T_e$ between IUGKWP and SN along $y=0.2,1.2$.}
        \label{fig_tophatb2}
    \end{minipage}
\end{figure}

\begin{figure}
  \centering
  \subfigure[Material temperature]{\includegraphics[width=0.49\textwidth]{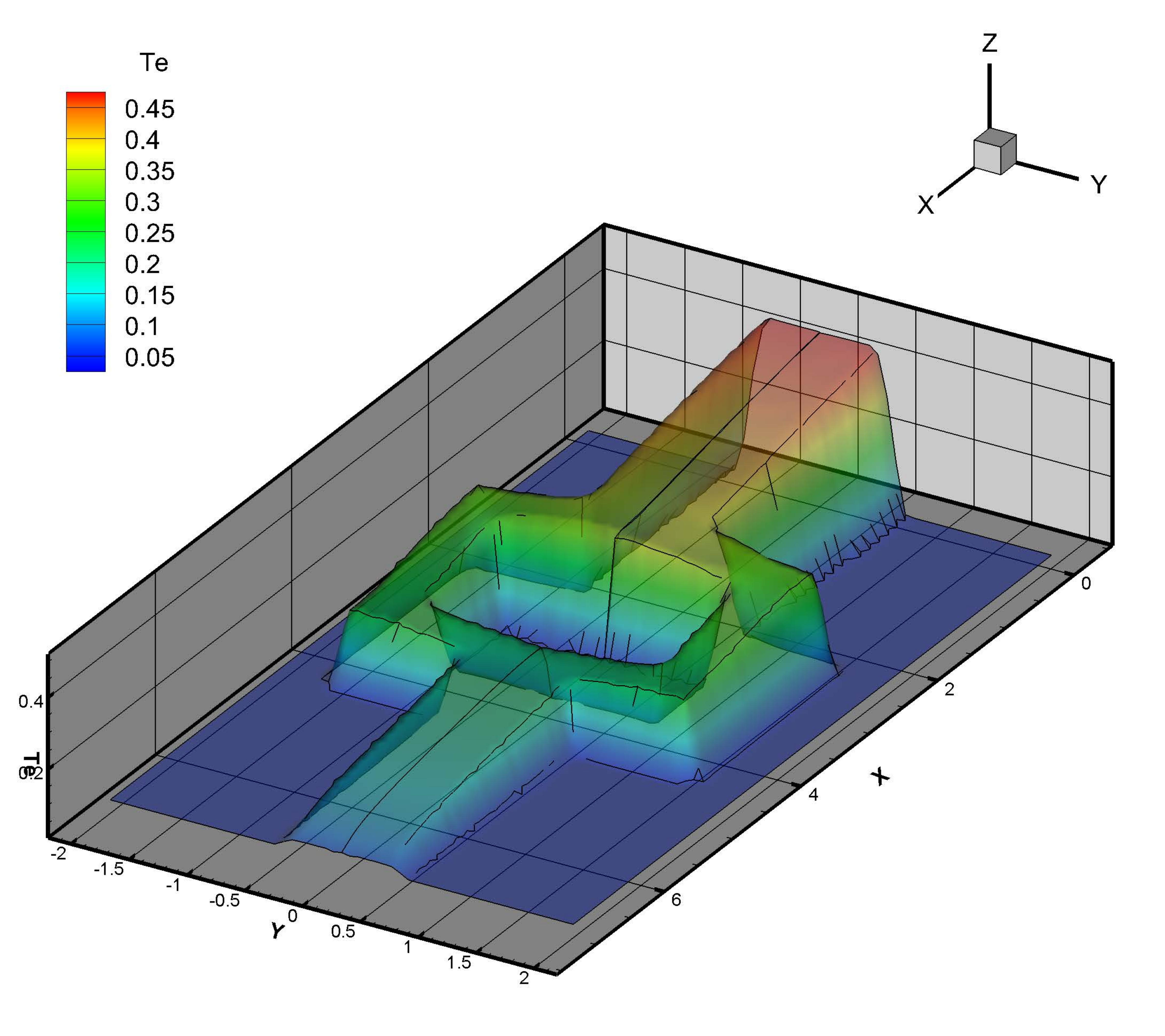}}
  \subfigure[Radiation temperature]{\includegraphics[width=0.49\textwidth]{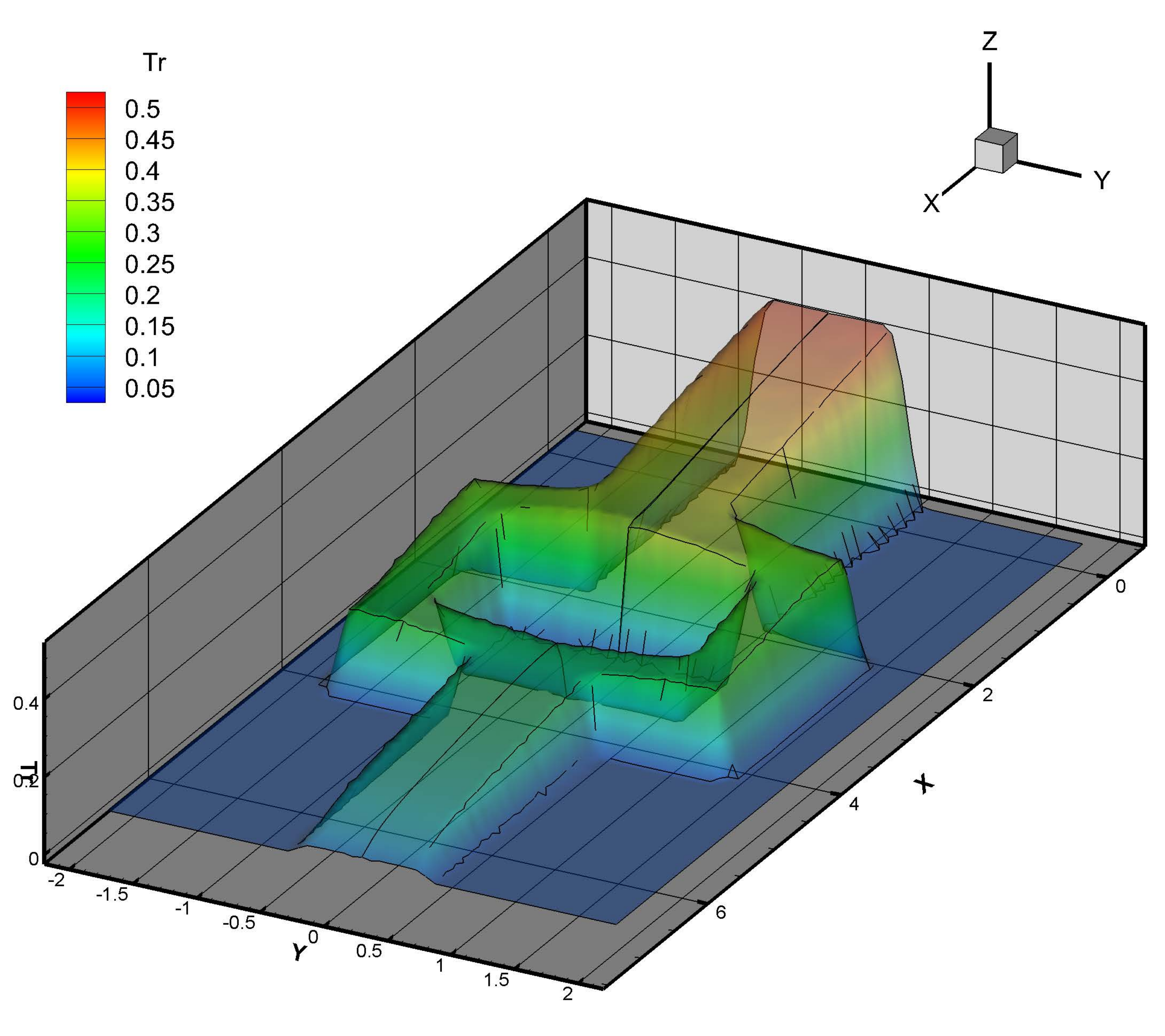}}
  \caption{Comparison of the material and radiation temperature contour between IUGKWP $(y>0)$ and SN $(y<0)$ at $t=500\text{ns}$.}
  \label{fig_tophatb3}
\end{figure}

\subsection{Square hohlraum problem}
The study of hohlraum is one of the key topics in ICF.
The geometry of the square hohlraum is shown in Fig. \ref{fig_hohlraum2dgeo}.
The hohlraum boundary and capsule are filled with optically thick material with $\sigma=2.0\times10^{3}$ and capacity $C_v=1$.
The hohlraum cavity is filled with optically thin material with $\sigma=2.0\times10^{-1}$ and capacity $C_v=1\times10^{-2}$.
The system is initially in equilibrium at the temperature of $0.05\text{KeV}$.
A $0.5\text{KeV}$ isotropic surface source is applied on the left boundary.
The computational domain is discretized into triangular mesh with diameter $\Delta x=0.01$,
and the CFL number is set to 30.
The material and radiation temperature evolution predicted by IUGKWP agrees well with SN results.
As shown in Fig. \ref{fig_hohlraum2da1},
the IUGKWP temperature contour ($y>0$) and SN temperature contour ($y<0$) agree well at $t=1\text{ns}$.
We also compare the material temperature on the capsule surface in Fig. \ref{fig_hohlraum2da2}
and the temperature profile along $x=0$ and $y=0$ in Fig. \ref{fig_hohlraum2da3}.

\begin{figure}
  \centering
  \subfigure[Material temperature]{\includegraphics[width=0.43\textwidth]{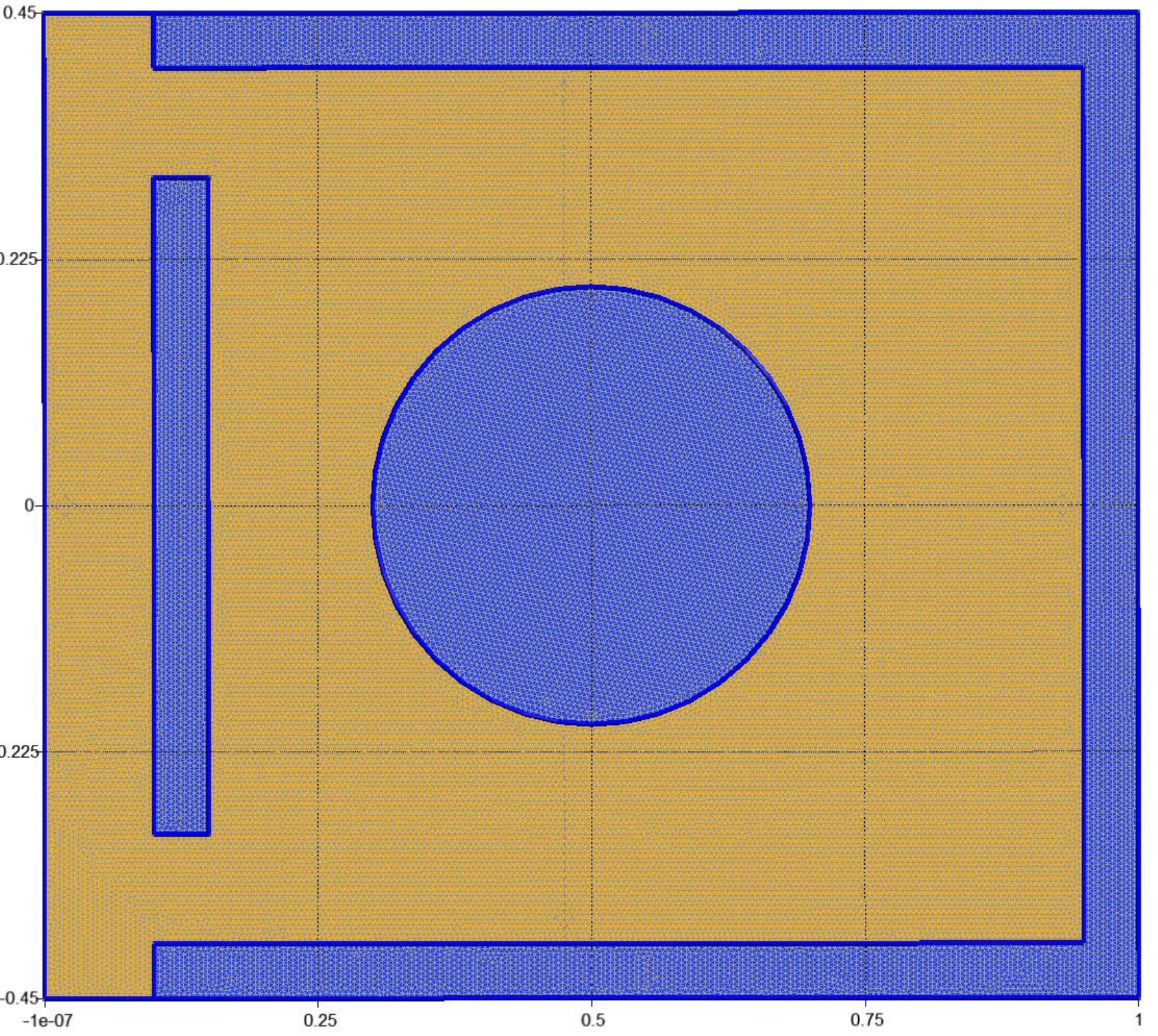}}
  \hspace{15mm}
  \subfigure[Radiation temperature]{\includegraphics[width=0.387\textwidth]{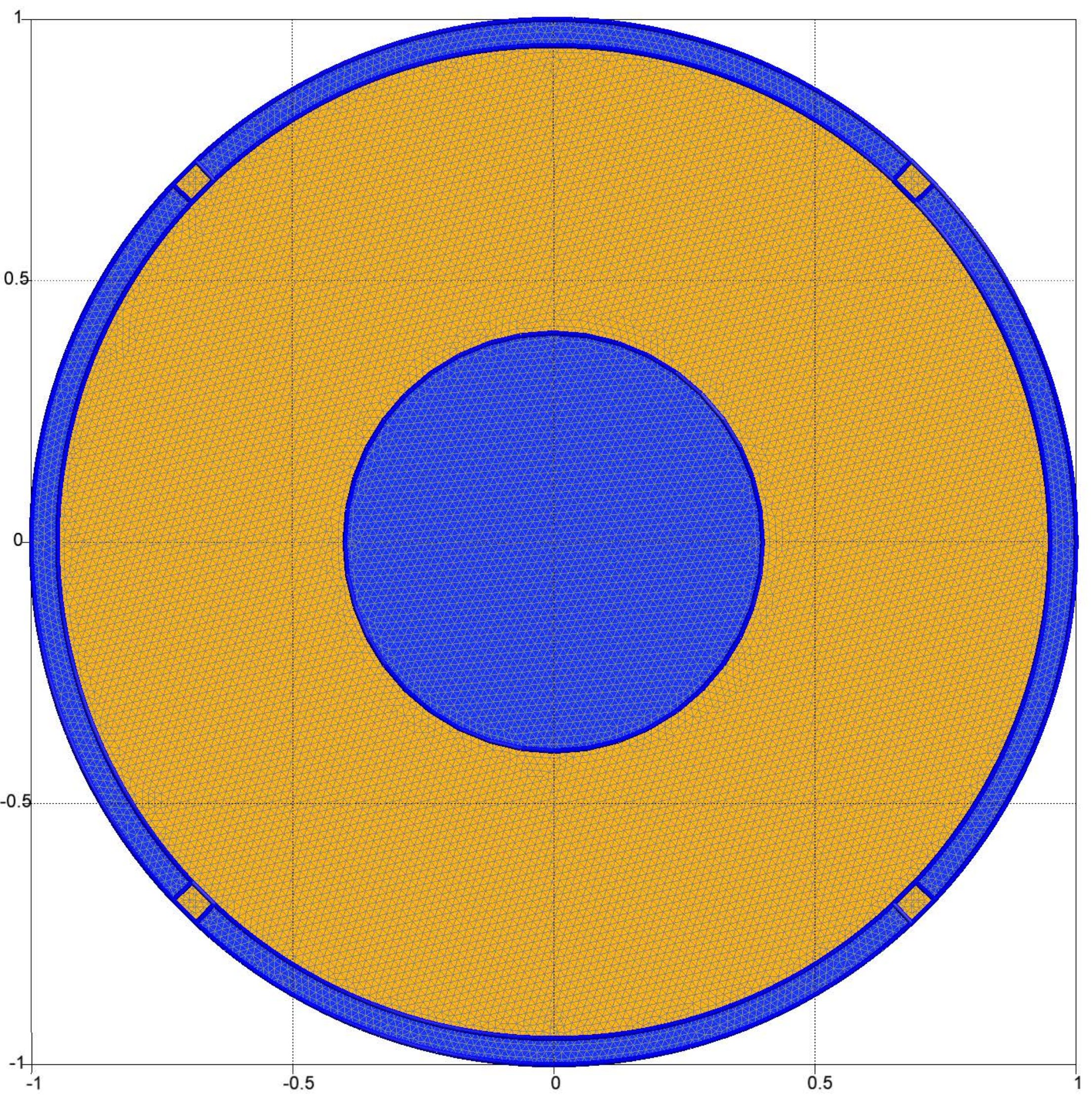}}
  \caption{The geometry and mesh of 2D square and circular hohlraum problem.}
  \label{fig_hohlraum2dgeo}
\end{figure}

\begin{figure}
  \centering
  \subfigure[Material temperature]{\includegraphics[width=0.49\textwidth]{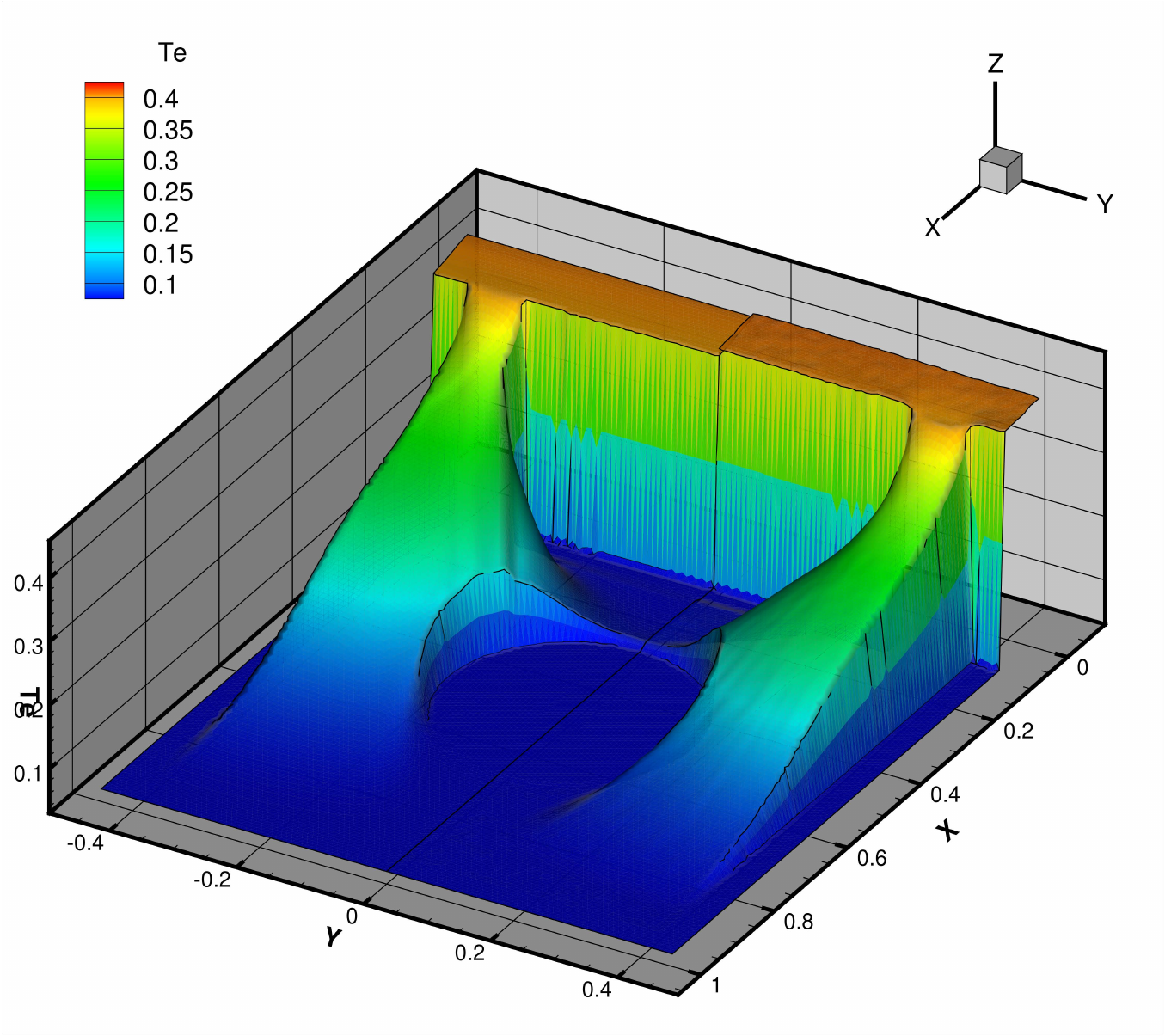}}
  \subfigure[Radiation temperature]{\includegraphics[width=0.49\textwidth]{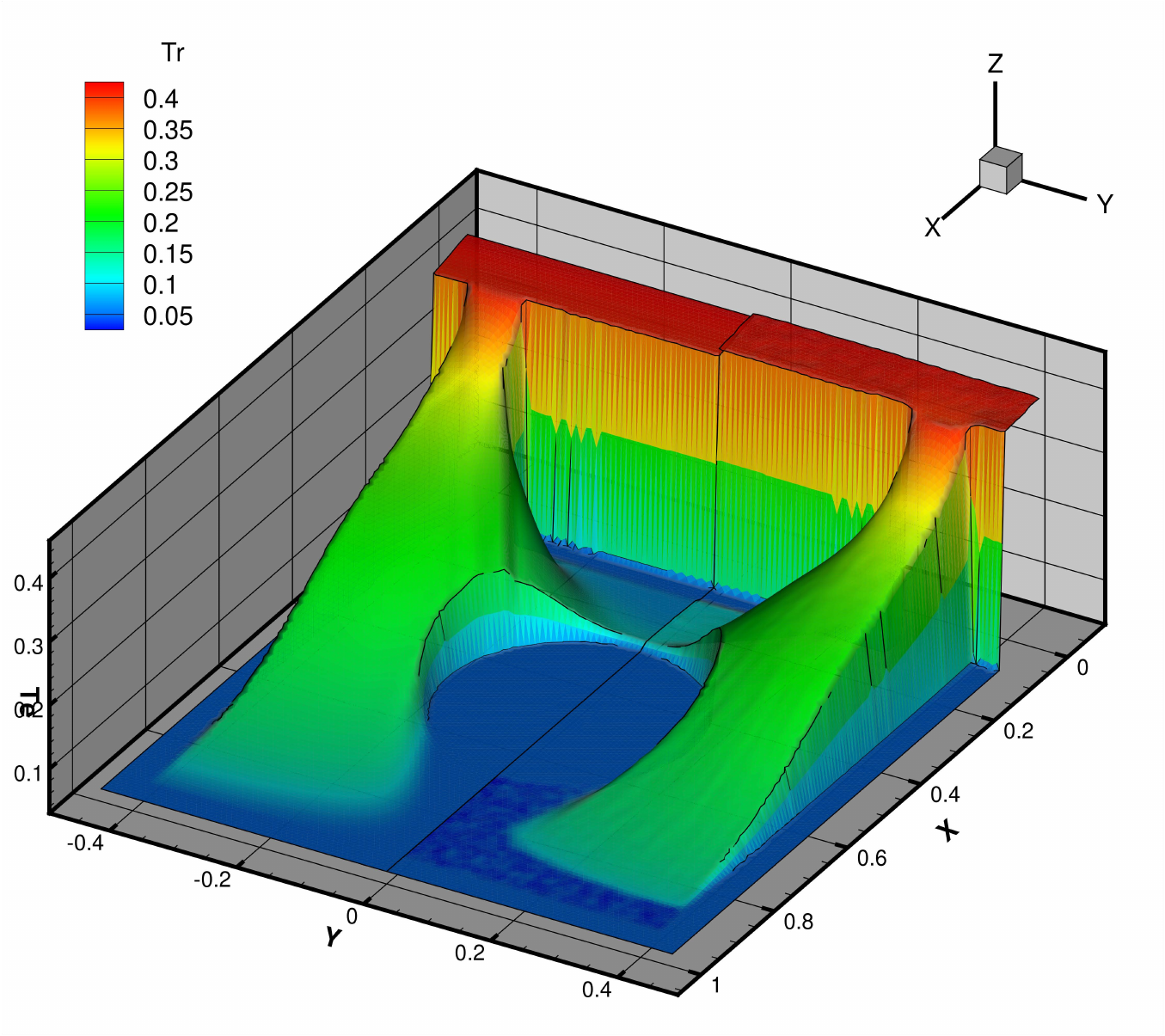}}
  \caption{Comparison of the material and radiation temperature contour between IUGKWP $(y>0)$ and SN $(y<0)$ at $t=1\text{ns}$.}
  \label{fig_hohlraum2da1}
\end{figure}

\begin{figure}
  \centering
    \begin{minipage}{0.49\textwidth}
        \centering 
        \includegraphics[width=1.0\textwidth]{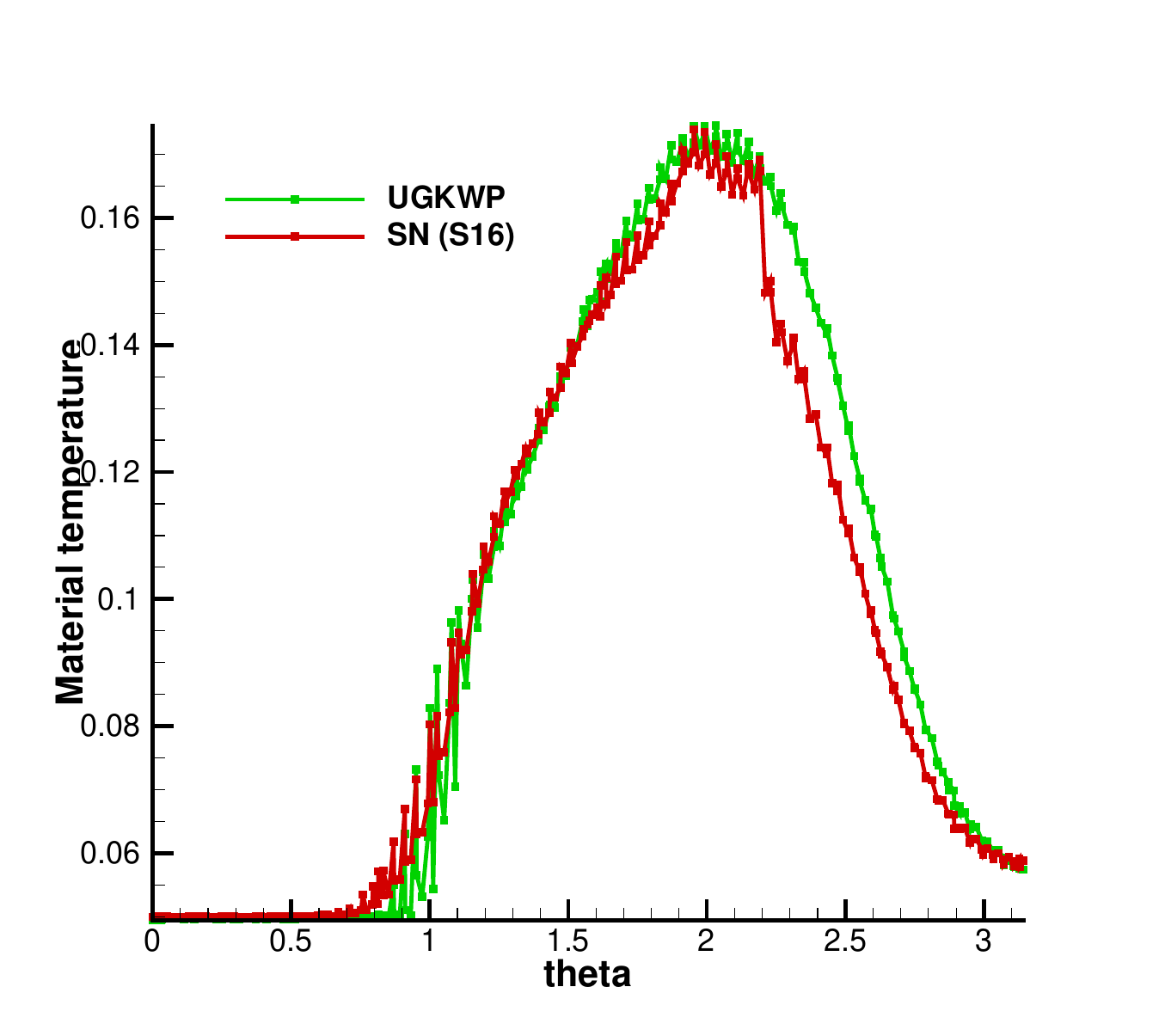}
        \caption{Material temperature on capsule surface.}
        \label{fig_hohlraum2da2}
    \end{minipage}
    \begin{minipage}{0.49\textwidth}
        \centering
        \includegraphics[width=1.0\textwidth]{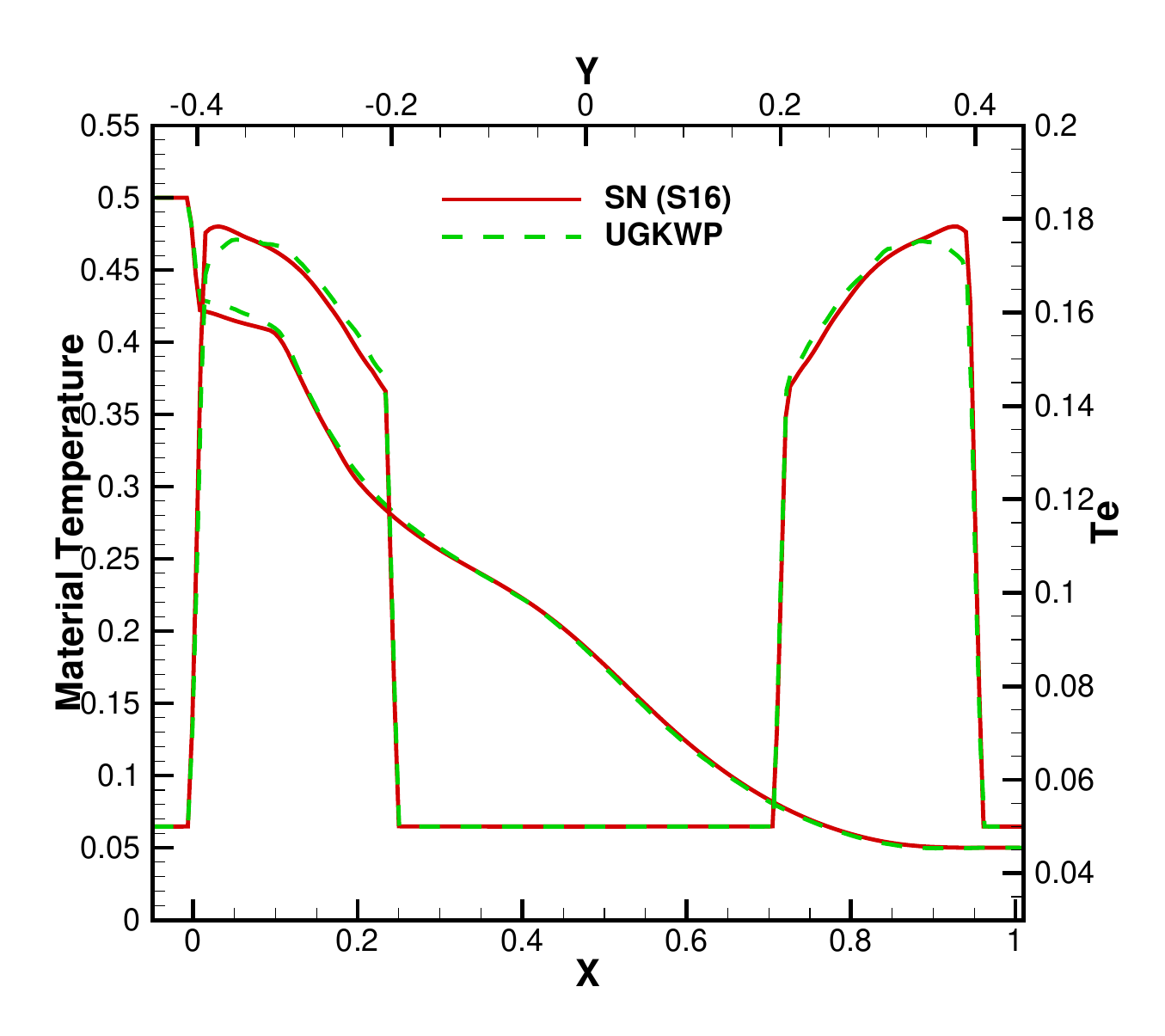}
        \caption{Comparison of $T_e$ IUGKWP and SN along $x=0$ and $y=0$.}
        \label{fig_hohlraum2da3}
    \end{minipage}
\end{figure}

\subsection{Circular hohlraum problem}
The geometry of the circular hohlraum is shown in Fig. \ref{fig_hohlraum2dgeo}.
The hohlraum boundary and capsule are filled with optically thick material with $\sigma=2.0\times10^{3}$ and capacity $C_v=1$.
The hohlraum cavity is filled with optically thin material with $\sigma=2.0\times10^{-1}$ and capacity $C_v=1\times10^{-2}$.
The system is initially in equilibrium at the temperature of $0.05\text{KeV}$.
A $0.5\text{KeV}$ isotropic surface source is applied to the injection hole.
The computational domain is discretized into triangular mesh with cell size $\Delta x=0.01$,
and the CFL number is set to 30.
The material and radiation temperature evolution predicted by IUGKWP agrees well with SN results.
As shown in Fig. \ref{fig_hohlraum2db1},
the IUGKWP temperature contour ($y>0$) and SN temperature contour ($y<0$) agree well at $t=1\text{ns}$.
We also compare the material temperature on the capsule surface in Fig. \ref{fig_hohlraum2db2}
and the temperature profile along $x=0$ and $y=0$ in Fig. \ref{fig_hohlraum2db3}.
The IUGKWP shows robustness and accuracy in the two-dimensional square and circular hohlraum tests.

\begin{figure}
  \centering
  \subfigure[Material temperature]{\includegraphics[width=0.49\textwidth]{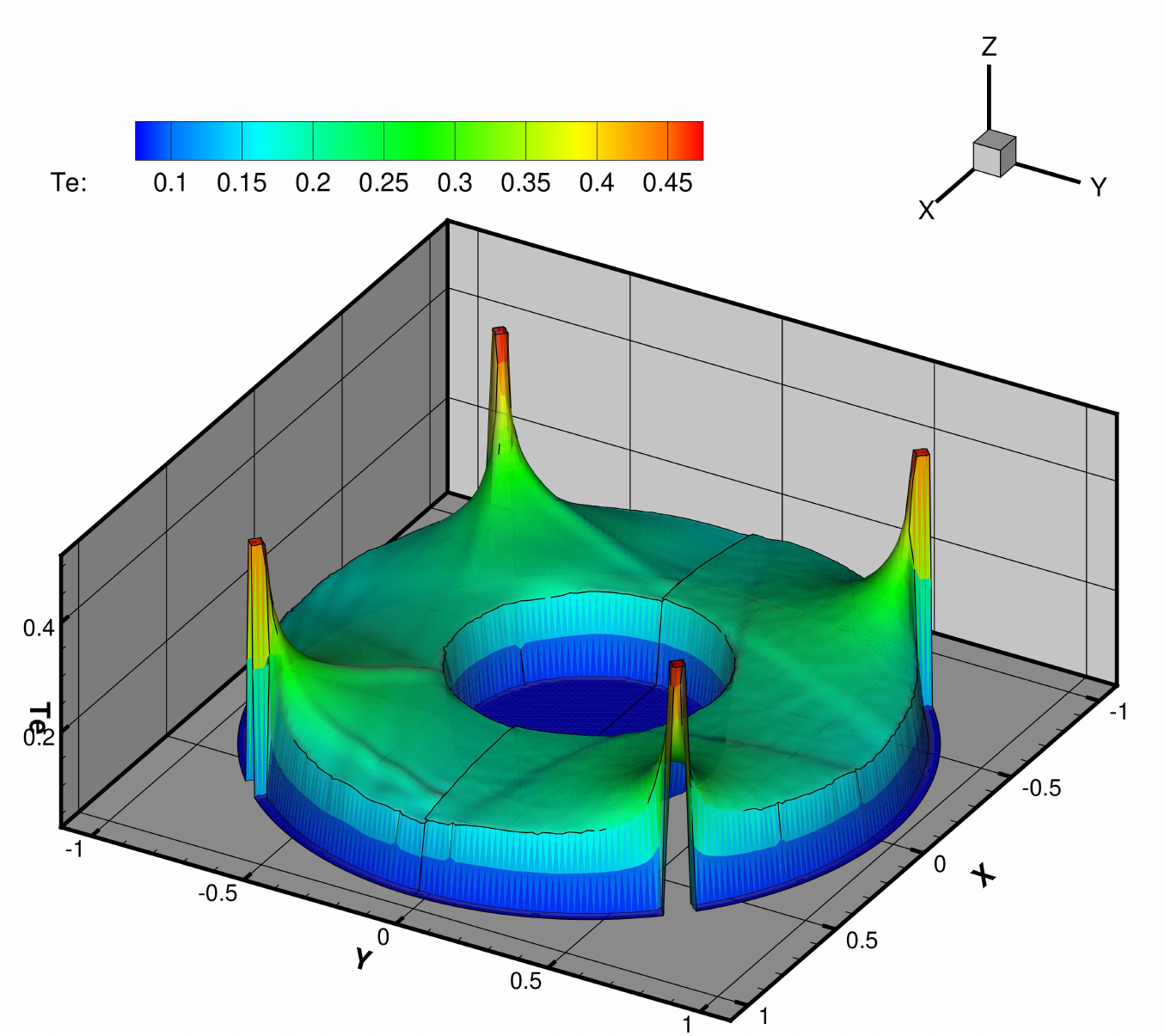}}
  \subfigure[Radiation temperature]{\includegraphics[width=0.49\textwidth]{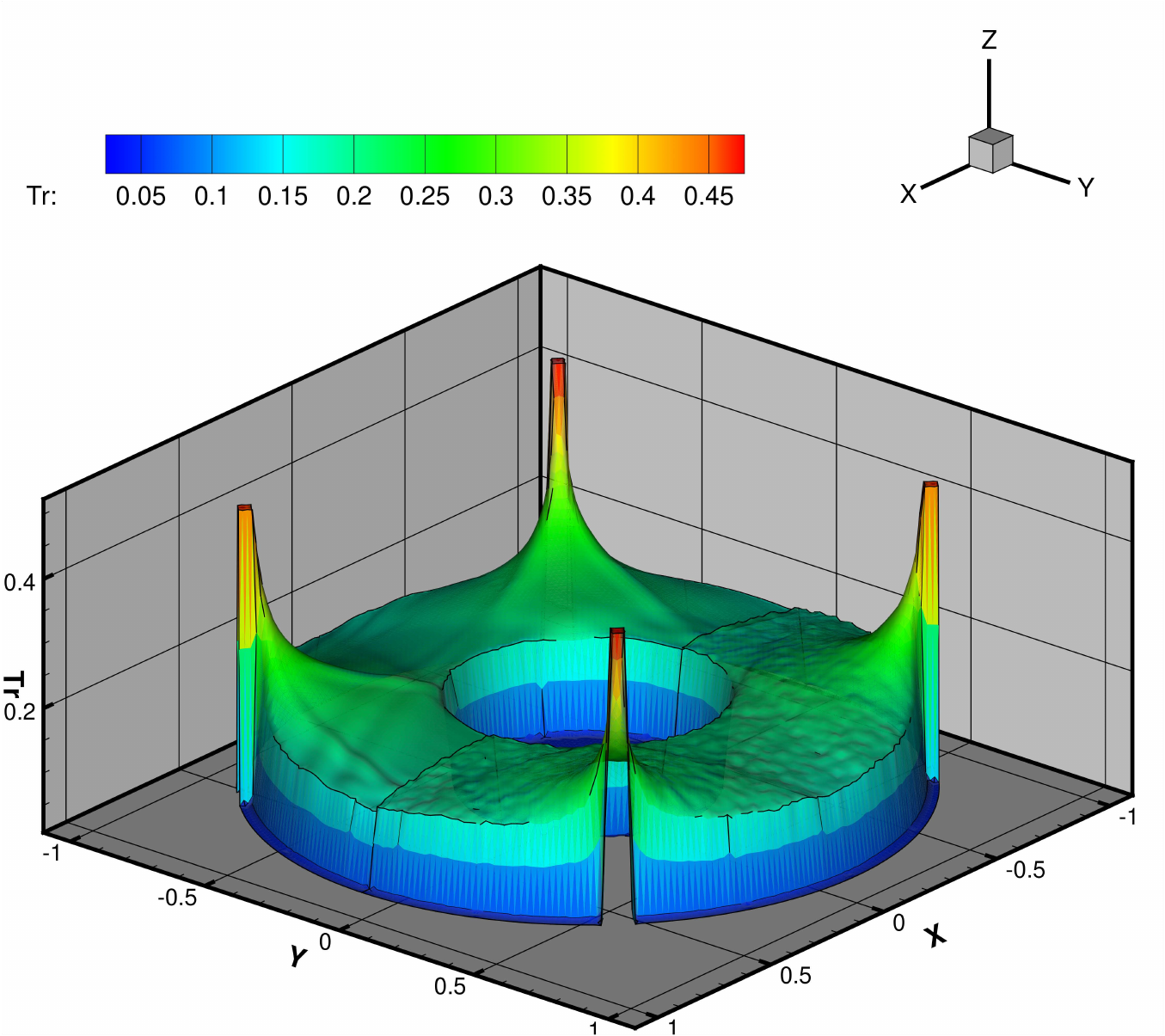}}
  \caption{Comparison of the material and radiation temperature contour between IUGKWP $(y>0)$ and SN $(y<0)$ at $t=1\text{ns}$.}
  \label{fig_hohlraum2db1}
\end{figure}

\begin{figure}
  \centering
    \begin{minipage}{0.49\textwidth}
        \centering 
        \includegraphics[width=1.0\textwidth]{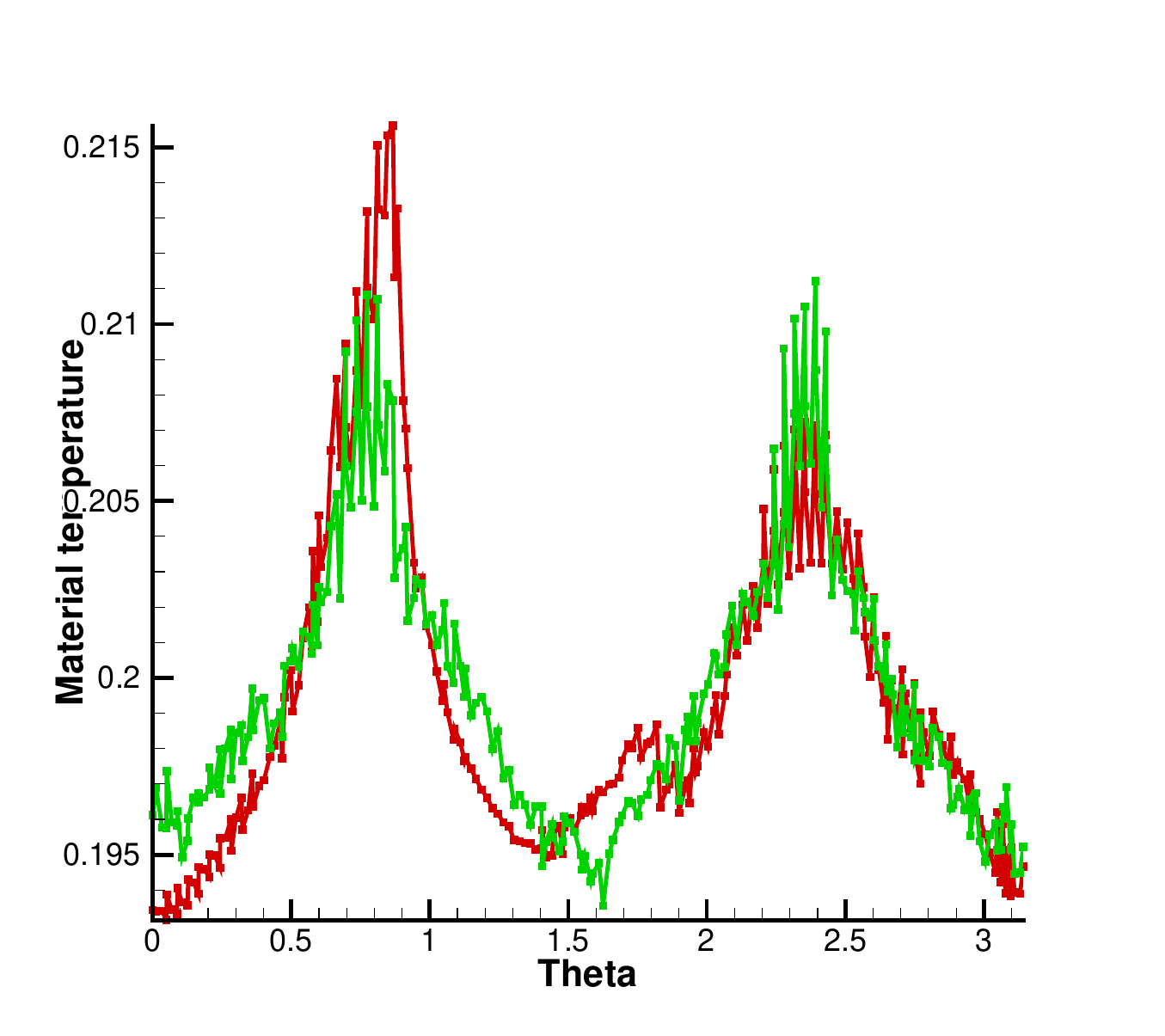}
        \caption{Material temperature on capsule surface.}
        \label{fig_hohlraum2db2}
    \end{minipage}
    \begin{minipage}{0.49\textwidth}
        \centering
        \includegraphics[width=1.0\textwidth]{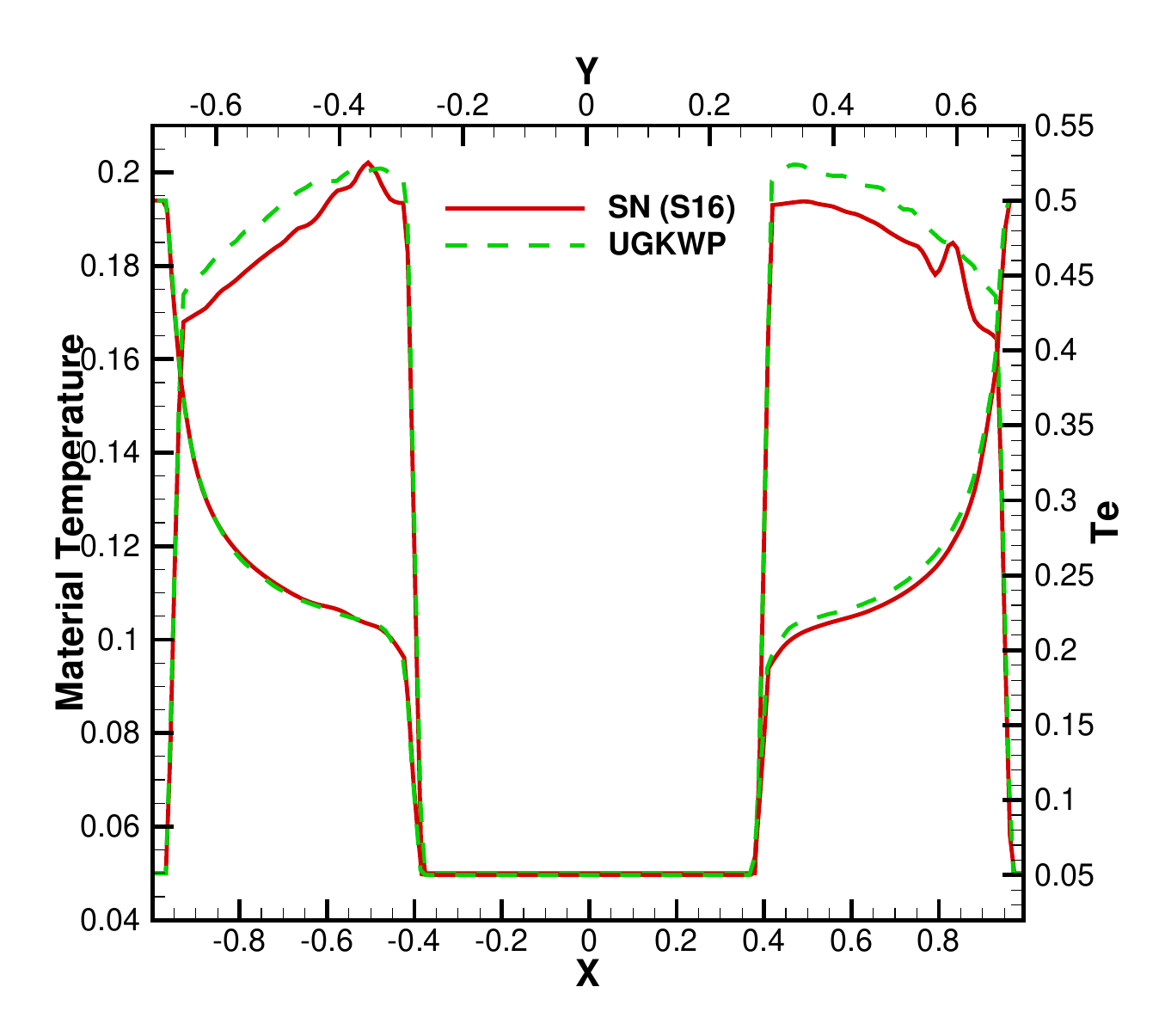}
        \caption{Comparison of $T_e$ IUGKWP and SN along $x=0$ and $y=0$.}
        \label{fig_hohlraum2db3}
    \end{minipage}
\end{figure}

\subsection{3D cylindrical hohlraum problem}
The practical ICF engineering applications require three-dimensional simulations to capture the three-dimensional effects,
such as the 3D radiation-driven asymmetry.
A three-dimensional program has been developed based on the IUGKWP algorithm to study the cylindrical hohlraum energetics in ICF.
A sketch of hohlraum geometry and mesh is shown in Fig. \ref{fig_hohlraum3d1}.
The diameter of hohlraum is $2.3\text{mm}$ in length, and the mesh size is $\Delta x=10\mu m$.
The hohlraum boundary and capsule are filled with optically thick material with $\sigma=2.0\times10^{4}$ and capacity $C_v=1$.
The hohlraum cavity and injection hole are filled with optically thin material with $\sigma=2.0\times10^{-3}$ and capacity $C_v=1\times10^{-2}$.
The system is initially in equilibrium at the temperature of $0.05\text{KeV}$,
and a $0.5\text{KeV}$ isotropic surface source is applied on the hohlraum inner surface.
It takes $mins$ to reach a simulation time of $10\text{ns}$.
The material temperature distribution in hohlraum at $t=1,5,10\text{ns}$ is shown in Fig. \ref{fig_hohlraum3d2},
and material temperature on capsule surface at $t=1,5,10\text{ns}$ is shown in Fig. \ref{fig_hohlraum3d3}.

\begin{figure}
  \centering
  \subfigure[Surface mesh]{\includegraphics[width=0.45\textwidth]{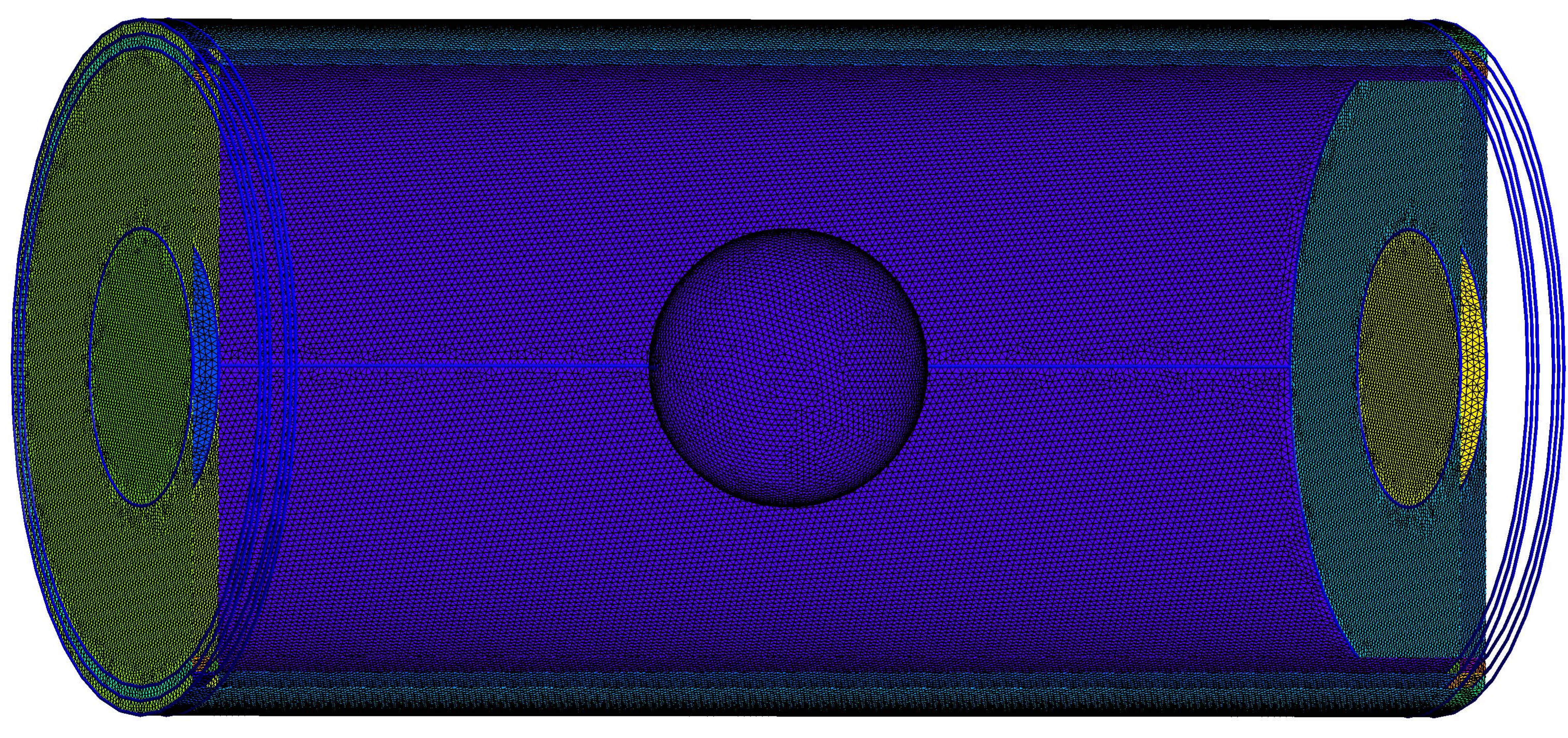}}
  \hspace{5mm}
  \subfigure[Volume mesh]{\includegraphics[width=0.45\textwidth]{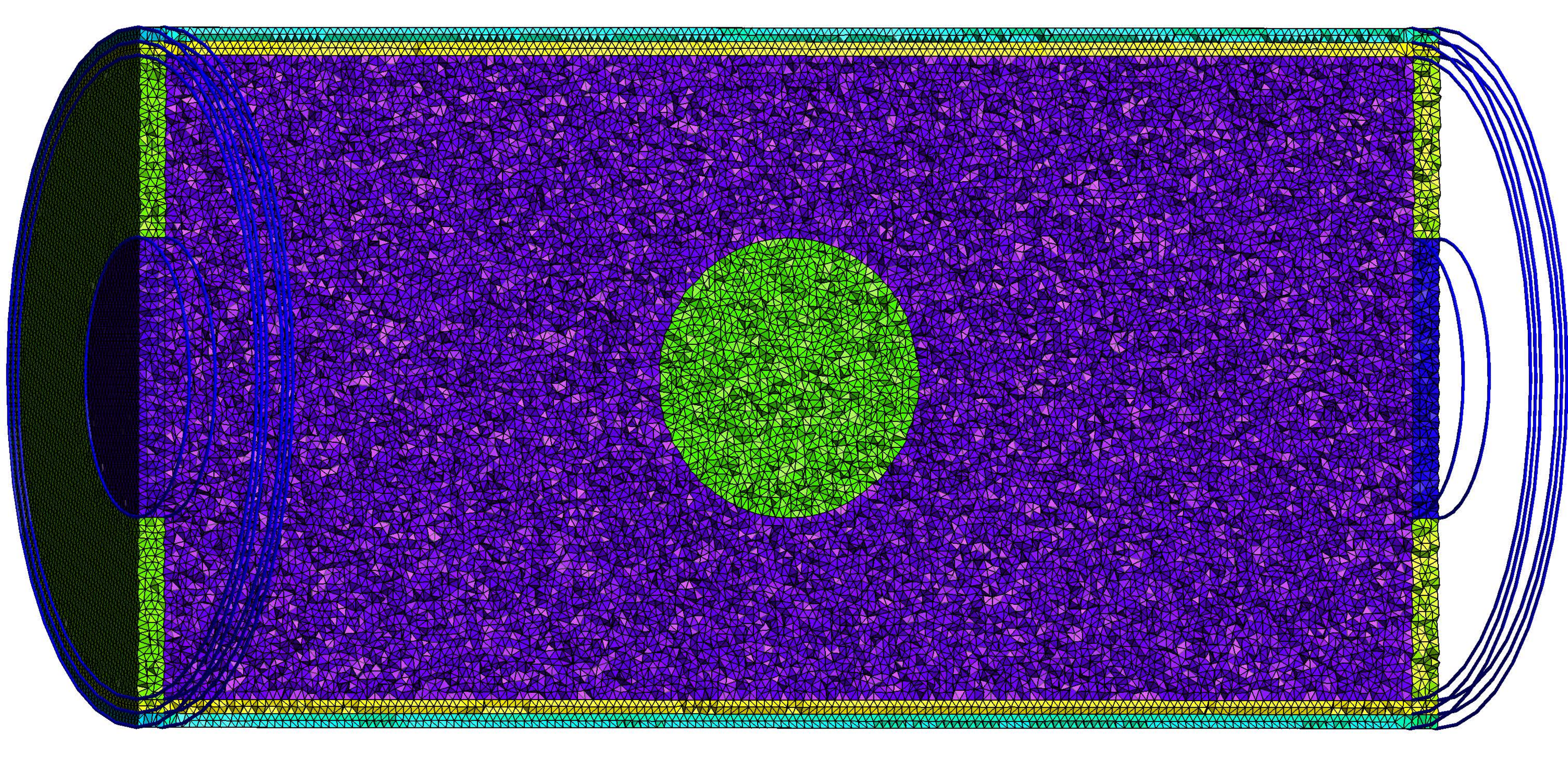}}
  \caption{Mesh distribution of the 3D hohlraum problem.}
  \label{fig_hohlraum3d1}
\end{figure}

\begin{figure}
  \centering
  \subfigure[$t=1\text{ns}$]{\includegraphics[width=0.325\textwidth]{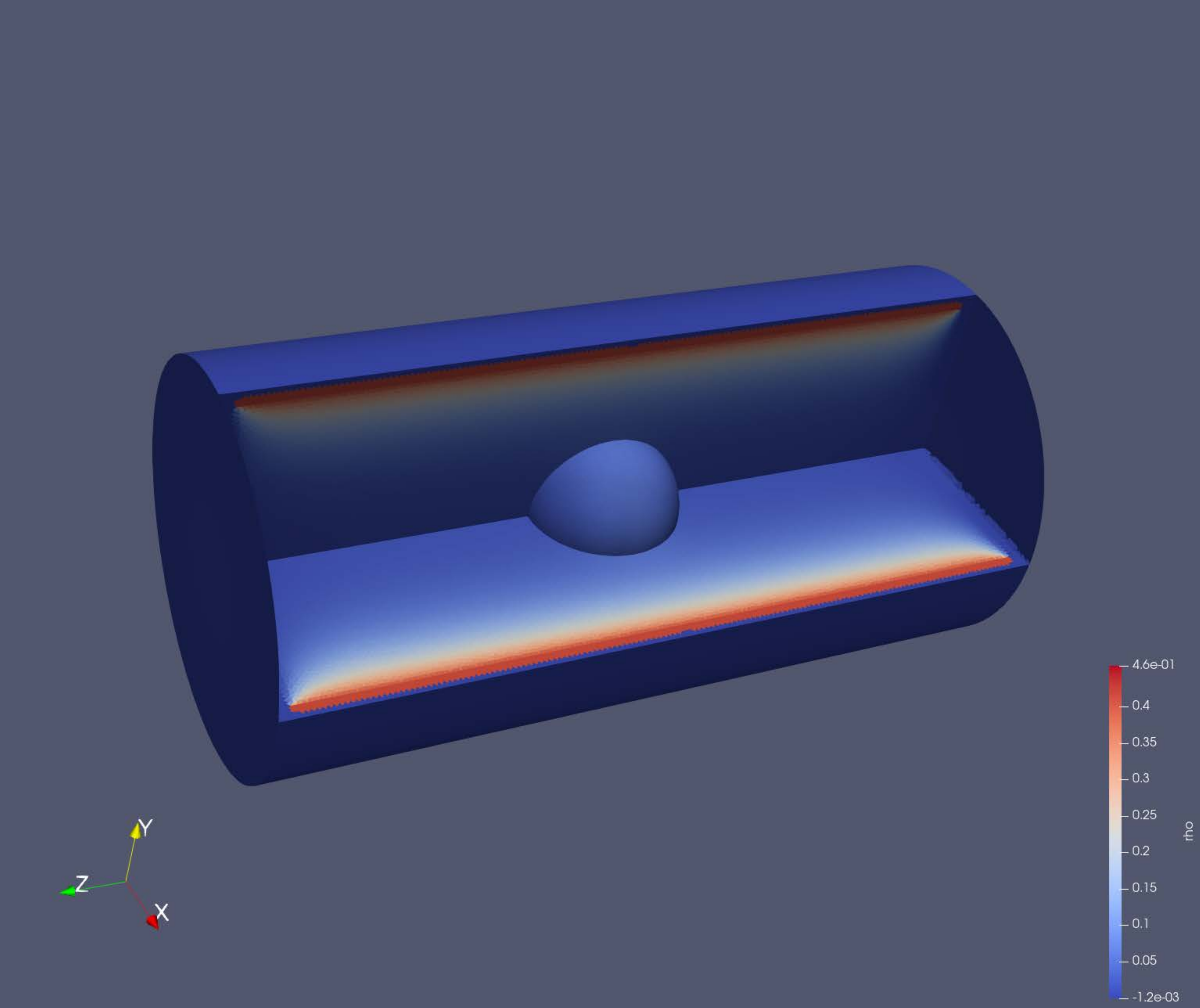}}
  \subfigure[$t=5\text{ns}$]{\includegraphics[width=0.325\textwidth]{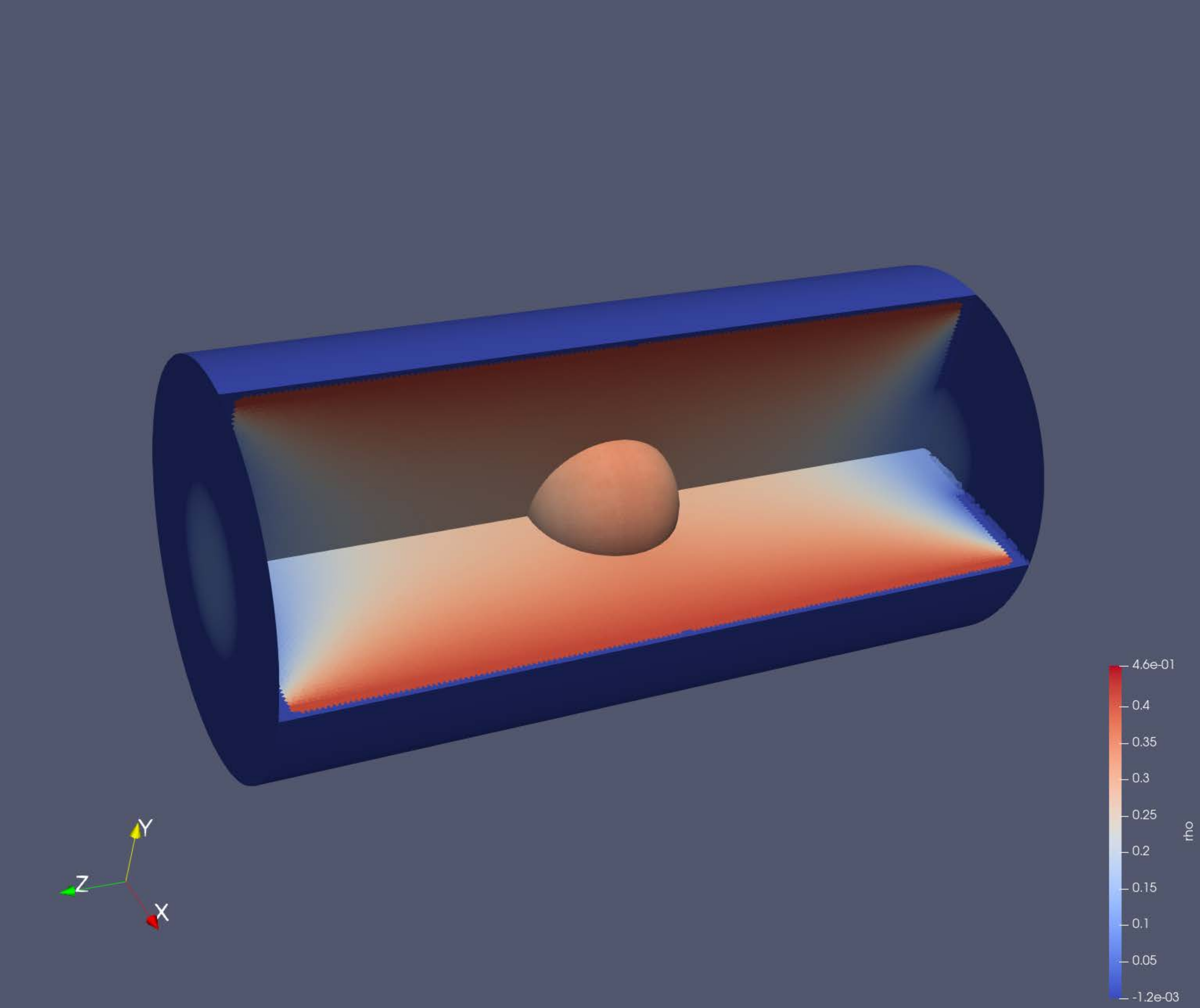}}
  \subfigure[$t=10\text{ns}$]{\includegraphics[width=0.325\textwidth]{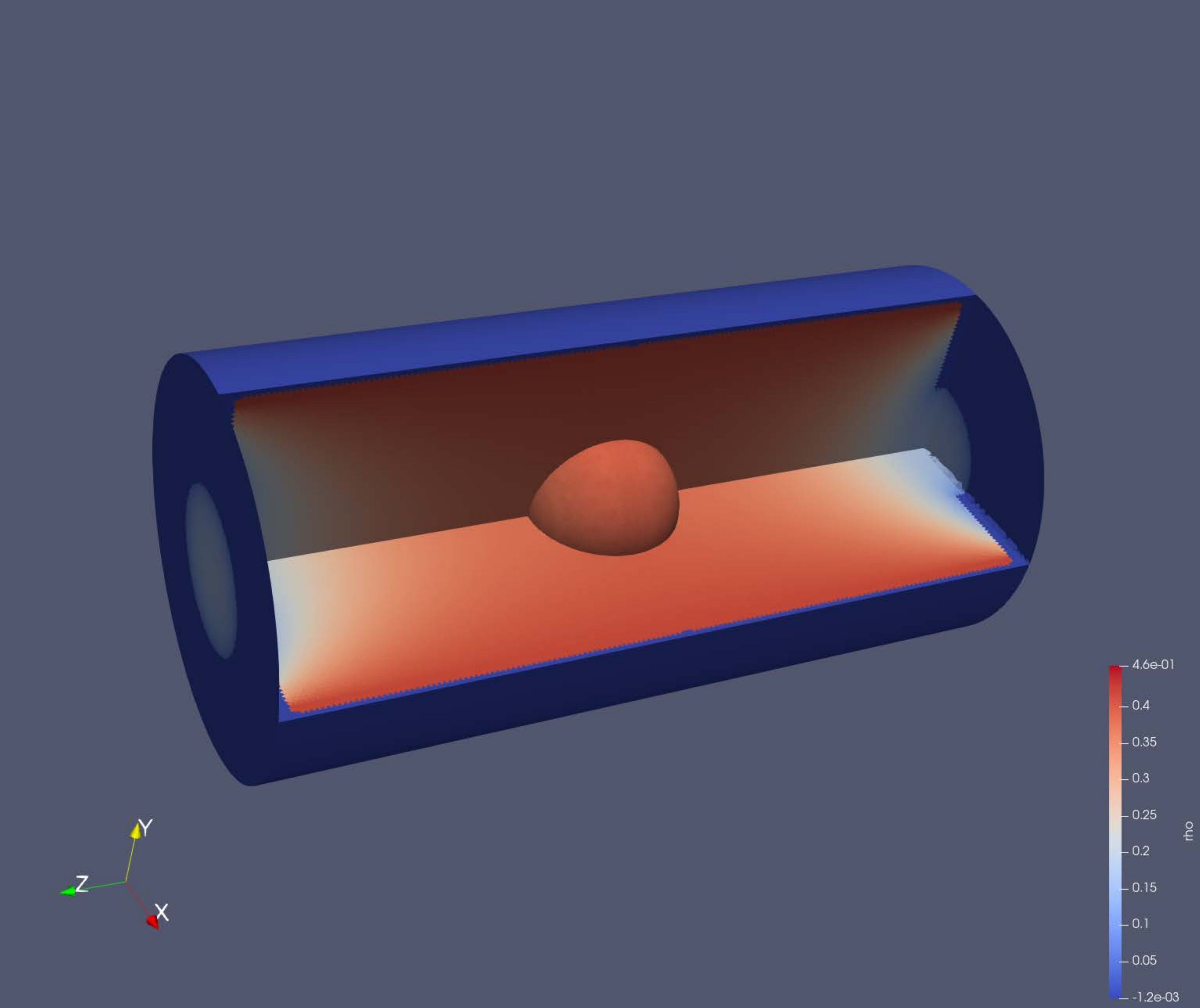}}
  \caption{Material temperature distribution in hohlraum at $t=1,5,10\text{ns}$.}
  \label{fig_hohlraum3d2}
\end{figure}

\begin{figure}
  \centering
  \subfigure[$t=1\text{ns}$]{\includegraphics[width=0.325\textwidth]{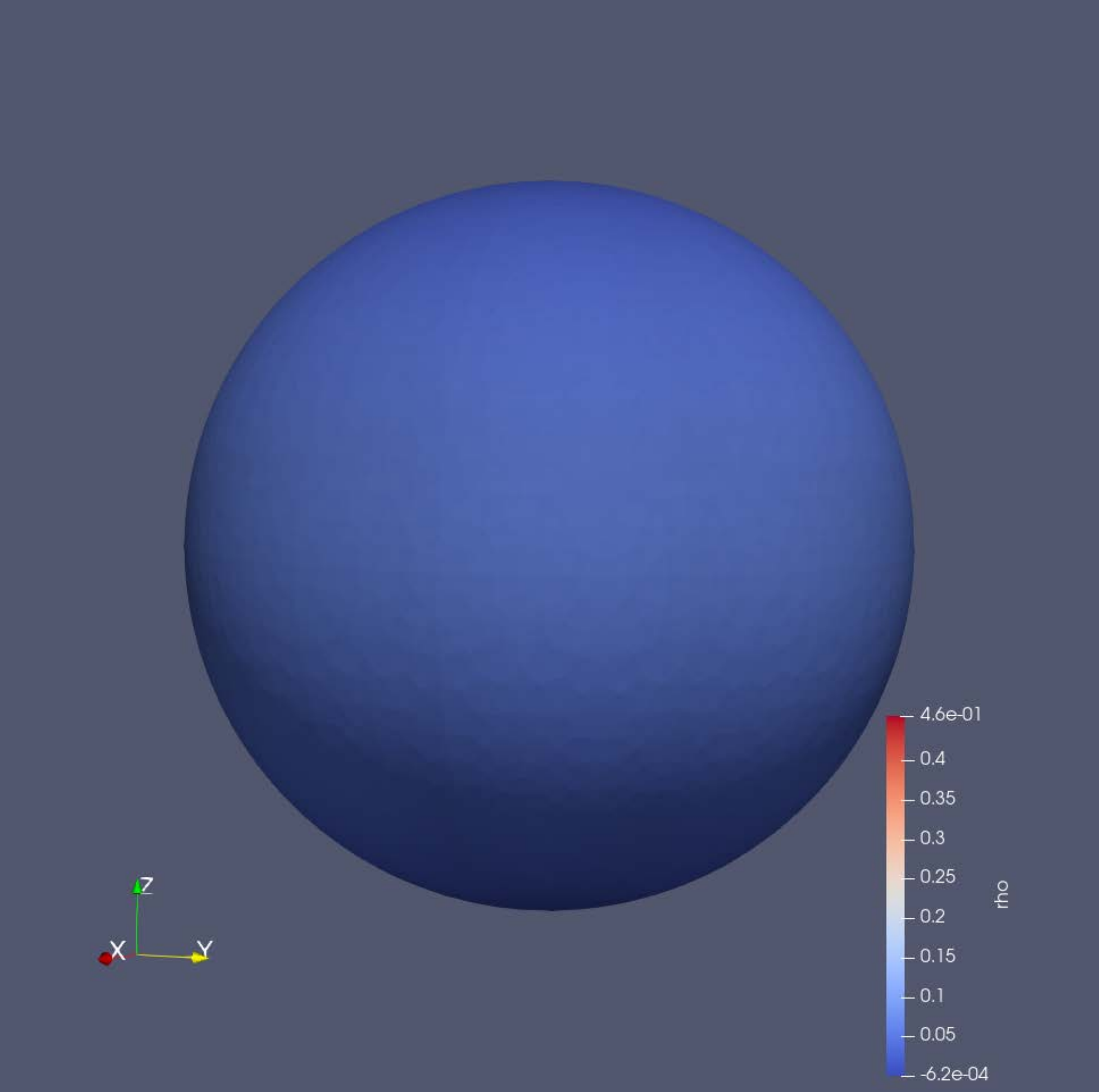}}
  \subfigure[$t=5\text{ns}$]{\includegraphics[width=0.325\textwidth]{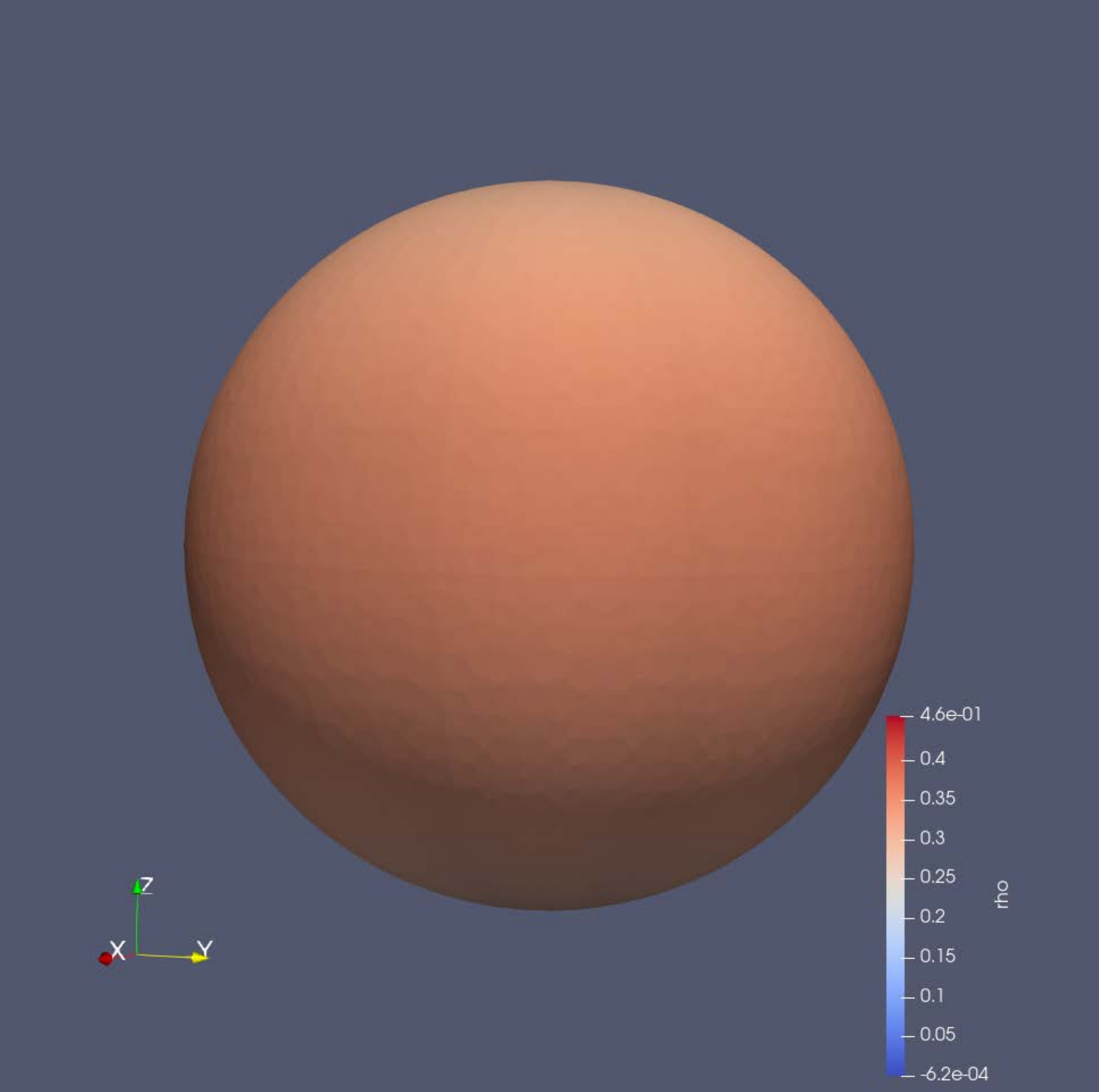}}
  \subfigure[$t=10\text{ns}$]{\includegraphics[width=0.325\textwidth]{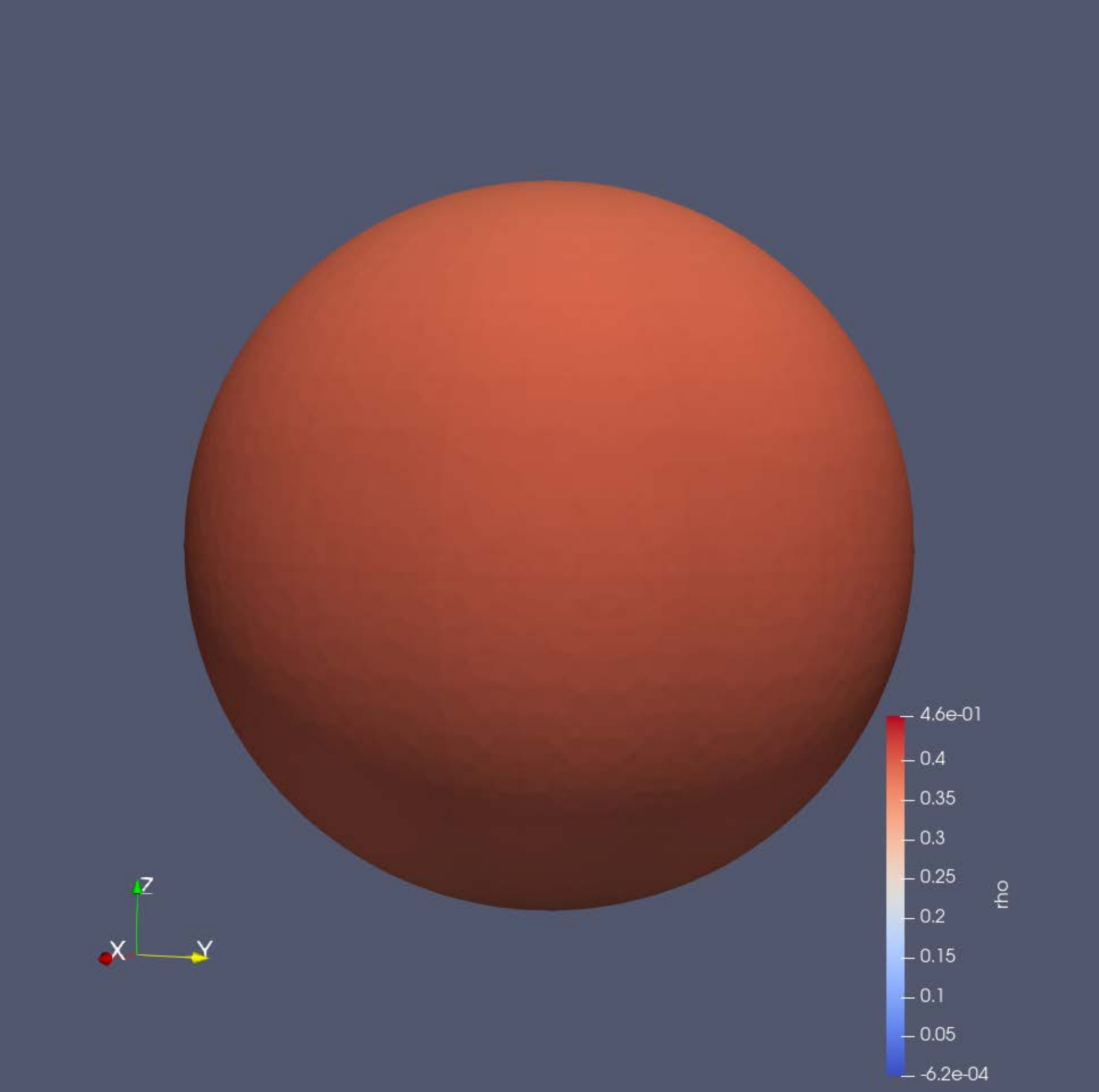}}
  \caption{Material temperature on capsule surface at $t=1,5,10\text{ns}$.}
  \label{fig_hohlraum3d3}
\end{figure}

\section{Conclusion}\label{section_conclusion}
The implicit IUGKWP method is developed in this paper.
In the scheme construction, we propose a physical time step $t_p$, 
which determines the flow physics modeling 
and categorize the particles into long-$\lambda$ and short-$\lambda$ particles.
The algorithm of IUGKWP has three steps. First, the long-$\lambda$ particles are tracked by the IMC method;
Second, the short-$\lambda$ transport process is evolved by an implicit diffusion system;
Third, the photon distribution is closed based on the integral solution.
The scheme has the numerical properties of asymptotic-preserving and regime-adaptive.
Multidimensional codes for the IUGKWP method are developed 
, and the accuracy and efficiency are verified by 2D and 3D tests.
The implicit UGKWP method and codes will be applied and tested 
in the engineering applications of inertial confinement fusion.

\section*{Acknowledgement}
The authors are partially supported by the National Key R\&D Program of China (2022YFA1004500).
Chang Liu is partially supported by the National Natural Science Foundation of China (12102061).
Weiming Li is partially supported by the National Natural Science Foundation of China (12001051).
Yanli Wang is partially supported by the National Natural Science Foundation of China (12171026, U1930402, and 12031013) and 
the Foundation of President of China Academy of Engineering Physics (YZJJZQ2022017).
Peng Song is partially supported by the National Natural Science Foundation of China (12031001).

\bibliographystyle{unsrt}
\bibliography{generalmesh}
\end{document}